\documentclass[sigconf]{aamas_hacked} 

	\usepackage{fix-cm}
	\usepackage[T1]{fontenc}
	\usepackage{xspace}
	\usepackage{amsmath}
	\usepackage{amsthm}

	\usepackage{nicefrac}
	\usepackage{tikz}
	\usepackage{pgflibraryshapes}
	\usetikzlibrary{positioning,arrows,shapes,snakes,calc}
	\usepackage{colortbl}
	\usepackage{booktabs}
	\usepackage{cleveref}
	\usepackage{array}
	\usepackage{graphicx}
	\usepackage{xcolor}
	\usepackage{enumitem}
	\usepackage{tabularx}
	\usepackage{thm-restate}
	\usepackage{ragged2e}

	\usepackage{multirow}
	\newcounter{remark}
	\newenvironment{remark}[1][]{\refstepcounter{remark} \ifthenelse{\equal{#1}{}}{\noindent\textbf{Remark~\theremark. }}{\noindent \textbf{Remark~\theremark~(#1).}}}{\medskip}

    \newcommand{\Tau}{\mathrm{T}}

	\theoremstyle{definition}
	\newtheorem{definition}{Definition}
	
	\theoremstyle{plain}
	\newtheorem{theorem}{Theorem}
	\newtheorem{lemma}{Lemma}
	
	\newtheorem{proposition}{Proposition}
	\newtheorem{corollary}{Corollary}

	\title{Characterizations of Sequential Valuation Rules}

	\author{Chris Dong}
	\affiliation{
	  \institution{Technical University of Munich}
	  \city{Munich}
	  \country{Germany}}
	\email{chris.dong@tum.de}

	\author{Patrick Lederer}
	\affiliation{
	  \institution{Technical University of Munich}
	  \city{Munich}
	  \country{Germany}}
	\email{ledererp@in.tum.de}

	\begin{abstract}
		Approval-based committee (ABC) voting rules elect a fixed size subset of the candidates, a so-called committee, based on the voters' approval ballots over the candidates. While these rules have recently attracted significant attention, axiomatic characterizations are largely missing so far. We address this problem by characterizing ABC voting rules within the broad and intuitive class of sequential valuation rules. These rules compute the winning committees by sequentially adding candidates that increase the score of the chosen committee the most. In more detail, we first characterize almost the full class of sequential valuation rules based on mild standard conditions and a new axiom called consistent committee monotonicity. This axiom postulates that the winning committees of size $k$ can be derived from those of size $k-1$ by only adding candidates and that these new candidates are chosen consistently. By requiring additional conditions, we derive from this result also a characterization of the prominent class of sequential Thiele rules. Finally, we refine our results to characterize three well-known ABC voting rules, namely sequential approval voting, sequential proportional approval voting, and sequential Chamberlin-Courant approval voting.
	\end{abstract}

	\newcommand{\BibTeX}{\rm B\kern-.05em{\sc i\kern-.025em b}\kern-.08em\TeX}

\begin{document}

	\pagestyle{fancy}
	\fancyhead{}

	\maketitle
	
	\section{Introduction}
 
	Whether it is choosing dishes for a shared lunch, shortlisting candidates for interviews, or electing a parliament of a country---all these problems require us to elect a fixed size subset of the available candidates based on the voters' preferences. This problem, commonly studied under the term \emph{approval-based committee (ABC) voting}, has recently attracted significant attention within the field of social choice theory because of its versatile applications \cite{EFSS17a,FSST17a,LaSk22b}. In more detail, the study objects for this problem are ABC voting rules which choose a subset of the candidates of predefined size, a so-called committee, based on the voters' approval ballots, i.e., each voter reports the set of candidates she finds acceptable. 
	
	Due to the large amount of work on ABC voting, there is a wide variety of ABC voting rules, e.g., Thiele methods, sequential Thiele methods, Phragmen's rules, the method of equal shares, and many more (we refer to \cite{LaSk22b} for an overview of these rules). For deciding which rule to use in a given situation, social choice theorists commonly reason about their properties: if a voting rule satisfies desirable properties, it seems to be a good choice for the election at hand. However, such reasoning does not rule out the existence of an even more attractive voting rule satisfying  the required properties. For narrowing down the choice to a single ABC voting rule, a characterization of this rule is required, i.e., one needs to show that the rule is the unique method that satisfies a set of properties. Unfortunately, such characterizations are largely missing in the literature on ABC voting rules and it is therefore an important open problem to derive such results (see, e.g., \cite[Q1]{LaSk22b}).
	
	The goal of this paper thus is to provide such characterizations for ABC voting rules within the new but broad and intuitive class of sequential valuation rules. For computing the winning committees, these rules rely on a valuation function which assigns a score to each pair of ballot and committee. A simple example of such a function is $v(A_i, W)=|A_i\cap W|$, where $A_i$ is an arbitrary ballot and $W$ is a committee. Based on a valuation function, a sequential valuation rule proceeds in rounds and, in each round, it extends the previously chosen committees with the candidates that increase the total score by the most. Clearly, the prominent class of sequential Thiele rules, which only rely on the size of the intersection of the given ballot and committee to compute the score, forms a subset of the class of sequential valuation rules. However, our class is much more general as it contains, for instance, step-dependent sequential scoring rules, whose valuation functions depend on the sizes of the ballot, the committee, and the intersection of these two.
		
	\paragraph{Our Contribution.} As our main contribution, we characterize the class of sequential valuation rules that satisfy mild standard conditions based on a new axiom called consistent committee monotonicity. This property combines the well-known notions of committee monotonicity \citep[e.g.,][]{BaCo08a,KiMa12a,EFSS17a} and consistency \citep[e.g.,][]{Youn75a,Fish78d,LaSk21a}. Roughly, committee monotonicity requires that the winning committees of size $k$ can be derived from those of size $k-1$ by simply adding candidates. On the other hand, the idea of consistency is that whenever two disjoint electorates separately elect the same candidates, these candidates should be the winners when we consider both electorates simultaneously. Consistent committee monotonicity combines these two axioms by requiring that the candidates that extend the committees of size $k$ are chosen consistently: if some common candidates extend a committee $W$ in two disjoint elections, these candidates should also extend $W$ in the combined election. 
	Or, to put it simpler, consistent committee monotonicity restricts committee monotonicity by requiring that the newly added candidates are chosen in a reasonable way. 
	
	Based on this axiom, we characterize the class of sequential valuation rules that satisfy anonymity, neutrality, non-imposition, and continuity (\Cref{thm:characterization}). These four conditions are mild standard axioms that are satisfied by almost
	all ABC rules considered in the literature and we henceforth summarize them by the term \emph{proper}. In more detail, we first show that every proper sequential valuation rule is a step-dependent sequential scoring rule, i.e., its valuation function only depends on the sizes of the ballot, the committee, and the intersection of these two. As second step, we then characterize step-dependent sequential scoring rules as the only proper and consistently committee monotone ABC voting rules. Or, put differently, when the winning committees should be computed sequentially and the newly added candidates are chosen in a consistent way, we naturally arrive at the class of step-dependent sequential scoring rules, thus giving a strong argument for using these rules.
	
	Based on our characterization of step-dependent sequential valuation rules, we also infer characterizations of more restricted classes of voting rules by requiring additional axioms. In particular, we present such results for step-dependent sequential Thiele rules (whose valuation functions only depend on the size of the committee and the size of the intersection of the ballot and the committee) and sequential Thiele rules (whose valuation functions only depend on the size of the intersection of the ballot and the committee). Hence, we derive a hierarchy of characterizations based on our first theorem and, in particular, provide a full characterization of the prominent class of sequential Thiele rules. Finally, we leverage these results to characterize three commonly studied ABC voting rules, namely sequential approval voting, sequential proportional approval voting, and sequential Chamberlin-Courant approval voting, by investigating how they treat clones. An overview of our results can also be found in \Cref{fig:overview}.
	
	\begin{figure}[t]
		\centering
	      \begin{tikzpicture}
		  \tikzstyle{arrow}=[->,>=angle 60, shorten >=1pt,draw]
		  \tikzstyle{mynode}=[align=center, anchor=north]

		  		\node[mynode] (SVR) at (0,0)   {Sequential valuation rules};
		  		\node[mynode] (SSSR) at (0, -1.5){Step-dependent sequential scoring rules\\$=$\emph{Proper and consistently committee monotone ABC voting rules}};
				\node[mynode] (SSTR) at (0, -3.5) {Step-dependent sequential Thiele rules};
				\node[mynode] (STR)  at (0, -5){Sequential Thiele rules};
				
				\node[mynode] (SAV)  at (-3.7, -6.5){\texttt{seqAV}};
				\node[mynode] (SPAV)  at (3.7, -6.5){\texttt{seqPAV}};
				\node[mynode] (SCCAV)  at (0, -6.5){\texttt{seqCCAV}};
		
		  		\draw[-latex] (SVR) to node[midway, right]{\emph{Properness}} (SSSR);
		  		\draw[-latex] (SSSR) to node[midway, right]{\emph{Independence of Losers}} (SSTR);
				\draw[-latex] (SSTR) to node[midway, right]{\emph{Committee Separability}} (STR);
				\draw[-latex] (STR) to node[midway, left,xshift=0.4cm,yshift=0.3cm,align=center]{\emph{Clone-acceptance}\\+\emph{Distrust}} (SAV);
				\draw[-latex] (STR) to node[midway, right, yshift=0.3cm,xshift=-0.4cm,align=center]{\emph{Clone-propor-}\\\emph{tionality}} (SPAV);
				\draw[-latex] (STR) to node[midway, left,align=center]{\emph{Clone-}\\\emph{rejection}} (SCCAV);
	      \end{tikzpicture}
		  \caption{Overview of our results. An arrow from class $X$ to class $Y$ with label $Z$ means an ABC voting rule in the class $X$ is in the class $Y$ if and only if it satisfies property $Z$.}
		  \label{fig:overview}
	\end{figure}
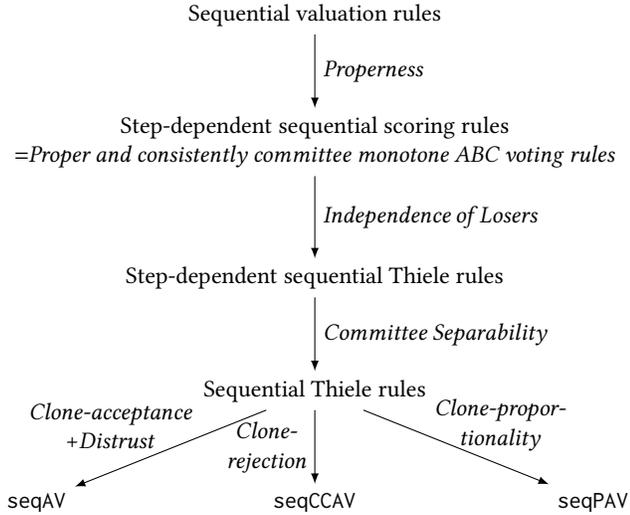

	\paragraph{Related Work.} The study of committee monotone ABC voting rules has a long tradition as already Thiele \cite{Thie95a} suggested the class of functions nowadays known as sequential Thiele rules. In particular, for a number of applications such as choosing finalists for a competition or shortlisting candidates for an interview, it is frequently reasoned that committee monotonicity is a desirable property \cite{BaCo08a,KiMa12a,EFSS17a}. More generally, \citeauthor{FSST17a} \cite{FSST17a} view committee monotonicity as the fundamental property when choosing candidates only based on their quality because in such settings, there is no reason why a candidate that is elected for a committee of size $k$ should not be elected for a committee of size $k+1$. 
 
    Another important advantage of such sequential ABC rules is that they are easy to compute, whereas rules that directly optimize the score (e.g., Thiele rules) are usually NP-hard to compute \cite{SFL16a}. Indeed, sequential ABC voting rules have even been considered as approximation algorithms for these optimizing rules \cite{LuBo11d,SFL16a}. On the other hand, committee monotonicity conflicts with other desirable properties. For instance, \citeauthor{BaCo08a} \cite{BaCo08a} show that this axiom is incompatible with a variant of Condorcet-consistency when voters report strict rankings over the candidates, and it has been repeatedly observed that committee monotone ABC voting rules are less proportional than other rules \cite{EFSS17a, LaSk22b, SLB+17a}.
	
	Even more work has focused on specific committee monotone ABC voting rules \citep[e.g.,][]{ABC+16a,BLS18a,LaSk20a,DDE+22a}. For instance, \citeauthor{DDE+22a} \cite{DDE+22a} show that all sequential Thiele rules but sequential approval voting fail strategyproofness, and \citeauthor{BLS18a} \cite{BLS18a} investigate these rules with respect to proportionality axioms. An interesting observation in this context is that Phragmen's sequential rule is committee monotone and satisfies strong proportionality conditions \cite{BFJL16a,PeSk20a}; unfortunately, this rule fails our consistency criterion. 
 
    From a conceptual standpoint our results are also related to theorems for different settings as consistency led to a number of important characterizations. In particular, based on this axiom, \citet{Youn75a} characterizes scoring rules for single winner elections, \citet{Fish78d} characterizes approval voting for single winner elections with dichotomous preferences, \citet{YoLe78a} characterize a method called Kemeny's rule in a model where the outcome is set of rankings over the candidates, and \citet{Bran13a} characterize a voting rule called maximal lotteries in a randomized setting. More recently, \citet{LaSk21a} characterized ABC scoring rules based on a consistency condition for committees instead of single candidates in a model where the output is a ranking over committees. To the best of our knowledge, this result is the only complete characterization in the realm of ABC voting.
		
	\section{The Model}
	Let $\mathbb{N}=\{1,2,3,\dots\}$ denote an infinite set of voters and let $\mathcal{C}=\{a_1, \dots, a_m\}$ denote a fixed set of $m$ candidates. We define $\mathcal{F}(\mathbb{N})$ as the set of finite and non-empty subsets of $\mathbb{N}$. Intuitively, an element $N\in\mathcal{F}(\mathbb{N})$ represents a concrete electorate, whereas $\mathbb{N}$ is the set of all possible voters. Given an electorate $N\in\mathcal{F}(\mathbb{N})$, we assume that every voter $i\in N$ has dichotomous preferences over the candidates, i.e., she partitions the candidates into approved and disapproved ones. Thus, voters report \emph{approval ballots $A_i$} which are non-empty subsets of $\mathcal{C}$. Let $\mathcal{A}$ denote the set of all possible approval ballots. 
	An \emph{approval profile $A$} for an electorate $N$ is an element of $\mathcal A^N$, i.e., a function that maps every voter $i\in N$ to her approval ballot $A_i$. We define $\mathcal{A}^*=\bigcup_{N\in\mathcal{F}(\mathbb{N})} \mathcal A^N$ as the set of all possible approval profiles. Given a profile $A\in\mathcal{A}^*$, we let $N_A$ indicate the set of voters who report a ballot in the profile $A$ and we say that two profiles $A, A'$ are disjoint if $N_A\cap N_{A'}=\emptyset$. Moreover, for two disjoint profiles $A$ and $A'$, we define $A+A'$ as the profile with $N_{A+A'}=N_A\cup N_{A'}$, $(A+A')_i=A_i$ for all $i\in N_A$, $(A+A')_i=A_i'$ for all $i\in N_{A'}$.
	
	Given an approval profile, the goal is to choose a committee. Formally, a \emph{committee} is a subset of the candidates with a specific size. We denote by $\mathcal{W}_k$ the set of all committees of size $k$ and by $\mathcal{W}=\bigcup_{k=0}^m \mathcal{W}_k$ the set of all committees. For selecting the winning committees for an approval profile $A$, we use \emph{approval-based committee (ABC) voting rules}. 
	These rules are functions which take an arbitrary approval profile $A\in \mathcal{A}^*$ and target committee size $k\in \{0,\dots, m\}$ as input and return a non-empty subset of $\mathcal{W}_k$. 
	Intuitively, the chosen set contains the winning committees and we allow for sets of committees as output to indicate that multiple committees are tied for the win. Furthermore, note that ABC voting rules are also defined for committees of size $0$: $f(A,0)=\{\emptyset\}$ for all profiles $A$ since the empty set is the only committee of size $0$. This definition is only used for notational convenience.
	
	In this paper, we will restrict our attention to \emph{proper} ABC voting rules which satisfy the following four conditions. Note that almost all commonly studied ABC voting rules are proper voting rules as the subsequent axioms are extremely mild.\footnote{Indeed, we are only aware of a single studied voting rule that fails to be proper: the minimax rule \citep{BKS07a}, which chooses the committees that minimize the maximal Hamming-distance to a ballot. This rule fails continuity as it completely ignores how many voters report a specific ballot. We view this rule as unreasonable in light of our axioms.}
	\begin{itemize}
		\item Anonymity: An ABC voting rule $f$ is \emph{anonymous} if $f(A,k)=f(\pi(A), k)$ for all $A\in\mathcal{A}^*$, $k\in \{0,\dots, m\}$, and bijections $\pi:\mathbb{N}\rightarrow\mathbb{N}$. Here, $A'=\pi(A)$ denotes the profile such that $N_{A'}=\pi(N_A)$ and $A'_{\pi(i)}=A_{i}$ for all $i\in N_{A}$.
		\item Neutrality: An ABC voting rule $f$ is \emph{neutral} if $f(\tau(A),k)=\{\tau(W)\colon W\in f(A,k)\}$ for all $A\in\mathcal{A}^*$, $k\in \{0,\dots, m\}$, and  bijections $\tau:\mathcal{C}\rightarrow\mathcal{C}$. $A'=\tau(A)$ denotes here the profile such that $N_{A'}=N_A$ and $A'_i=\tau(A_i)$ for all $i\in N_A$.
		\item Continuity: An ABC voting rule $f$ is \emph{continuous} if for all disjoint profiles $A,A'\in\mathcal{A}^*$ and committee sizes $k\in \{0,\dots, m\}$ such that $|f(A,k)|=1$, there is an integer $j\in \mathbb N$ such that $f(jA+A',k)=f(A,k)$. Here, $jA$ denotes a profile consisting of $j$ disjoint copies of $A$; the identities of the voters are irrelevant for proper rules due to anonymity. 
		\item Non-imposition: An ABC voting rule $f$ is \emph{non-imposing} if for every committee $W\in\mathcal{W}$, there is a profile $A\in\mathcal{A}^*$ such that $f(A,|W|)=\{W\}$.
	\end{itemize}
	
	Anonymity and neutrality are common fairness conditions which require that voters and candidates, respectively, are treated equally. Continuity, also known as overwhelming majority axiom \citep{Myer95b}, requires that a sufficiently large group can force the voting rule to choose their desired committee. Finally, non-imposition states that each committee has a chance to be uniquely chosen. 

	Aside of these standard conditions, we will use two new axioms in our analysis: independence of losers and committee separability. The idea of independence of losers is that a chosen committee $W\in f(A,k)$ should still be chosen if some voters change their preferences by disapproving candidates $c\not\in W$ because, intuitively, this does not affect the quality of $W$. Formally, we say an ABC voting rule $f$ is \emph{independent of losers} if $W\in f(A,|W|)$ implies that $W\in f(A',|W|)$ for all profiles $A, A'\in\mathcal{A}^*$ and committees $W\in\mathcal{W}_k$ with $N_A=N_{A'}$, $W\cap A_i=W\cap A_i'$, and $A_i'\subseteq A_i$ for all $i\in N_A$. Note that this axiom is well-known in single winner voting and choice theory \citep[e.g.,][]{BrHa11a,BrPe19a}. While this axiom has not been considered for ABC elections before, we find it intuitive and it is satisfied by all commonly considered ABC voting rules which do not depend on the ballot size (e.g., Thiele rules, sequential Thiele rules, Phragmen's rule). On the other hand, satisfaction approval voting fails independence of losers as it depends on the sizes of the voters' approval ballots (see \cite{LaSk22b} for definitions of these rules).
	
	Our second non-standard axiom is committee separability. The rough intuition of this axiom is that if there are two disjoint profiles $A$ and $B$ such that no voters $i\in N_A$, $j\in N_B$ approve a common candidate, we can decompose every chosen committee $W$ into two subcommittees which are chosen for $A$ and $B$ separately. For formally defining this axiom, let $C_A=\bigcup_{i\in N_A} A_i$ denote the set of candidates that are approved by the voters in a profile $A$. Then, an ABC voting rule $f$ is \emph{committee separable} if $W\in f(A+B,|W|)$ implies that $W\cap C_A\in f(A,|W\cap C_A|)$ and $W\cap C_B\in f(B, |W\cap C_B|)$ for all disjoint profiles $A$, $B$ with $C_B=\mathcal{C}\setminus C_A$ and committees $W\in \mathcal{W}$. Indeed, since $C_A\cap C_B=\emptyset$, it seems reasonable that the choice of candidates from $C_A$ (resp. $C_B$) only depends on $A$ (resp. $B$). All proper rules named in this paper satisfy committee separability.
	
	\subsection{Consistent Committee Monotonicity}\label{subsec:comMon}
	
	The key axiom for our results is consistent committee monotonicity, which is a strengthening of the well-known axiom of committee monotonicity. The idea of the latter property is that the winning committees of size $k$ are derived by adding candidates to those of size $k-1$. While this is straightforward to define for ABC voting rules that always choose a single winning committee, it becomes less clear how to formalize committee monotonicity when allowing for multiple tied winning committees. We use the definition of \citet{EFSS17a} in this paper which requires that every winning committee of size $k$ is derived from a winning committee of size $k-1$ and every winning committee of size $k-1$ is extended to a winning committee of size $k$.
	
	\begin{definition}
		An ABC voting rule $f$ is \emph{committee monotone} if for every profile $A\in\mathcal{A}^*$ and $k\in \{1,\dots, m\}$, it holds that:
		\begin{enumerate}[label=(\arabic*), leftmargin=*,topsep=4pt]
			\item $W\!\in f(A,k)$ implies that there is $W'\!\in f(A,k-1)$ with $W'\subseteq W$. 
			\item $W\!\in f(A,k-1)$ implies that there is $W'\!\in f(A,k)$ with $W\subseteq W'$. 
		\end{enumerate}
	\end{definition}
		
	Committee monotone ABC voting rules are closely connected to \emph{generator functions} $g$, which take a profile $A$ and a committee $W\neq \mathcal{C}$ as input and output a possibly empty subset $g(A,W)$ of $\mathcal{C}\setminus W$. In particular, generator functions induce committee monotone ABC voting rules in a natural way: a generator function $g$ \emph{generates} an ABC voting rule $f$ if $W\in f(A, k-1)$ implies $g(A, W)\neq\emptyset$ and $f(A,k)=\{W\cup \{x\}\colon W\in f(A,k-1), x\in g(A,W)\}$ for all $k\in \{1,\dots,m\}$ and $A\in \mathcal{A}^*$. Since $f(A,0)=\{\emptyset\}$, this recursion is well-defined. 
 As we show next, committee monotonicity is equivalent to the existence of a generator function. 
	
	\begin{proposition}
		An ABC voting rule $f$ is committee monotone if and only if it is generated by a generator function $g$. 
	\end{proposition}
	\begin{proof}
		Consider an arbitrary ABC voting rule $f$ and first assume that $f$ is generated by a generator function $g$, i.e., $f(A,k)=\{W\cup \{x\}\colon W\in f(A,k-1), x\in g(A,W)\}$ for all profiles $A$ and committee sizes $k$. Now, fix a profile $A\in \mathcal{A}^*$ and a committee size $k\in \{1,\dots, m\}$. If $W\in f(A,k)$, then there is $W'\in f(A,k-1)$ and $x\in g(A,W')$ such that $W=W'\cup\{x\}$ because $g$ generates $f$. Conversely, if $W'\in f(A,k-1)$, then $g(A,W')$ cannot be empty and there is a candidate $x\in \mathcal{C}\setminus W'$ such that $W\cup \{x\}\in f(A,k)$. This shows that $f$ is committee monotone. 
		
		Next, suppose that $f$ is committee monotone. We define the generator function of $g$ as follows: if $W\not\in f(A,|W|)$, then $g(A,W)=\emptyset$. On the other hand, if $W\in f(A,|W|)$ and $W\neq \mathcal{C}$, there is a committee $W'\in f(A, |W|+1)$ with $W\subseteq W'$ due to the committee monotonicity of $f$. 
		We thus define $g(A,W)=\{x\in\mathcal{C}\setminus W\colon W\cup \{x\}\in f(A, |W|+1)\}$ if $W\in f(A,|W|)$ and let $f_g$ denote the ABC voting rule defined by $f_g(A,0)= \{\emptyset\}$ and $f_g(A,k)=\{W\cup \{x\}\colon W\in f_g(A,k-1), x\in g(A,W)\}$ for all $k>0$. We prove inductively that $f_g(A,k)=f(A,k)$ for all profiles $A$ and $k\in \{0,\dots, m\}$, which implies that $f_g$ is well-defined and that $g$ generates $f$. The induction basis $k=0$ is true since $f_g(A,0)=\{\emptyset\}=f(A,0)$ for all profiles $A$. Hence, consider a fixed $k\in \{0,\dots, m-1\}$ and $A\in\mathcal{A}^*$ and suppose that $f_g(A,k)=f(A,k)$. 
		First, let $W\in f(A,k+1)$. Due to committee monotonicity, there is $W'\in \mathcal{W}_k$ and $x\in W\setminus W'$ such that $W'\in f(A,k)=f_g(A,k)$ and $W'\cup\{x\}=W$. This implies that $x\in g(A,W')$ and hence $W\in f_g(A,k+1)$. For the other direction, let $W\in f_g(A,k+1)$, which means that there are $W'\in f_g(A,k)=f(A,k)$ and $x\in g(A,W')$ such that $W= W'\cup\{x\}$. Hence, $f(A,k+1)=f_g(A,k+1)$ and we infer inductively that $g$ generates $f$.
	\end{proof}
	
	Since a generator function completely describes its generated ABC voting rule, we can expect that a well-behaved generator function yields an attractive committee monotone ABC voting rule. Consequently, we now introduce axioms for generator functions. Our main condition on these functions is \emph{consistency}, which is concerned with the behavior of the generator function when combining two disjoint profiles. In more detail, suppose that the choice of the generator $g$ intersects for two disjoint profiles $A$ and $A'$ and a committee $W$. Intuitively, the best candidates in the combined profile $A+A'$ should be exactly those in the intersection as they are winning for the individual electorates. Hence, consistency requires for such situations that, if $g(A+A',W)\neq\emptyset$, it contains precisely the elements in the intersection of $g(A,W)$ and $g(A',W)$. 
	Note that such consistency axioms have already led to several prominent results \citep[e.g.,][]{Youn75a,Fish78d,Bran13a,LaSk21a}. Subsequently, we formally define consistency and introduce the notion of consistent committee monotonicity. The latter axiom strengthens committee monotonicity by requiring that the voting rule is generated by a consistent generator function. 
	
	\begin{definition}
		A generator function $g$ is \emph{consistent} if $g(A,W)\cap g(A',W)\neq \emptyset$ and $g(A+A', W)\neq \emptyset$ imply that $g(A+A',W)=g(A,W)\cap g(A',W)$ for all disjoint profiles $A, A'\in\mathcal{A}^*$ and committees $W\in\mathcal{W}\setminus \{\mathcal{C}\}$. An ABC voting rule $f$ is \emph{consistently committee monotone} if it is generated by a consistent generator function. 
	\end{definition}
	
	Furthermore, analogous to ABC voting rules, we call a generator function $g$ \emph{proper} if it satisfies the following conditions: 
	\begin{itemize}
		\item \emph{anonymous}:  $g(A,W)=g(\pi(A),W)$ for all $A\in\mathcal{A}^*$, $W\in\mathcal{W}\setminus \{\mathcal{C}\}$, and permutations $\pi:\mathbb{N}\rightarrow\mathbb{N}$, 
		\item \emph{neutral}: $g(\tau(A), \tau(W))=\tau(g(A,W))$ for all $A\in\mathcal{A}^*$, $W\in\mathcal{W}\setminus \{\mathcal{C}\}$, and permutations $\tau:\mathcal{C}\rightarrow\mathcal{C}$,
		\item \emph{continuous}: for all $A,A'\in\mathcal{A}^*$ and $W\in \mathcal{W}\setminus \{\mathcal{C}\}$ with $|g(A,W)|=1$ and $g(A',W)\neq\emptyset$, there is $j\in\mathbb{N}$ such that $g(jA+A', W)=g(A,W)$, and
		\item \emph{non-imposing}: for every $W\in\mathcal{W}\setminus \{\mathcal{C}\}$ and $x\in\mathcal{C}\setminus W$, there is $A\in\mathcal{A}^*$ such that $g(A,W)=\{x\}$.
	\end{itemize}
	
	Just as for ABC voting rules, all these axioms are very mild. Finally, we say that a generator function $g$ is \emph{complete} if $g(A,W)\neq \emptyset$ for all profiles $A\in\mathcal{A}^*$ and committees $W\in\mathcal{W}$. 	
		
	\subsection{Sequential Valuation Rules}\label{subsec:rules}
	
		The main goal of this paper is to characterize the class of sequential valuation rules. These rules rely on \emph{valuation functions $v$}, which are mappings of the type $v:\mathcal{A}\times \mathcal{W}\rightarrow\mathbb{R}$, to compute the outcome. Less formally, a valuation function specifies for every ballot $A_i$ and committee $W$ the number of points that a voter with ballot $A_i$ assigns to the committee $W$. The score of a committee $W$ in a profile $A$ is defined as $s_v(A,W)=\sum_{i\in N_A} v(A_i,W)$. Now, a \emph{sequential valuation function} $f$ works as follows: $f(A,0)=\{\emptyset\}$ and for $k\geq 1$, $f(A,k)=\{W\cup \{x\}\colon W\in f(A,k-1) \land \forall y\in \mathcal{C}\setminus W\colon s_v(A,W\cup\{x\})\geq s_v(A, W\cup \{y\})\}$, i.e., $f$ extends in each step the currently chosen committees with the candidates that increase the score by the most.\footnote{It is also possible to choose the committees that maximize the score for a given valuation function. These rules are proper and satisfy a consistency property for chosen committees (see \cite{LaSk21a}). However, they fail consistent committee monotonicity and it is not clear why they should be more desirable than their sequential variants.}

		Note that our definition of sequential valuation functions is so general that it includes even non-proper ABC voting rules. For instance, if $v$ is constant, the corresponding sequential valuation rule always chooses all committees of the given size and thus fails non-imposition. Nevertheless, we will focus only on proper sequential valuation rules and, in particular, on the following three subclasses.
	
	\begin{itemize}
		\item \emph{Sequential Thiele rules} rely on a Thiele counting function to compute the outcome. A Thiele counting function is a mapping $h(x):\{0,\dots,m\}\rightarrow \mathbb{R}$ which is non-negative, non-decreasing, and satisfies $h(1)>h(0)$. Then, the valuation function of a sequential Thiele rule is $v(A_i, W)=h(|A_i\cap W_i|)$. In other words, every voter values a committee only based on how many of its members she approves.\footnote{There are multiple different definitions of Thiele counting functions in the literature (e.g., \cite{LaSk22b,DDE+22a}). Our definition agrees with the one of \citeauthor{ABC+16a} \cite{ABC+16a}.}
		\item \emph{Step-dependent sequential Thiele rules} use a step-dependent Thiele counting function as valuation function. A step-dependent Thiele counting function is a mapping $h(x,y):\{0,\dots, m\}\times \{1,\dots, m\}\rightarrow \mathbb{R}$ which is non-negative, non-decreasing in $x$, and satisfies for each $y\in \{1,\dots,m-1\}$ that there is $x\in \{1,\dots, y\}$ with $h(x,y)>h(x-1,y)$. The valuation function of a step-dependent sequential Thiele rule is then $v(A_i, W)=h(|A_i\cap W|, |W|)$. Intuitively, these rules can use in every step a different Thiele counting function. 
		\item \emph{Step-dependent sequential scoring rules} compute the winner based on a step-dependent counting function. A step-dependent counting function is a mapping $h(x,y,z):\{0,\dots, m\}\times \{1,\dots, m\}\times \{1,\dots, m\}\rightarrow \mathbb{R}$ such that for every $y\in \{1,\dots, m-1\}$, there is $x\in \{1,\dots, y\}$ and $z\in \{x,\dots, m-1-(y-x)\}$ with $h(x,y,z)\neq h(x-1,y,z)$. Then, the valuation function of a step-dependent sequential scoring rule is $v(A_i, W)=h(|A_i\cap W|, |W|, |A_i|)$.
	\end{itemize}	
	The class of sequential Thiele rules contains many prominent ABC voting rules, such as \emph{sequential approval voting}\footnote{Sequential approval voting is often called approval voting since the sequential and the optimizing variant coincide. For consistency in our names, we prefer to call this rule sequential approval voting.} (\texttt{seqAV}) defined by $h(x)=x$, \emph{sequential proportional approval voting} (\texttt{seqPAV}) defined by $h(0)=0$ and $h(x)=\sum_{i=1}^x \frac{1}{i}$ for $x>0$, and \emph{sequential Chamberlin-Courant approval voting} (\texttt{seqCCAV}) defined by $h(0)=0$ and $h(x)=1$ for $x>0$. 
    An example of a step-dependent sequential Thiele rule can be constructed by switching between \texttt{seqAV} and \texttt{seqCCAV} in the different steps. Finally, sequential satisfaction approval voting (\texttt{seqSAV}), defined by $h(x,y,z)=\frac{x}{z}$, is an example of a step-dependent sequential scoring rule. 
 
	It is easy to see that every sequential valuation function $f$ is consistently committee monotone as it can be verified that its generator function $g(A, W)=\{x\in\mathcal{C}\setminus W\colon \forall y\in\mathcal{C}\setminus W\colon s_v(A,W\cup \{x\})\geq s_v(A,W\cup\{y\}\}$ is consistent (here, $v$ denotes the valuation function of $f$). Furthermore, all step-dependent sequential scoring rules are proper ABC voting rules. In particular, the technical condition on $h$ is necessary to ensure that step-dependent sequential scoring rules are non-imposing. Finally, note that every sequential Thiele rule is a step-dependent sequential Thiele rule, which are in turn step-dependent sequential scoring rules. Consequently, all three classes of sequential valuation rules only contain proper ABC voting rules. We can even make the relation between these different types of rules precise as shown in the next proposition.
	
	\begin{proposition}\label{prop:relation}
		The following equivalences hold: 
		\begin{enumerate}[label=(\arabic*), leftmargin=*,topsep=4pt]
			\item A sequential valuation rule is a step-dependent sequential scoring rule if and only if it is proper.
			\item A step-dependent sequential scoring rule is a step-dependent sequential Thiele rule if and only if it is independent of losers. 
			\item A step-dependent sequential Thiele rule is a sequential Thiele rule if and only if it is committee separable. 
		\end{enumerate}
	\end{proposition}
	\begin{proof}[Proof Sketch]
	   The "only if" part of the claims is always easy to prove as it is, e.g., straightforward to see that every step-dependent sequential scoring rule is a proper sequential valuation rule. Hence, we focus on the "if" part. The key insight for (1) is that the valuation function $v$ of a proper sequential valuation rule is neutral, i.e., $v(A_i, W)=v(\tau(A_i), \tau(W))$ for all ballots $A_i$, committees $W$, and permutations $\tau:\mathcal{C}\rightarrow\mathcal{C}$. Since $|A_i|=|\tau(A_i)|$, $|W|=|\tau(W)|$, and $|A_i\cap W|=|\tau(A_i\cap W)|$, for all ballots $A_i$, committees $W$, and permutations $\tau$, the corresponding sequential valuation rule is a step-dependent sequential scoring rule. For (2), the "if" part intuitively holds because independence of losers excludes the possibility that the step-dependent Thiele counting function $h$ depends on the size of the ballot. By formalizing this insight, we can construct a step-dependent Thiele counting function that induces $f$, which proves (2). Finally, the "if" part of (3) follows since committee separability relates the different steps of the rule. In more detail, we can construct two disjoint profiles $A$, $B$ such that $f(A+B, |C_A|)=\{C_A\}$ and then, committee separability shows that all following steps must be equal to the choice for $B$. Formalizing this argument rules out that $h$ depends on $|W|$ and we thus end up with a sequential Thiele rule.
	\end{proof}
	
	\section{Characterizations of Sequential Valuation Rules}\label{sec:results}
	We are now ready to discuss our main result, a characterization of step-dependent sequential scoring rules: an ABC voting rule is a step-dependent sequential scoring rule if and only if it is proper and consistently committee monotone. Combined with \Cref{prop:relation}, we infer as corollary also characterizations of step-dependent sequential Thiele rules and sequential Thiele rules. Moreover, this proposition also emphasizes the generality of our result since characterizing step-dependent sequential scoring rules is equivalent to characterizing all proper sequential valuation rules. Due to space constraints, we defer the proofs of all auxiliary propositions to the appendix and discuss proof sketches instead. 
	
	While it is quite easy to show that every step-dependent sequential scoring rule is proper and consistently committee monotone, the converse claim is much more involved. Our main idea for proving this direction is to investigate the generator function of consistently committee monotone and proper ABC voting rules. Hence, we first verify the conjecture that attractive committee monotone ABC voting rules are generated by well-behaved generator functions.
	
	\begin{restatable}{proposition}{niceGenerator}\label{prop:niceGenerator}
		An ABC voting rule is proper and consistently committee monotone if and only if it is generated by a proper, consistent, and complete generator function.
	\end{restatable}
	\begin{proof}[Proof Sketch]
	If $f$ is generated by a proper, consistent, and complete generator function, it is fairly straightforward that it is consistently committee monotone and proper. We thus focus on the inverse direction and suppose that $f$ is a proper and consistently committee monotone ABC voting rule. The key insight for this direction is that non-imposition and continuity can be generalized to sequences of committees $W_1,\dots, W_\ell$ with $|W_k|=k$ and $W_{k-1}\subseteq W_k$ for all $k\in \{1,\dots, \ell\}$ (we assume subsequently that $W_0=\emptyset$):
	\begin{enumerate}[label=(\arabic*), leftmargin=*,topsep=4pt]
	    \item If $\ell<m$, there is a profile $A$ such that $f(A,k)=\{W_k\}$ for all $k\in \{1,\dots, \ell\}$ and $f(A,\ell+1)=\{W_{\ell}\cup \{x\}\colon x\in \mathcal{C}\setminus W_\ell\}$.
	    \item For any two profiles $A, A'$ such that $f(A,k)=\{W_k\}$ for all $k\in \{1,\dots, \ell\}$, there is an integer $j$ such that $f(jA+A', k)=\{W_k\}$ for all $k\in \{1,\dots, \ell\}$. 
	\end{enumerate}

    For instance, we prove (1) by an induction on the length of the sequence: by non-imposition, there is a profile $A^1$ for every committee $W_{\ell+1}\in\mathcal{W}_{\ell+1}$ such that $f(A^1,\ell+1)=\{W_{\ell+1}\}$. Committee monotonicity implies then that there is a sequence of committees $W_1,\dots, W_\ell$ such that $W_k\in f(A^1, k)$ and $W_{k+1}\setminus W_k\subseteq g(A, W_k)$ for all $k\in \{1,\dots,\ell\}$, where $g$ is a consistent generator function of $f$. By the induction hypothesis, there is a profile $A^2$ such that $f(A^2, k)=\{W_k\}$ for all $k\in \{1,\dots, \ell\}$ and $f(A^2,\ell+1)=\{W_\ell\cup\{x\}\colon x\in\mathcal{C}\setminus W_\ell\}$. We can now use the consistency of $g$ to infer that $f(A^1+A^2, k)=\{W_k\}$ for all $k\in \{1,\dots,\ell+1\}$. Finally, we can further modify the profile to ensure that $W_{\ell+1}$ is extended by all remaining candidates by using anonymity and neutrality.
    
    Now, we will extend the consistent generator function $g$ of $f$ to make it complete. Consider for this a sequence of committees $W_1,\dots, W_\ell$ with $|W_k|=k$ and $W_{k-1}\subseteq W_k$ for all $k\in \{1,\dots, \ell\}$. Due to (1), there is a profile $A^{W_\ell}$ with $f(A^{W_\ell}, k)=\{W_k\}$ for all $k\in \{1,\dots, \ell\}$ and $f(A^{W_\ell}, \ell+1)=\{W_\ell\cup \{x\}\colon x\in\mathcal{C}\setminus W_\ell\}$. We define the function $\hat g(A,W_\ell)=g(A+jA^{W_\ell}, W_\ell)$, where $j$ is the smallest integer such that $f(A+jA^{W_\ell},k)=\{W_k\}$ for all $k\in \{1,\dots,\ell\}$; such an integer exists because of (2). First, note that $\hat g$ generates $f$ since $\hat g(A,W)=g(A,W)$ for all $A\in\mathcal{A}^*$ and $W\in f(A,|W|)$. This follows from consistent committee monotonicity as $g(jA^W, W)=\mathcal{C}\setminus W$, $g(A,W)\neq\emptyset$, and $g(A+jA^W, W)\neq\emptyset$. Finally, $\hat g$ satisfies anonymity, neutrality, non-imposition, and continuity as it generates $f$ and $f$ would fail these properties otherwise. 
	\end{proof}
	
	As second step, we characterize the class of proper, consistent, and complete generator functions. In particular, we show that for every committee $W\neq \mathcal{C}$, $g(A,W)$ can be described by a weighted variant of single winner approval voting. For making this formal, let $v(x,y):\{0,\dots, m\}\times \{1,\dots, m\}\rightarrow \mathbb{R}$ be a weight function. Then, $v$-weighted approval voting is defined as the generator function $\mathit{AV}_{v}(A,W)=\{c\in\mathcal{C}\setminus W\colon \forall d\in\mathcal{C}\setminus W\colon \sum_{i\in N_A\colon c\in A_i} v(|W\cap A_i|, |A_i|)\geq \sum_{i\in N_A\colon d\in A_i} v(|W\cap A_i|, |A_i|)\}$. 
	
	\begin{restatable}{proposition}{hyperplane}\label{prop:hyperplane}
		Let $g$ denote a proper, consistent, and complete generator function. For every committee $W\neq\mathcal{C}$, there is a weight function $v^W$ such that $g(A,W)=\mathit{AV}_{v^W}(A,W)$ for all profiles $A\in\mathcal{A}^*$.
	\end{restatable}
	\begin{proof}[Proof Sketch]
	Let $g$ denote a proper, consistent, and complete generator function and fix a committee $W\neq\mathcal{C}$. We show the proposition by applying a separating hyperplane argument analogous to how \citet{Youn75a} derives his characterization of scoring rules. 
	
	For doing so, we first transform the domain of $g(\cdot,W)$ from preference profiles to a numerical space and we show thus that $g(\cdot,W)$ can be computed only based on the values $n(c,A,W,k,\ell)=|\{i\in N\colon c\in A_i\land |A_i\cap W|=k\land |A_i|=\ell\}|$ for $c\in\mathcal{C}\setminus W$, $k\in \{0,\dots,|W|\}$, and $\ell\in \{k+1, \dots, m-1-|W|+k\}$. For proving this, we first show that if $A_i\cap W=A_j\cap W$ and $|A_i|=|A_j|$ for all $i,j\in N_A$ and all candidates $x\in\mathcal{C}\setminus W$ are approved by the same number of voters, then $g(A,W)=\mathcal{C}\setminus W$. 
	Once this restricted claim is proven, we can use our axioms to generalize it; e.g., consistency, neutrality, and anonymity then entail that, for all $k, \ell$, $g(A^{k,\ell},W)=\mathcal{C}\setminus W$ if $|A_i^{k,\ell}\cap W|=k$ and $|A_i|=\ell$ for all $i\in N_A$ and all candidates $x\in\mathcal{C}\setminus W$ have the same approval score. Finally, this means that if there are constants $c_{k,\ell}$ such that $n(x,A,W,k,\ell)=c_{k,\ell}$ for all candidates $c\in\mathcal{C}\setminus W$ and indices $k$ and $\ell$, then $g(A,W)=\mathcal{C}\setminus W$ as we can decompose $A$ with respect to $k$ and $\ell$ into these profiles $A^{k,\ell}$. Together with consistency, we infer from this observation that $g(\cdot, W)$ can indeed be computed based on on the matrix $N(A,W)$ that contains all the values $n(c,A,W,k,\ell)$. 
	
	As next step, we use standard constructions to extend the domain of $g$ further from integer matrices $N(A,W)$ to rational matrices. To this end, let $Q_2$ be the matrix that corresponds to the profile in which each ballot is reported once and note that $g(Q_2, W)=\mathcal{C}\setminus W$ due to anonymity and neutrality. Based on this profile, we extend $g$ to negative numbers by defining $g(Q_1, W)=g(Q_1+jQ_2,W)$ (where $j\in\mathbb{N}$ is a scalar such that $Q_1+jQ_2$ contains only positive integers) and as second step to $g$ to rational numbers by defining $g(Q_1, W)=g(jQ_1, W)$ (where $j$ is the smallest integer such that $jQ_1$ only contains integers). For both steps, consistency ensures that $g$ remains well-defined. Moreover, the extension of $g(\cdot, W)$ to rational numbers preserves all desirable properties of $g$. 
 
    Finally, we partition the feasible input matrices $Q$ into sets $R_c=\{Q\colon c\in g(Q,W)\}$ for $c\in\mathcal{C}\setminus W$. These sets are convex (with respect to $\mathbb{Q}$) and symmetric since $g$ is consistent, anonymous, and neutral. Moreover, the interior of $R_c$ and $R_d$ is disjoint for $c,d\in\mathcal{C}\setminus W$ with $c\neq d$ and we can thus derive a separating hyperplane between these sets (see \citet{McLe18a}). As last step, we infer from these hyperplanes the weight function $v^W$.
	\end{proof}
		
	Based on \Cref{prop:hyperplane}, we finally prove our main result.
	
	\begin{theorem}\label{thm:characterization}
		An ABC voting rule is a step-dependent sequential scoring if and only if it is proper and consistently committee monotone. 
	\end{theorem}
	\begin{proof}
		We show in \Cref{prop:relation} that every step-dependent sequential scoring rule $f$ is proper. For proving that $f$ is consistently committee monotone, let $h$ denote its step-dependent counting function. Moreover, let $W^x=W\cup\{x\}$ for every committee $W$ and candidate $x\in\mathcal{C}\setminus W$. By definition, $f(A,0)=\emptyset$ and $f(A,k)=\{W^c\colon W\in f(A,k-1), c\in\mathcal{C}\setminus W\colon\forall d\in\mathcal{C}\setminus W\colon s_h(A,W^c)\geq s_h(A,W^d)\}$. Thus, $g(A,W)=\{c\in\mathcal{C}\setminus W\colon \forall d\in\mathcal{C}\setminus W\colon s_h(A,W^c)\geq s_h(A,W^d)\}$ is complete and generates $f$. Moreover, $g$ is consistent since the scores are additive, i.e., $s_h(A+A', W)=s_h(A,W)+s_h(A',W)$ for all profiles $A,A'$ and committees $W$. Hence, if $s_h(A,W^c)\geq s_h(A,W^d)$ and $s_h(A',W^c)\geq s_h(A',W^d)$, then $s_h(A+A', W^c)\geq s_h(A+A',W^d)$. Moreover, if one of the inequalities is strict for $A$ or $A$', so it is for $A+A'$. Thus, $g(A+A',W)=g(A,W)\cap g(A',W)$ if $g(A,W)\cap g(A',W)\neq\emptyset$, which proves that $g$ is consistent.
		
		For the other direction, consider a proper and consistently committee monotone ABC voting rule $f$. By \Cref{prop:niceGenerator}, $f$ is generated by a proper, consistent, and complete generator function $g$. Furthermore, by \Cref{prop:hyperplane}, there is for every committee $W\neq \mathcal{C}$ a weight function $v^W$ such that $g(A,W)=\mathit{AV}_{v^W}(A,W)$ for all $A\in\mathcal{A}^*$. Now, consider two committees $W$ and $W'$ with $|W|=|W'|<m$ and let $v^W$ and $v^{W'}$ denote the corresponding weight functions. We first show that $\mathit{AV}_{v^W}(A', W')=\mathit{AV}_{v^{W'}}(A', W')$ for every profile $A'$. For this, let $c'\in \mathit{AV}_{v^{W'}}(A', W')$ which is the case if and only if $\sum_{i\in N_{A'}\colon c'\in A_i'} v^{W'}(|W'\cap A_i'|, |A_i'|)\geq \sum_{i\in N_{A'}\colon d'\in A_i'} v^{W'}(|W'\cap A_i'|, |A_i'|)$ for all $d'\in\mathcal{C}\setminus W'$. Next, let $\tau:\mathcal{C}\rightarrow\mathcal{C}$ denote a permutation such that $\tau(W)=W'$, and let $A\in \mathcal{A}^*$ and $c\in\mathcal{C}$ such that $\tau(A)=A'$ and $\tau(c)=c'$. Because of $g(A,W)=\mathit{AV}_{v^W}(A,W)$, $g(A',W')=\mathit{AV}_{v^{W'}}(A',W')$, and the neutrality of $g$, it holds that $c'\in \mathit{AV}_{v^{W'}}(A',W')$ if and only if $c\in \mathit{AV}_{v^{W}}(A,W)$. By the definition of $\mathit{AV}_{v^W}$, the last claim is true if and only if $\sum_{i\in N_{A}\colon c\in A_i} v^{W}(|W\cap A_i|, |A_i|)\geq \sum_{i\in N_{A}\colon d\in A_i} v^{W}(|W\cap A_i|, |A_i|)$ for all $d\in\mathcal{C}\setminus W$. Finally, observe that $x\in A_i$ if and only if $\tau(x)\in A_i'$, $|A_i|=|A_i'|$, and $|W\cap A_i|=|W'\cap A_i'|$ for all candidates $x\in \mathcal{C}\setminus W$ and voters $i\in N_A$. Hence, we conclude that $\sum_{i\in N_{A}\colon c\in A_i} v^{W}(|W\cap A_i|, |A_i|)\geq \sum_{i\in N_{A}\colon d\in A_i} v^{W}(|W\cap A_i|, |A_i|)$ if and only if $\sum_{i\in N_{A}\colon c'\in A_i'} v^{W}(|W'\cap A_i'|, |A_i'|)\geq \sum_{i\in N_{A}\colon \tau(d)\in A_i'} v^{W}(|W'\cap A_i'|, |A_i'|)$ for all $d\in\mathcal{C}\setminus W$. So, $c'$ obtains the maximal score in $A'$ with respect to $v^{W'}$ if and only if the same holds with respect to $v^W$. This proves that $\mathit{AV}_{v^W}(A', W')=\mathit{AV}_{v^{W'}}(A', W')$ for all profiles $A'$ and committees $W,W'$ with $|W|=|W'|<m$.
		
		Next, let $W_0, \dots, W_{m-1}$ denote committees such that $|W_i|=i$ and let $v^i=v^{W_i}$. We define the function $v(x,y,z):\{0,\dots, m\}\times\{0,\dots, m-1\}\times \{1,\dots, m\}\rightarrow\mathbb{R}$ by $v(x,y,z)=v^y(x,z)$. By our previous reasoning, it holds that $g(A,W)=\mathit{AV}_{v^{|W|}}(A,W)=\{c\in\mathcal{C}\setminus W\colon \forall d\in\mathcal{C}\setminus W\colon \sum_{i\in N_A\colon c\in A_i} v(|A_i\cap W|, |W|, |A_i|)\geq \sum_{i\in N_A\colon d\in A_i} v(|A_i\cap W|, |W|, |A_i|)\}$. Our next goal is to derive a valuation function from $v$. For doing so, define the function $h(x,y,z):\{0,\dots, m\}\times\{1,\dots, m\}\times \{1,\dots, m\}\rightarrow\mathbb{R}$ as follows: $h(0,y,z)=0$ for all $y,z\in\{1,\dots, m\}$ and $h(x,y,z)=h(x-1,y,z)+v(x-1,y-1,z)$ for all $x,y,z\in \{1,\dots,m\}$. We claim that $f$ is the sequential valuation rule induced by the valuation function $w(A_i, W)=h(|A_i\cap W|, |W|, |A_i|)$. For this, let $g_w(A,W)=\{c\in\mathcal{C}\setminus W\colon \forall d\in\mathcal{C}\setminus W\colon \sum_{i\in N_A} w(A_i, W\cup \{c\})\geq \sum_{i\in N_A} w(A_i, W\cup \{d\})\}$. We will show that $g_w(A,W)=g(A,W)$ for all profiles $A\in\mathcal{A}^*$ and committees $W\neq \mathcal{C}$. Note for this that for all profiles $A$, committees $W$, and candidates $c\in\mathcal{C}\setminus W$, the following equation holds:
		\begin{align*}
			&\sum_{i\in N_A} h(|W^c\cap A_i|, |W^c|, |A_i|) - h(|W\cap A_i|, |W^c|, |A_i|)\\
			&=\!\!\sum_{i\in N_A\colon c\in A_i} \!\!h(|W\cap A_i|+1, |W^c|, |A_i|) - h(|W\cap A_i|, |W^c|, |A_i|) \\
			&\qquad+\!\! \sum_{i\in N_A\colon c\not \in A_i} \!\! h(|W\cap A_i|, |W^c|, |A_i|) - h(|W\cap A_i|, |W^c|, |A_i|)\\
			&=\sum_{i\in N_A\colon c\in A_i} v(|W\cap A_i|, |W|, |A_i|).
		\end{align*}

		Now, define $C(A,W)=\sum_{i\in N_A} h(|W\cap A_i|, |W|+1, |A_i|)$. Then, the above equation shows that $s_w(A, W^c)\geq s_w(A, W^d)$ if and only if $s_w(A, W^c)-C(A,W)\geq s_w(A, W^d)-C(A,W)$ if and only if $\sum_{i\in N_{A}\colon c\in A_i} v(|W\cap A_i|, |W|, |A_i|)\geq \sum_{i\in N_{A}\colon d\in A_i} v(|W\cap A_i|, |W|, |A_i|)$. Hence, $g_w(A,W)=g(A,W)$ for all profiles $A$ and committees $W$ and $f$ is the sequential valuation rule generated by $g$. Finally, since $f$ is proper, \Cref{prop:relation} shows that it is a step-dependent sequential valuation rule.
	\end{proof}
	
	Due to \Cref{prop:relation}, \Cref{thm:characterization} entails also characterizations of step-dependent sequential Thiele rules and sequential Thiele rules.
	
	\begin{corollary}\label{corollary} The following statements hold:
		\begin{enumerate}[label=(\arabic*), leftmargin=*,topsep=4pt]
			\item An ABC voting rule is a step-dependent sequential Thiele rule if and only if it is consistently committee monotone, independent of losers, and proper. 
			\item An ABC voting rule is a sequential Thiele rule if and only if it is consistently committee monotone, independent of losers, committee separable, and proper.
		\end{enumerate}
	\end{corollary}
		
	\begin{remark}\label{rem:independence}
			All axioms are required for \Cref{thm:characterization} as there are ABC voting rules other than step-dependent sequential scoring rules that satisfy all but one condition. If we omit anonymity, we can use \texttt{seqAV} but count the vote of voter $1$ twice. When omitting neutrality, we can use \texttt{seqAV} but count the votes for candidate $a$ twice. When omitting non-imposition, the rule that always returns all committees of the given size satisfies all remaining conditions. The rule that refines the generator of $\texttt{seqAV}$ by breaking ties based on the Chamberlin-Courant score only fails continuity. Finally, when omitting consistent committee monotonicity, Thiele rules satisfy all remaining conditions. We can also not weaken consistent committee monotonicity to committee monotonicity as reverse sequential Thiele rules then satisfy all given conditions. 
	\end{remark}
	
	\begin{remark}
		Our hierarchy of sequential valuation rules misses the class of sequential scoring rules, which are defined by a valuation function of the form $v(A_i, W)=h(|A_i\cap W|, |A_i|)$. These rules form a subclass of step-dependent sequential scoring rules, but committee separability does not characterize them within the class of step-dependent sequential scoring rules. 
	\end{remark}
	
	\begin{remark}
		A natural follow-up question to \Cref{thm:characterization} is whether sequential valuation rules can be characterized by consistent committee monotonicity, anonymity, and continuity since they satisfy these three axioms. Unfortunately, this is not the case as we can still treat candidates differently (see \Cref{rem:independence}). On the other hand, it might be possible to characterize the rules that satisfy anonymity, neutrality, continuity, and consistent committee monotonicity. 
	\end{remark}
	
	\section{Characterizations of Specific ABC Voting Rules}

	Finally, we leverage our results to derive characterizations of specific voting rules. Note here that our characterizations can be combined with known results that single out rules within the class of, e.g., sequential Thiele rules, to derive full characterizations (see, e.g., \cite{LaSk18a,LaSk21a}). Nevertheless, we prefer to present our own characterizations for \texttt{seqCCAV}, \texttt{seqAV}, and \texttt{seqPAV} to highlight new aspects of these rules. We state our results restricted to the class of sequential Thiele rules; \Cref{corollary} turns them into full characterizations by adding the necessary axioms. Moreover, we focus on the case $m\geq 3$ since every sequential Thiele rule coincides with \texttt{seqAV} if $m=2$.

	The main idea for our characterizations is to study how ABC voting rules treat clones. To this end, we say that two candidates $c,d$ are \emph{clones} in a profile $A$ if $c\in A_i$ if and only if $d\in A_i$ for all voters $i\in N_A$. Depending on the goal of the election, clones should be treated differently. For instance, if our goal is to choose a committee that is as diverse as possible, there is no point in choosing both clones. We formalize this new condition as follows: an ABC voting rule $f$ is \emph{clone-rejecting} if $f(A,|W|)=\{W\}$ implies that $\{c,d\}\not\subseteq W$ for all profiles $A$ with clones $c,d$ and committees $W\neq\mathcal{C}$. The requirement that a single committee is chosen is necessary since, for instance, in the profile where all voters approve all candidates, we need to choose clones but we will also choose multiple committees. As our next result shows, this axiom characterizes \texttt{seqCCAV}.
	
	\begin{restatable}{theorem}{seqCCAV}\label{thm:seqCCAV}
		\emph{\texttt{seqCCAV}} is the only sequential Thiele rule that satisfies clone-rejection if $m\geq 3$. 
	\end{restatable}
	\begin{proof}
	Since \texttt{seqCCAV} clearly satisfies clone-rejection, we focus on the inverse direction. Hence, consider a sequential Thiele rule $f$ other than \texttt{seqCCAV} and let $h$ denote its Thiele counting function. Since sequential Thiele functions are invariant under scaling and shifting $h$, we can suppose that $h(0)=0$ and $h(1)=1$. Moreover, because $f$ is not \texttt{seqCCAV}, there is an integer $x\in \{2,\dots, m-1\}$ such that $h(x)>1$ and $h(x')=1$ for all $x'\in \{1,\dots,x-1\}$. Now, let $\Delta=h(x)-1$ and $\ell\in\mathbb{N}$ such that $\ell\Delta>1$. We consider the following profile $A$ to show that $f$ fails clone-rejection: there are $\ell$ voters who approve the candidates $c_1,\dots, c_x$, $x$ voters who approve $c_1$ and $c_2$, and for each $i\in \{3,\dots,x+1\}$ there are $x+2-i$ voters who approve only $c_i$. Now, due to the minimality of $x$, $f$ agrees in the first $x-1$ rounds with \texttt{seqCCAV} and we thus have that $f(A,x-1)=\{\{c_1,c_3,\dots, c_x\}, \{c_2,c_3,\dots, c_x\}\}$. On the other hand, it holds that $s_h(A, \{c_1,\dots, c_x\})\geq s_h(A, \{c_1,c_3,\dots, c_x\})+\ell\Delta>s_h(A, \{c_1,c_3,\dots, c_x\})+1$ and $s_h(A, \{c_1,c_3,\dots, c_x,c_{x+1}\})=s_h(A, \{c_2,c_3,\dots, c_x,c_{x+1}\})=s_h(A, \{c_1,c_3,\dots, c_x\})+1$. Thus, $f(A,x)=\{\{c_1,\dots, c_x\}\}$. However, this committee contains the clones $c_1$ and $c_2$ which proves that $f$ fails clone-rejection.
	\end{proof}
	
    The polar opposite to diverse committees are quality-based ones, where the goal is to find the $k$ best candidates regardless of how well they represent the voters. In such a setting, clones should be treated as equal as possible and we thus propose the following notion: an ABC voting rule $f$ is \emph{clone-accepting} if for all profiles $A$ with clones $c,d$ and committees $W\subseteq \mathcal{C}\setminus \{c,d\}$, it holds that $W\cup \{c\}\in f(A,|W\cup \{c\}|)$ implies that $W\cup \{c,d\}\in f(A,|W\cup\{c,d\}|)$. Or, in words, the only reason that a winning committee does not contain both clones is if this conflicts with the committee size. Perhaps surprisingly, clone-acceptance does not characterize \texttt{seqAV} as, e.g., the sequential Thiele rule defined by $h(0)=0$, $h(1)=1$, and $h(x)=2x+1$ for $x\geq 2$ satisfies this axiom, too. However, this rule prefers to choose candidates that are approved by voters who already approve a chosen candidate. This behavior can be interpreted as trust in a voter's recommendation and can be reasonable for quality-based elections. Nevertheless, to single out \texttt{seqAV}, we use a mild condition prohibiting this behavior: an ABC voting rule $f$ is \emph{distrusting} if for all profiles $A$, committees $W\neq\mathcal{C}$ with $f(A,|W|)=\{W\}$, and candidates $b,c$, it holds that $b\in W$ implies $c\in W$ if more voters in $A$ report the ballot $\{c\}$ than there are voters who approve $b$. Based on these two axioms, we derive the following theorem. 
	
	\begin{restatable}{theorem}{seqAV}\label{thm:seqAV}
		\emph{\texttt{seqAV}} is the only sequential Thiele rule that is clone-accepting and distrusting if $m\geq 3$.
	\end{restatable}
	\begin{proof}[Proof Sketch]
	    We focus on the direction from right to left and thus consider a sequential Thiele rule $f$ other than \texttt{seqAV}. Moreover, let $h$ denote the corresponding Thiele counting function and suppose again that $h(0)=0$ and $h(1)=1$. Since $f$ is not \texttt{seqAV}, there is a integer $x\in \{2,\dots, m-1\}$ such that $h(x)\neq x$ but $h(x')=x'$ for $x'\in \{1,\dots, x-1\}$. Now, let $\Delta=|h(x)-x|$ and $\ell\in\mathbb{N}$ such that $\ell\Delta>1$. If $h(x)>x$, $f$ fails distrust in the following profile $A$, where $W$ is a committee of size $x-1\leq m-2$ and $c,d\in\mathcal{C}\setminus W$: $\ell$ voters approve $W\cup \{c\}$, $\ell+1$ voters approve $d$, and two voters approve $W$. Indeed, it can be checked that $f(A,x)=\{W\cup \{c\}\}$ but distrust requires that $d$ is not chosen after $c$. On the other hand, if $h(x)<x$, $f$ fails clone-acceptance in the following profile $A$, where $W$ is a committee $W$ with $|W|=x-2\leq m-3$ and $b,c,d\in\mathcal{C}\setminus W$: $\ell$ voters report $W\cup \{c,d\}$ and $\ell-1$ voters report $b$. Indeed, $f(A,x-1)=\{W\cup \{c\}, W\cup \{d\}\}$ but $f(A,x)=\{W\cup \{b,c\}, W\cup \{b,d\}\}$. Thus, \texttt{seqAV} is the only distrusting and clone-accepting sequential Thiele rule.
	\end{proof}
	
	Finally, a large stream of research on ABC voting rules tries to find proportional committees, i.e., the chosen committee should proportionally reflect the voters' preferences. For defining this concept, we rely on heavily restricted profiles $A$ in which $n_1$ voters report the same ballot $A_1$ and $n_2$ voters approve a single candidate $c\not\in A_1$. In such a profile, each clone $d\in A_1$ that is in the elected committee $W$ represents on average $\frac{n_1}{|A_1\cap W|}$ voters, whereas the candidate $c$ represents $n_2$ voters. Following the idea of proportionality, we should choose a subset of $A_1$ for a committee size $k$ if $\frac{n_1}{k}>n_2$ as every candidate $d\in A_1$ represents on average more voters than $c$.
	Conversely, if $\frac{n_1}{k}<n_2$, the chosen committee should contain $c$. Thus, we say an ABC voting rule is \emph{clone-proportional} if for all such profiles $A$, committee sizes $k\leq |A_1|$, and committees $W\in f(A,k)$, it holds that $c\not\in W$ if $\frac{n_1}{k}>n_2$ and $c\in W$ if $\frac{n_1}{k}<n_2$. Note that clone-proportionality is closely related to D'Hondt proportionality \citep{BLS18a,LaSk21a}. Next, we show that this axiom characterizes \texttt{seqPAV}.
	
	\begin{restatable}{theorem}{seqPAV}\label{thm:seqPAV}
		\emph{\texttt{seqPAV}} is the only sequential Thiele rule that satisfies clone-proportionality if $m\geq 3$.
	\end{restatable}
	\begin{proof}[Proof Sketch]
	We only show that no other sequential Thiele rule $f$ but \texttt{seqPAV} satisfies clone-proportionality. For this, let $h$ denote the Thiele counting function of $f$ and normalize $h$ such that $h(0)=0$ and $h(1)=1$. Since $f$ is not \texttt{seqPAV}, there is a minimal integer $x\in \{2,\dots, m-1\}$ such that $h(x)\neq \sum_{i=1}^x\frac{1}{i}$. As in the proofs of \Cref{thm:seqCCAV,thm:seqAV}, we can now construct a profile in which $f$ fails clone-proportionality. For instance, if $h(x)>\sum_{i=1}^x \frac{1}{i}$, let $\Delta=h(x)-\sum_{i=1}^x \frac{1}{i}$ and $\ell\in\mathbb{N}$ such that $\ell x\cdot \Delta>1$ and consider the following profile $A$: $\ell x$ voters report $\{c_1,\dots,c_x\}$ and $\ell+1$ voters approve a single candidate $c\not\in \{c_1,\dots,c_x\}$. It can be checked that $f(A,x)=\{\{c_1,\dots, c_x\}\}$ but clone-proportionality requires that $c\in W$ for $W\in f(A,x)$ as $\ell+1>\frac{\ell x}{x}$. A similar counter example can be constructed if $h(x)<\sum_{i=1}^x\frac{1}{i}$ and thus, \texttt{seqPAV} is the only sequential Thiele rule that satisfies this axiom.
	\end{proof}

	\begin{remark}
		Notably, clone-acceptance characterizes \texttt{seqAV} within the class of sequential Thiele rules with non-increasing partial sums $h(j)-h(j-1)$. In the literature, the definition of sequential Thiele rules often includes this condition. 
	\end{remark}
	
	\section{Conclusion}
	
	In this paper, we provide axiomatic characterizations for the new class of sequential valuation rules. These rules are based on valuation functions, which assign each pair of ballot and committee a score and compute the winning committees greedily by extending the current winning committees with the candidates that increase the score by the most. Clearly, sequential valuation rules generalize the prominent class of sequential Thiele rules whose valuation function only depends on the size of the intersection between the given ballot and committee. Our main result characterizes the class of proper (=anonymous, neutral, continuous, and non-imposing) sequential valuation rules based on a new axiom called consistent committee monotonicity. This axiom combines the well-known notions of committee monotonicity and consistency by requiring that the winning committees of size $k$ are derived from those of size $k-1$ by only adding new candidates, and that these newly added candidates are chosen in a consistent way across the profiles. By adding additional conditions, we also derive characterizations of important subclasses such as sequential Thiele rules and of prominent ABC voting rules such as sequential proportional approval voting. For a full overview of our results, we refer to \Cref{fig:overview}.
	
	Our theorems address one of the major open problems in the field of ABC voting: while there is an enormous number of different voting rules, there are almost no characterizations. Such characterizations are crucial for reasoning about which rule to use because without a characterization, there is always the possibility that a more attractive rule exists. Moreover, many ideas of our results seem rather universal and it might be possible to re-use them to characterize other rules such as Phragmen's rule or Thiele rules.

 \section*{Acknowledgements}
 
 We thank Felix Brandt for helpful feedback. This work was supported by the Deutsche Forschungsgemeinschaft under grants \mbox{BR 2312/11-2} and \mbox{BR 2312/12-1}.

	\clearpage
	\appendix
	
	\section{Appendix: Proofs}
	
	In this appendix, we discuss the proofs omitted in the main body. Note that we use some additional notation which has already been used in some of the proofs in the main body. In particular, we define a sequence of committees $W_1,\dots, W_\ell$ as a set of committees such that $|W_k|=k$ and $W_{k-1}\subseteq W_k$ (where $W_0=\emptyset$) for all $k\in \{1,\dots, m\}$. Moreover, given a committee $W$ and a candidate $x\in\mathcal{C}\setminus W$, we let $W^c=W\cup \{c\}$. Finally, we want to mention that we place the proofs of more involved statements in own subsections, and that we do not order these subsections according to the appearance of the corresponding result in the main body. 
	
	\subsection{Proof of \Cref{prop:niceGenerator}}
	
	As first result, we prove \Cref{prop:niceGenerator}: a proper ABC voting rule is consistently committee monotone if and only if there is a proper, consistent, and complete generator function that generates $f$. For proving this proposition, we show a number of auxiliary claims, which will also be helpful for proving other results. We start by proving that the consistency of the generator function implies a mild variant of consistency for the generated ABC voting rule itself.
	
	\begin{lemma}\label{lem:merge}
		Let $f$ denote a consistently committee monotone ABC voting rule and let $g$ denote a consistent generator function of $f$. For all profiles $A, A'\in\mathcal{A}^*$ and sequences of committees $W_1,\dots,W_\ell$ such that $f(A, k)=\{W_k\}\subseteq f(A',k)$, and $W_k\setminus W_{k-1}\subseteq g(A',W_{k-1})$ for all $k\in \{1,\dots, \ell\}$, it holds that $f(A+A',k)=\{W_k\}$ for all $k\in \{1,\dots,\ell\}$.
	\end{lemma}
	\begin{proof}
		Let $f$, $g$, $A$, $A'$, and $W_1,\dots, W_{\ell}$ be defined as in the lemma. The lemma follows by repeatedly using the consistency of $g$, which results formally in an induction. The induction basis $k=0$ is by definition true since $f(A+A',0)=f(A,0)=f(A',0)=\{\emptyset\}$. For the induction step, fix some $k\in \{0,\dots, \ell-1\}$ and assume that $f(A+A',k)=\{W_k\}$. Hence, for each $X\in\{A, A', A+A'\}$, we have that $W_k\in f(X,k)$ which entails that $g(X, W_k)\neq \emptyset$. Moreover, we assume that $W_{k+1}\setminus W_k\subseteq g(A',W_k)$ and it holds that $g(A,W_k)=W_{k+1}\setminus W_k$ since $f(A,k)=\{W_k\}$ and $f(A,k+1)=\{W_{k+1}\}$. Hence, consistency shows that $g(A+A',W_k)=W_{k+1}\setminus W_k$, which implies that $f(A+A',k+1)=\{W_{k+1}\}$. This proves the induction step and thus also the lemma. 
	\end{proof}
	
	Based on \Cref{lem:merge}, we show next that every proper and consistently committee monotone ABC voting rule $f$ satisfies a stronger variant of non-imposition since there is even for every sequence of committees $W_1,\dots, W_\ell$ a profile $A$ such that $f(A, |W_k|)=\{W_k\}$ for all $k\in \{1,\dots, \ell\}$. Even more, we may additionally assume that $f(A, \ell+1)=\{W_\ell\cup \{x\}\colon x\in\mathcal{C}\setminus W_\ell\}$.

	\begin{lemma}\label{lem:seqNI}
		Let $f$ denote a proper ABC voting rule that satisfies consistent committee monotonicity. For every $\ell\in \{1,\dots, m-1\}$ and every sequence of committees $W_1,\dots, W_\ell$, there is a profile $A$ such that $f(A,k)=\{W_k\}$ for all $k\leq \ell$ and $f(A,\ell+1)=\{W_\ell\cup \{x\}\colon x\in\mathcal{C}\setminus W_\ell\}$.
	\end{lemma}
	\begin{proof}
		Let $f$ denote a proper ABC voting rule that satisfies consistent committee monotonicity and let $g$ denote a corresponding consistent generator function. We will inductively show that for every $\ell\in \{1,\dots, m-1\}$ and every sequence of committees $W_1,\dots, W_\ell$ there is a profile $A$ such that $f(A, k)=\{W_k\}$ for $k\in \{1,\dots,\ell\}$ and $f(A, W_{\ell})=\{W_ell\cup \{x\}\colon x\in\mathcal{C}\setminus W_\ell\}$. 
		For the induction basis $\ell=1$, observe that non-imposition shows that for every committee $W=\{x\}\in\mathcal{W}_1$ a profile $A$ such that $f(A, 1)=\{W\}$. Now, let $\tau:\mathcal{C}\rightarrow\mathcal{C}$ denote a permutation such that $\tau(x)=x$. By the neutrality of $f$, it follows that $f(\tau(A), 1)=\{W\}$ for every such permutation. This means that $g(A,\emptyset)=g(\tau(A),\emptyset)=\{x\}$ and, since $g(A+\tau(A),\emptyset)$ cannot be empty, consistency of $g$ entails that $g(A+\tau(A),\emptyset)=\{x\}$. Now, let $A^*$ denote the profile consisting of $\tau(A)$ for every permutation $\tau$ with $\tau(x)=x$. Following the above reasoning, it holds that $g(A^*,\emptyset)=\{x\}$ and thus, $f(A^*, 1)=\{W\}$. On the other hand, all candidates $y\in\mathcal{C}\setminus W$ are completely symmetric and thus, $f(A^*,2)=\{\{x,y\}\colon y\in\mathcal{C}\setminus W\}$, which proves the induction basis.  
		
		Next, we assume that the lemma holds for a fixed $\ell\in \{1,\dots, m-2\}$ and prove it for $\ell+1$. 
		For this, consider an arbitrary committee $W$ of size $\ell+1$. By non-imposition, there is a profile $A$ such that $f(A, \ell+1)=\{W\}$. Moreover, by consistent committee monotonicity, there is an order over the candidates $c_1, \dots, c_{\ell+1}$ in $W$ such that $W_k=\{c_1,\dots, c_k\}\in f(A,k)$ and $c_{k+1}\in g(A,W_k)$ for all $k \in \{1,\dots, \ell\}$. Observe that the committees $W_1, \dots, W_\ell$ form a sequence of profiles of length $\ell$. Thus, the induction hypothesis proves that there is a profile $A'$ such that $f(A', k)=\{W_k\}$ for all $k\in \{1,\dots, \ell\}$ and $f(A', \ell+1)=\{W_\ell\cup \{x\}\colon x\in\mathcal{C}\setminus W_{\ell}\}$. In particular, this means that $g(A', W_k)=\{c_{k+1}\}$ for all $k<\ell$ and $g(A', W_\ell)=\mathcal{C}\setminus W_\ell$. Using \Cref{lem:merge}, we can therefore infer that $f(A+A', k)=\{W_k\}$ for all $k\in \{1,\dots, \ell\}$. Furthermore, $|f(A,\ell+1)|=1$ implies that $g(A, W_\ell)=\{c_{\ell+1}\}$. Thus, the consistency of $g$ proves that $g(A+A', W_{\ell})=g(A, W_\ell)\cap g(A', W_{\ell})=\{c_{\ell+1}\}$. We derive therefore that $f(A'+A, \ell+1)=\{W_{\ell+1}\}$ and hence, $f(A+A', k)=\{W_k\}$ for all $k\in \{1,\dots, \ell+1\}$. 
		
		It remains to construct a profile $B$ such that $f(B, k)=\{W_k\}$ for all $k\in \{1,\dots, \ell+1\}$ and $f(B, \ell+2)=\{W_{\ell+1}\cup \{x\}\colon x\in\mathcal{C}\setminus W_{\ell+1}\}$. For doing so, define $B^\tau=\tau(A+A')$ as the profile derived from $A+A'$ by permuting the candidates according to $\tau$. If $\tau(x)=x$ for all $x\in W_{\ell+1}$, neutrality shows that $f(B^\tau, k)=\{W_k\}$ for all $k\in \{1,\dots, \ell+1\}$. Now define $B$ as the profile that precisely consists of the profiles $B^\tau$ for all permutations $\tau:\mathcal{C}\rightarrow\mathcal{C}$ such that $\tau(x)=x$ for all $x\in W_{\ell+1}$. A repeated application of \Cref{lem:merge} proves that $f(B, k)=\{W_k\}$ for all $k\in \{1,\dots, \ell+1\}$ since all $B^\tau$ agree on these committees. On the other hand, all candidates $c\in\mathcal{C}\setminus W_{\ell+1}$ are completely symmetric in $B$. Hence, neutrality and anonymity show that if $W_{\ell+1}\cup \{x\}\in f(B, \ell+1)$ for some $x\in\mathcal{C}\setminus W_{\ell+1}$, then the same holds for all $x\in\mathcal{C}\setminus W_{\ell+1}$. 
		Finally, consistent committee monotonicity shows that there only such committees can be chosen since $f(B, \ell+1)=\{W_{\ell+1}\}$. Thus, 
		$f(B, \ell+2)=\{W_{\ell+1}\cup \{x\}\colon x\in\mathcal{C}\setminus W_{\ell+1}\}$, which proves the induction step. 
	\end{proof}
	
	Note that \Cref{lem:seqNI} also allows us to construct for every sequence of committees $W_1,\dots, W_m$ (i.e., a sequence with length $m$) a profile $A$ such that $f(A, k)=\{W_k\}$ for all $k\in \{1,\dots,m\}$. The reason for this is that every sequence of length $m-1$ automatically extends to such a sequence since $\mathcal{C}$ is the only committee of size $m$.
	
	Analogous to \Cref{lem:seqNI}, we strengthen next continuity by showing that for all integers $\ell\in \{1,\dots, m\}$ and profiles $A, A'\in\mathcal{A}^*$ such that $|f(A,k)|=1$ for all $k\in \{1,\dots, \ell\}$, there is an integer $j$ such that $f(jA+A',k)=f(jA+A',k)$ for all $k\in \{1,\dots, \ell\}$. Or, more informally, a sufficient majority can enforce the outcome for multiple committee sizes at once.

	\begin{lemma}\label{lem:seqCon}
		Let $f$ denote a proper ABC voting rule that satisfies consistent committee monotonicity. Given two profiles $A$, $A'$ and an integer $\ell\in\{1,\dots, m\}$ such that $|f(A,k)|=1$ for all $k\in \{1,\dots, \ell\}$, there is an integer $j$ such that $f(jA+A',k)=f(A,k)$ for all $k\in \{1,\dots, \ell\}$. 
	\end{lemma}
	\begin{proof}
	Let $f$ denote a proper ABC voting rule that satisfies consistent committee monotonicity, and consider two profiles $A, A'\in\mathcal{A}^*$ and an integer $\ell\in \{1,\dots, m\}$ such that $|f(A,k)|=1$ for all $k\leq \ell$. 
	We will prove the lemma by induction on $\ell$. First, note that the induction basis $\ell=1$ follows immediately from the continuity of $f$, which states that there is an integer $j_1$ such that $f(j_1A+A', 1)=f(A,1)$. Now, assume that the lemma holds up to a fixed $\ell\in \{1,\dots, m-1\}$, i.e., there is an index $j_\ell$ such that $f(A'', k)=f(A,k)$ for all $k\leq \ell$, where $A''=j_\ell A + A'$. Because of continuity, there is an integer $j$ such that $f(j A+A'', \ell+1)=f(A, \ell+1)$. On the other hand, we can use \Cref{lem:merge} to show that $f(jA+A'', k)=f(A'',k)=f(A,k)$ for all $k\leq \ell$ because $f(A, k)=f(A'',k)$ for all $k\leq \ell$. In summary, this means that $f(jA+A'', k)=f(A,k)$ for all $k\leq \ell+1$. Finally, note that $j A+A''=(j+j_\ell)A +A'$, so the integer $j_{\ell+1}=j+j_\ell$ proves the induction step and thus also the lemma. 
	\end{proof}
	
	Note that \Cref{lem:seqNI} and \Cref{lem:seqCon} are important tools for the proofs of most of our results. Next, we will use these insights to show that every consistent and complete generator function of a proper and consistently committee monotone ABC voting rule must be proper itself.

	\begin{lemma}\label{lem:A+N}
		Every complete and consistent generator function of a proper and consistently committee monotone ABC voting rule is proper. 
	\end{lemma}
	\begin{proof}
		Let $f$ denote a proper and consistently committee monotone ABC voting rule and assume that $g$ is a generator function of $f$ that is both complete and consistent. We will show that $g$ satisfies anonymity, neutrality, non-imposition, and continuity and is thus proper.\medskip
		
		\textbf{Claim 1: $g$ is non-imposing.} 
		
		First, we show that $g$ is non-imposing and consider thus a committee $W$ with $|W|<m$ and a candidate $x\not\in W$. Consider an arbitrary order of the candidates $c_1, \dots, c_\ell\in W$ and define $W_k=\{c_1, \dots, c_k\}$ for all $k\leq \ell$ and $W_{\ell+1}=W_\ell\cup \{x\}$. In particular, $W=W_\ell$ and $W_{\ell+1}=W\cup \{x\}$. By \Cref{lem:seqNI}, there is a profile $A$ such that $f(A,k)=W_k$ for all $k\leq \ell+1$. This entails that $g(A, W)=\{x\}$ and thus proves that $g$ is non-imposing.\medskip
		
		\textbf{Claim 2: $g$ is anonymous.}
		
		Assume for contradiction that $g$ fails anonymity, which means that there is a profile $A$, a committee $W\neq \mathcal{C}$, and a permutation $\pi:\mathbb{N}\rightarrow\mathbb{N}$ such that $g(A,W)\neq g(\pi(A),W)$. As first step, we consider an arbitrary order of the candidates $c_1, \dots, c_\ell\in W$ and define $W_k=\{c_1, \dots, c_k\}$ for all $\in \{1,\dots, \ell\}$. In particular, $W=W_\ell$. By \Cref{lem:seqNI}, there is a profile $A'$ such that $f(A',k)=\{W_k\}$ for all $k\in \{1,\dots, \ell\}$ and $f(A', \ell+1)=\{W_\ell\cup\{x\}\colon x\in\mathcal{C}\setminus W_{\ell}\}$. The last point means for $g$ that $g(A', W_\ell)=\mathcal{C}\setminus W_\ell$. 
		
		Furthermore, by \Cref{lem:seqCon}, there is an integer $j$ such that $f(jA'+A,k)=\{W_k\}$ for all $k\in \{1,\dots, \ell\}$. Thus, the completeness and consistency of $g$ imply that $g(jA'+A, W_\ell)=g(jA',W_\ell)\cap g(A, W_\ell)=g(A, W_\ell)$. Now, since $g$ generates $f$, we infer therefore that $f(jA'+A, \ell+1)=\{W_\ell\cup x\colon x\in g(A, W_\ell)\}$. 
		
		Finally, consider the profile $\pi(jA'+A)$. By the anonymity of $f$, it follows that $f(\pi(jA'+A),k)=f(\pi(jA'), k)=\{W_k\}$ for all $k\in \{1,\dots, \ell\}$ and $f(\pi(jA'), \ell+1)=\{W_\ell\cup\{x\}\colon x\in\mathcal{C}\setminus W\}$. This entails again that $g(\pi(jA'), W_\ell)=\mathcal{C}\setminus W_\ell$. Hence, an analogous reasoning as for $jA'+A$ shows that $g(\pi(jA'+A), W_\ell)=g(\pi(jA'), W_\ell)\cap g(\pi(A), W_\ell)=g(\pi(A), W_\ell)$. This implies that $f(\pi(jA'+A), \ell+1)=\{W_\ell\cup x\colon x\in g(\pi(A), W_\ell)\}$. However, we have by assumption that $g(A, W_\ell)\neq g(\pi(A), W_\ell)$ and thus $f(jA'+A, \ell+1)\neq f(\pi(jA'+A), \ell+1)$. This conflicts with the anonymity of $f$ and shows therefore that the assumption that $g$ fails anonymity is false.\medskip
		
		\textbf{Claim 3: $g$ is neutral}
		
		We assume for contradiction that $g$ is not neutral, which means that there is a profile $A$, a committee $W\neq\mathcal{C}$, and a permutation $\tau:\mathcal{C}\rightarrow\mathcal{C}$ such that $g(\tau(A), \tau(W))\neq \tau(g(A,W))$. Analogous to the last claim, let $c_1, \dots, c_\ell$ denote the candidates in $W$ and define $W_k=\{c_1,\dots, c_k\}$. Thus, $W=W_\ell$. By \Cref{lem:seqNI}, there is a profile $A'$ such that $f(A',k)=\{W_k\}$ for all $k\in \{1,\dots, \ell\}$ and $f(A', \ell+1)=\{W_\ell\cup\{x\}\colon x\in\mathcal{C}\setminus W_{\ell}\}$. This means again that $g(A', W_{\ell})=\mathcal{C}\setminus W_k$. A completely analogous reasoning as in Claim 2 shows now that there is an integer $j$ such that $f(jA'+A, k)=\{W_k\}$ for all $k\in \{1,\dots, \ell\}$ and $f(jA'+A,\ell+1)=\{W_{\ell}\cup x\colon x\in g(A,W)\}$.
		
		Next, consider the profile $\tau(jA'+A)$. 
		By the neutrality of $f$, we have that $f(\tau(jA'+A), k)=f(\tau(jA',k)=\{\tau(W_k)\}$ for all $k\in \{1,\dots, \ell\}$. 
		Moreover, this axiom also shows that $f(\tau(jA'), \ell+1)=\{\tau(W_\ell\cup \{x\})\colon x\in\mathcal{C}\setminus W_\ell\}$. 
		This implies that $g(\tau(jA'), \tau(W_\ell))=\mathcal{C}\setminus \tau(W_\ell)=\tau(\mathcal{C}\setminus W_\ell)$. 
		In turn, we infer from consistency and completeness that $g(\tau(jA'+A), \tau(W_\ell))=g(\tau(jA'), \tau(W_\ell))\cap g(\tau(A), \tau(W_\ell))=g(\tau(A), \tau(W_\ell))$. 
		Since $g$ generates $f$, this means that $f(\tau(jA'+A),k)=\{\tau(W_\ell)\cup \{x\}\colon x\in g(\tau(A), \tau(W_\ell))\}$. 
		However, since $g(\tau(A), \tau(W_\ell))\neq \tau(g(A, W_\ell))$, this means that $f(\tau(jA'+A),k)=\{\tau(W_\ell)\cup \{x\}\in g(\tau(A), \tau(W_\ell))\}\neq \{\tau(W_\ell\cup\{x\})\colon x\in g(A, W_\ell)\}=\tau(f(jA'+A, \ell+1))$.
		This contradicts that $f$ satisfies neutrality and thus, the assumption that $g$ is not neutral must be wrong.\medskip
		
		\textbf{Claim 4: $g$ is continuous.}
		
		Finally, we show that $g$ is continuous and consider therefore two profiles $A$ and $A'$ and a committee $W$ such that $g(A,W)=\{c\}$ for some candidate $c\in\mathcal{C}\setminus W$. We need to show that there is an integer $j$ such that $g(jA+A', W)=g(A,W)$. For doing so, let $c_1, \dots, c_\ell$ denote the candidates in $W$ and define $W_k=\{c_1,\dots, c_k\}$ for all $k\in \{1,\dots, \ell\}$. Using again \Cref{lem:seqNI}, there is a profile $A''$ such that $f(A'', k)=\{W_k\}$ for all $k\in \{1,\dots, \ell\}$ and $f(A'', \ell+1)=\{W_{\ell}\cup x\colon x\in \mathcal{C}\setminus W_{\ell}\}$. By \Cref{lem:seqCon}, we can find an integer $j$ such that $f(jA''+A,k)=f(A'',k)$ for all $k\in \{1,\dots, \ell\}$. On the other hand, we infer that $g(jA''+A, W_\ell)=g(jA'', W_{\ell})\cap g(A, W_{\ell})=g(A, W_\ell)=\{c\}$ due to the consistency and completeness of $g$. Since $g$ generates $f$, this means that $f(jA''+A, \ell+1)=\{W_{\ell}\cup\{c\}\}$.
		
		Now, using again \Cref{lem:seqCon}, we can find another integer $j'$ such that $B=j'(jA''+A)+A'$ and $f(B, k)=f(jA''+A, k)$ for all $k\in \{1,\dots, \ell+1\}$. In particular, this implies that $g(B, W_\ell)=g(A, W_\ell)=\{c\}$. We prove that $g$ is continuous by showing that $g(B, W_\ell)=g(j'A+A', W_{\ell})$. For doing so, let $j''=j\cdot j'$. It clearly holds that $g(j'' A'', W_{\ell})=\mathcal{C}\setminus W_{\ell}$ and therefore, $g(j''A''+(j'A+A'), W_{\ell})=g(j''A'', W_{\ell})\cap g(j'A+A', W_{\ell})=g(j'A+A', W_{\ell})$. Hence, $g(j'A+A',W_{\ell})=g(A,W_{\ell})$, which proves that $g$ is continuous.
	\end{proof}
	
	Finally, we have all ingredients to prove \Cref{prop:niceGenerator}.
	
	\niceGenerator*
	\begin{proof}
		We show both directions independently from each other.\medskip
		
		\textbf{Claim 1: If $f$ is generated by a proper, consistent, and complete generator function $g$, it is a proper and consistently committee monotone ABC voting rule.}
		
		Let $g$ denote a proper, consistent, and complete generator function $g$, and let $f$ denote the ABC voting rule generated by $g$. Note that the completeness of $g$ ensures that it indeed induces an ABC voting rule. Now, $f$ is by definition consistently committee monotone since $g$ is consistent. Hence, it remains to show that $f$ is anonymous, neutral, non-imposing, and continuous.\smallskip
		
		\emph{Anonymity:} First, we show that $f$ is anonymous. For doing so, consider an arbitrary profile $A$ and a permutation $\pi:\mathbb{N}\rightarrow\mathbb{N}$. We will show by an induction on the committee size $k$ that $f(A,k)=f(\pi(A),k)$ for all $k\in \{1,\dots, m\}$. Hence, note that the induction basis $k=0$ is trivial since $f(A,0)=\{\emptyset\}=f(\pi(A),0)$ by definition. Next, fix some $k\in \{0,\dots, m-1\}$ and suppose that $f(A,k)=f(\pi(A),k)$. Since $g$ is anonymous, we have that $g(A,W)=g(\pi(A), W)$ for all $W\in f(A,k)=f(\pi(A),k)$, which implies that $f(A, k+1)=f(\pi(A), k+1)$. This proves the induction step and thus shows that $f$ is anonymous.\smallskip
		
		\emph{Neutrality:} Our next goal is to show that $f$ is neutral. Hence, consider again an arbitrary profile $A$ and a permutation $\tau:\mathcal{C}\rightarrow\mathcal{C}$. we use again an induction on the committee size $k$ to show that $f(\tau(A),k)=\tau(f(A,k))$. The induction basis $k=0$ is trivial because $f(A,0)=\{\emptyset\}=f(\tau(A), 0)$ by definition. Hence, assume that there is fixed $k\in \{0,\dots, m-1\}$ such that $f(\tau(A), k)=\tau(f(A,k))$. Since $g$ is neutral, it follows that $g(\tau(A), \tau(W))=\tau(g(A,W))$ for all $W\in f(A,k)$. This means that $f(\tau(A), k+1)=\{W\cup \{x\}\colon W\in f(\tau(A),k), x\in g(\tau(A), W)\}=\{\tau(W)\cup \{x\}\colon W\in f(A,k), x\in g(\tau(A),\tau(W))\}=\{\tau(W\cup \{x\})\colon W\in f(A,k), x\in g(A,W)\}=\tau(f(A,k+1))$, which proves that $f$ is neutral.\smallskip
		
		\emph{Continuity:} As third point, we show that $f$ is continuous. For doing so, assume that for every two profiles $A$ and $A'$ and every committee $W$, there is an integer $j$ such that $g(jA+A',W)\subseteq g(A,W)$; we will prove that this claim is correct later on. Next, consider two profiles $A$ and $A'$ and a committee size $k$ such that $|f(A, k)|=1$. Our goal is to show that there is an integer $j$ such that $f(jA+A',k)=f(A,k)$. 
		
		For doing so, let $F=\bigcup_{\ell=0}^{k-1} f(A,\ell)$ denote the set of all committees of size at most $k-1$ that $f$ chooses for $A$. By our auxiliary claim, there is for every $W\in F$ a integer $j_W$ such that $g(j_WA+A',W)\subseteq g(A,W)$. Now, let $j^*$ denote the maximum among these integers and note that consistency and completeness imply that $g(j^*A+A',W)=g(j_W A+A',W)\cap g((j^*-j_W) A, W)\subseteq g(A,W)$ for all committees $W\in F$ with $j_W<j^*$. This means that $g(j^*A+A',W)\subseteq g(A,W)$ for all $W\in F$, which clearly implies that $f(j^*A+A',\ell)\subseteq f(A,\ell)$ for all $\ell\in \{1,\dots, k\}$. Since $f$ is an ABC voting rule, $f(j^*A+A',k)\neq \emptyset$. Thus, $|f(A,k)|=1$ entails $f(j^*A+A',k)=f(A,k)$, which proves that $f$ is continuous. 
		
		It remains to show our auxiliary claim that for all profiles $A$ and $A'$ and all committees $W$, there is an integer $j$ such that $g(jA+A', W)\subseteq g(A,W)$. Assume for contradiction that this is not the case, which means that there are profiles $A$, $A'$ and a committee $W$ such that $g(jA+A',W)\not\subseteq g(A,W)$ for all $j\in\mathbb{N}$. First, note that this requires that $|g(A,W)|\geq 2$ as otherwise, this assumption directly contradicts the continuity of $g$. Next, let $I_c=\{j\in\mathbb{N}\colon c\in g(jA+A',W)\}$ denote the set of integers $j$ such that $g$ chooses $c$ for $jA+A'$. It must hold that $I_c=\emptyset$ for all $c\in g(A,W)$: if there is $j\in I_c$, consistency implies for $g((j+1)A+A', W)$ that $g((j+1)A+A,W)=g(A,W)\cap g(jA+A',W)\subseteq g(A,W)$ since $g(A,W)\cap g(jA+A',W)\neq \emptyset$. This, however, contradicts our assumptions. 
		
		Next, let $c\in \mathcal{C}\setminus g(A,W)$ denote a candidate such that $I_c$ contains infinitely many elements; such a candidate must exist since there is only a finite number of candidates but $\mathbb{N}$ is infinite. Moreover, let $a$ denote a candidate in $g(A,W)$. For each other candidate $x\in g(A,W)\setminus \{a\}$, we define $\tau^x$ as the permutation that maps $x$ to $c$, $c$ to $x$, and every other candidate to itself. By neutrality, we have that $c\in g(\tau^x(A),\tau^x(W))$, $x\not\in g(\tau^x(A),\tau^x(W))$. Now, note that $\tau^x(W)=W$ since $x\in g(A,W)$ and $c\in g(jA+A',W)$ for $j\in I_c$ imply that $c,x\not\in W$. Moreover, consistency shows that $c\in g(\tau^x(jA), W)$, $x\not\in g(\tau^x(jA),W)$ for all integers $j\in\mathbb{N}$. 
		
		Finally, consider the profiles $B^j$ which consists of $j$ copies of $A$, $j$ copies of $\tau^x(A)$ for every $x\in g(A,W)\setminus \{a\}$, and one copy of $A'$. Now, since $c\in g(jA+A',W)$ for all $j\in I_c$ and $c\in g(\tau^x(jA), W)$ for all $j\in \mathbb{N}$ and $x\in g(A,W)\setminus \{a\}$, consistency shows that $g(B^j, W)=g(jA+A',W)\cap \bigcap_{x\in g(A,W)\setminus \{a\}} g(\tau^x(jA), W)=\{c\}$ for $j\in I_c$. In particular, the last equality holds since $g(jA,W)\cap g(jA+A',W)=\emptyset$ and $\tau^x(A)$ only swaps $c$ and $x$. 
		
		On the other side, it holds for the profile $B$, which consists only of a single copy of $A$ and one copy of $\tau^x(A)$ for every $x\in g(A,W)\setminus \{a\}$, that $g(B,W)=\{a\}$ because of consistency. Hence, continuity implies that there is an index $j^*$ such that $g(j^*B+A', W)=\{a\}$. However, $I_c$ is infinite and there is thus an index $j\in I_c$ with $j>j^*$ such that $c\in g(B^j,W)=g(j B+A',W)$. This contradicts consistency: $g(j B+A', W)=g((j-j^*) B, W)\cap g(j^* B+A', W)=\{a\}$. Hence the assumption that there is no integer $j$ such that $g(jA+A', W)\subseteq g(A,W)$ is wrong and the claim therefore proven.\smallskip
		
		\emph{Non-imposition:} Finally, we show that $f$ is non-imposing. Hence, consider an arbitrary committee $W=\{c_1,\dots,c_\ell\}$. We need to show that there is a profile $A$ such that $f(A, |W|)=\{W\}$. For this, let $W_k=\{c_1,\dots, c_k\}$ for all $k \in \{1,\dots,\ell\}$. We will inductively construct a profile $A$ such that $f(A,k)=\{W_k\}$ for all $k \in \{1,\dots,\ell\}$. The base case $k=1$ follows immediately from the non-imposition of $g$ because this axiom implies that there is a profile $A$ such that $g(A,\emptyset)=\{c_1\}$. 
		
		Next, consider a fixed $k'\in \{1,\dots, \ell-1\}$ and assume that we have a profile $A$ such that $f(A,k)=\{A_{k}\}$ for all integers $k\in \{1,\dots,k'\}$. In particular, this means that $g(A, W_{k-1})=\{c_{k}\}$ for all $k\in \{1,\dots, k'\}$ (where $W_0=\emptyset$). Furthermore, let $\tau$ denote a permutation on the candidates such that $\tau(x)=x$ for all $x\in W_{k'}$. Neutrality implies that $g(\tau(A), \tau(W_{k-1}))=g(\tau(A), W_{k-1})=\{c_k\}$ for all $k \in \{1,\dots,k'\}$. In turn, consistency and completeness prove then that $g(A+\tau(A), W_{k-1})=\{c_k\}$ for all such permutations $\tau$ and $k\in \{1,\dots,k'\}$. Hence, consider now the profile $A'$ that consists of a copy of $\tau(A)$ for every permutation $\tau:\mathcal{C}\rightarrow\mathcal{C}$ such that $\tau(x)=x$ for $x\in W_{k'}$. By repeating the previous arguments, we infer that $g(A', W_{k-1})=\{c_k\}$ for all $k \in \{1,\dots,k'\}$. Moreover, the profile $A'$ is completely symmetric with respect to the candidates $c\in \mathcal{C}\setminus W_{k'}$. Hence, neutrality and anonymity require that $g(A', W_{k'})=\mathcal{C}\setminus W_{k'}$. 
		
		Finally, since $g$ is non-imposing, there is a profile $A''$ such that $g(A'',W_{k'})=\{c_{k'+1}\}$. Moreover, by using the same auxiliary claim as for our analysis on continuity, we can find an integer $j$ such that $g(jA'+A'', W_{k-1})\subseteq g(A', W_{k-1})$ for all $k\in \{1,\dots, k'\}$. Since $g$ is complete and $|g(A',W_{k-1})|=1$ for all these $k$, this implies $g(jA'+A'', W_{k-1})=\{c_k\}$ for all $k\in \{1,\dots, k'\}$. On the other hand, completeness and consistency require that $g(jA'+A'', W_{k'})=g(jA',W_{k'})\cap g(A'',W_{k'})=\{c_{k'+1}\}$. This entails for $f$ that $f(jA'+A'', k)=\{W_k\}$ for all $k\in \{1,\dots, k'+1\}$ and thus proves the induction step. Hence, $f$ is indeed non-imposing.\medskip
		
		\textbf{Claim 2: If $f$ is a proper and consistently committee monotone ABC voting rule, it is generated by a proper, consistent, and complete generator function.}
		
		Let $f$ denote a proper ABC voting rule that satisfies consistent committee monotonicity. Thus, $f$ can be generated by a consistent generator function $g$. We only need to extend $g$ to a complete generator function to prove this claim since \Cref{lem:A+N} then shows that it is also proper. For doing so, we will define a second generator function $\hat g(A,W)$ for $f$. To this end, consider a single committee $W\neq\mathcal{C}$ and let $W_1,\dots, W_\ell$ denote a sequence of committees such that $W_\ell=W$. Moreover, let $A_W$ denote a profile such that $f(A_W,k)=\{W_k\}$ for all $k\in \{1,\dots, \ell\}$ and $f(A_W,\ell+1)=\{W\cup \{x\}\colon x\in\mathcal{C}\setminus W\}$; such a profile exists because of \Cref{lem:seqNI}. Then, we define $\hat g(A,W)=g(jA_W+A, W)$ for all profiles $A$ and committees $W$, where $j$ denotes the smallest integers such that $f(jA_W+A,k)=\{W_k\}$ for all $k\in \{1,\dots,|W|\}$; such an integer exists because of \Cref{lem:seqCon}. In particular, observe that this condition ensures that $g(jA_W+A,W)\neq \emptyset$ and thus, $\hat g$ is indeed complete. In the remainder of this proof, we will show that $\hat g$ generates $f$ and is consistent because \Cref{lem:A+N} then shows that it satisfies all our requirements.
				
		First, we show that $\hat g$ generates $f$. For this, we will show that for all profiles $A$ and committees $W$ with $W\in f(A,|W|)$, it holds that $g(A,W)=
		\hat g(A,W)=g(jA_W+A, W)$. By our definition, we have that $W\in f(A_W,|W|)$ and $W\in f(jA_W+A, |W|)$, and we assume that $W\in f(A, |W|)$. Moreover, a repeated application of consistent committee monotonicity shows that $f(A_W,|W|)=f(jA_W,|W|)$ and we hence derive that $g(X, W)\neq \emptyset$ for all $X\in\{jA_W, A, jA_W+A\}$. Finally, our definition of $A_W$ also implies that $g(jA_W, W)=\mathcal{C}\setminus W$ and thus, consistency entails that $\hat g(A,W)=g(jA_W+A, W)=g(jA_W,W)\cap g(A,W)=g(A,W)$. Hence, $\hat g$ and $g$ coincide for all profiles $A$ and committees $W$ with $W\in f(A, |W|)$, which implies that $\hat g$ generates $f$.
			
	Finally, we show that $\hat g$ is consistent. Thus, consider a fixed committee $W\neq\mathcal{C}$ and let $W_1,\dots,W_\ell$ denote the sequence of profiles used for defining $\hat g(A,W)$. In particular, we have that $W_\ell=W$. Moreover, let $A_W$ denote the profile such that $\hat g(A,W)=g(jA_W+A, W)$ for some $j\in\mathbb{N}$. Recall that $f(A_W,k')=\{W_{k}\}$ for all $k\in \{1,\dots, \ell\}$ and $f(A_W,\ell+1)=\{W_\ell\cup \{x\}\colon x\in\mathcal{C}\setminus W_\ell\}$. Finally, for proving that $\hat g$ is consistent, consider two profiles $A$, $A'$ such that $\hat g(A,W)\cap \hat g(A',W)\neq \emptyset$, and define $j$, $j'$, and $j''$ as the smallest integers such that $f(j A_W +A, k)=f(j' A_W +A', k)=f(j'' A_W +(A+A'), k)=\{W_k\}$ for all $k\in \{1,\dots, |W|\}$. By definition of $\hat g$, $\hat g(A,W)= g(jA_W+A,W)$, $\hat g(A',W)=g(j'A_W+A,W)$, and $\hat g(A+A',W)=g(j''A_W+(A+A'), W)$. 
	We proceed with a case distinction with respect to the relation between $j+j'$ and $j''$. 
	\begin{itemize}
		\item First, assume that $j+j'=j''$. Then, the consistency of $g$ implies that $\hat g(A+A',W)=g(j'' A_W+A+A',W)=g(jA_W+A,W)\cap g(j'A_W+A',W)=\hat g(A,W)\cap \hat g(A',W)$. Hence, consistency holds. 
		\item As second case, assume that $j+j'>j''$ and define $\ell=j+j'-j''$. First, we note that $f(\ell A_W,k)=\{W_k\}$ for all $k\in \{1,\dots, |W|\}$ and $g(\ell A_W, W)=\mathcal{C}\setminus W$ because of the consistency of $g$. Hence, we can use \Cref{lem:merge} to deduce that $f((j+j')A_W+A+A',k)=f(\ell A_W + j''A_W + A + A', k)=\{W_k\}$ for all $k\in \{1,\dots, |W|\}$. Consistency therefore implies that $g(j''A_W+A+A',W)=g(j''A_W+A+A',W)\cap g(\ell A_W,W)=g((j+j')A_W+A+A',W)$. Moreover, analogous to the above case, we have that $g((j+j')A_W+A,W)=g(jA_W+A)\cap g(j'A_W+A',W)$. In summary, we thus have that $\hat g(A+A',W)=\hat g(A,W)\cap \hat g(A',W)$.
		\item As last case, note that $j+j'<j''$ is not possible since we choose the integer $j''$ minimal. In particular, it follows from \Cref{lem:merge} that $f(j A_W +A + j' A_W +A', k)=\{W_k\}$ for all $k\in \{1,\dots, |W|\}$. This contradicts the minimality of $j''$ if $j+j'<j''$ and thus, this case cannot happen.
	\end{itemize}
	\end{proof}
	
	\subsection{Proof of \Cref{prop:hyperplane}}

In this section, we show \Cref{prop:hyperplane} which states that for a fixed committee $W\neq\mathcal{C}$, every proper, consistent, and complete generator function $g$ is equivalent to $v$-weighted approval voting for some weight function $v$. Note that this claim is trivial if $|W|=m-1$ because $g(A,W)=\mathcal{C}\setminus W$ for every complete generator function as only a single candidate remains. Thus, it holds for every committee $W$ of size $m-1$ that $g(A,W)=AV_{v^W}(A,W)$, where $v^W(x,y)=1$ for all $x,y$. As a consequence, we will focus in the majority of this section on the case that $|W|\leq m-2$. 

For proving \Cref{prop:hyperplane}, we use the fact that, for a fixed committee $W$, $g(\cdot,W)$ is closely connected to single winner voting rules, so-called social choice functions (SCFs). In particular, $g(\cdot, W)$ can be seen as an SCF on the candidates $\mathcal{B}=\mathcal{C}\setminus W$. To simplify our analysis, we will only consider SCFs with restricted domain: we require that all voters approve the same number of candidates in $\mathcal{B}$. To make this formal, let $\mathcal{B}_k$ denote the size $k$ subsets of the set of available candidates $\mathcal{B}$. Moreover, let $\mathcal{B}_k^*$ denote the set of approval profiles in which all voters report a preference in $\mathcal{B}_k$. Finally, a social choice function (SCF) on the domain $\mathcal{B}_k^*$ maps every profile $A\in \mathcal{B}_k^*$ to a non-empty subset of $\mathcal{B}$, which is interpreted as the set of winners of the elections. 

Three SCFs will turn out to be particular important for our analysis: approval voting ($\mathit{AV}$) chooses the candidates that are approved by the most voters, anti-approval voting ($\overline{\mathit{AV}}$) chooses the candidates that are approved by the least voters, and the trivial rule ($\mathit{TRIV}$) always chooses all candidates. Furthermore, analogous to ABC voting rules, we introduced axioms for SCFs: anonymity requires that permuting the voters does not affect the outcome ($f(A)=f(\pi(A))$), neutrality that permuting the candidates in the profile permutes them correspondingly in the choice set ($f(\tau(A))=\tau(f(A))$), and consistency that if the intersection of the choice sets for two disjoint profiles is non-empty, then precisely this intersection is chosen when considering both profiles combined ($f(A)\cap f(A')\neq\emptyset$ implies that $f(A+A')=f(A)\cap f(A')$). As we show next, these three axioms characterize $\mathit{AV}$, $\overline{\mathit{AV}}$, and $\mathit{TRIV}$. 

\begin{lemma}\label{lem:AV}
    Let $\mathcal{B}=\mathcal{C}\setminus W$ denote the non-empty set of feasible candidates and let $k\in \{1,\dots, |\mathcal{B}|\}$. $\mathit{AV}$, $\overline{\mathit{AV}}$, and $\mathit{TRIV}$ are the only SCFs on $\mathcal{B}_k^*$ that satisfy anonymity, neutrality, and consistency. 
\end{lemma}
\begin{proof}
    Fix a set $W\subsetneq\mathcal{C}$ and let $\mathcal{B}=\mathcal{C}\setminus W$. Moreover, we define $m=|\mathcal{B}|$ for this proof and let $k\in\{1,\dots,|\mathcal{C}\setminus W|\}$.
	It is straightforward to check that $\mathit{AV}$, $\overline{\mathit{AV}}$, and $\mathit{TRIV}$ are anonymous, neutral, and consistent for $\mathcal{B}_k^*$ and we thus focus on the converse direction. Hence, let $f$ denote an anonymous, neutral, and consistent SCF on $\mathcal{B}_k^*$. First, note that if $m=1$, then every rule has to always return the single available candidate. Similar, if $k=m$, every anonymous and neutral SCF on $\binom{\mathcal{C}\setminus W}{k}^*$ must always return all candidates because all voters have to approve all candidates. Hence, $f$ coincides with $\mathit{AV}$ (as well as $\overline{\mathit{AV}}$ and $\mathit{TRIV}$) for all profiles in this case and we therefore assume that $1\leq k<m$. 
	
	For proving that $f$ is one of our three SCFs, we will consider profiles consisting of a single ballot $A_i$. First, if $f(A_i)=\mathcal{C}$ for some $A_i\in \mathcal{B}_k$, then neutrality implies that $f(A_i')=\mathcal{C}$ for all $A_i'$ since $|A_i|=|A_i'|$ for all ballots in our domain. Consistency then shows that $f$ is $\mathit{TRIV}$. Because of neutrality, there are two possible cases left: $f(A_i)=A_i$ or $f(A_i)=\mathcal{B}\setminus A_i$ for all $A_i\in\mathcal{B}_k$. We will show in Step 2 that the former implies that $f$ is $\mathit{AV}$ and in Step 3 that the latter implies that $f$ is $\overline{\mathit{AV}}$. Before proving these claims, we will discuss an auxiliary statement showing that if two candidates $a,b$ have the same approval score $s(A,x)=|\{i\in N_A\colon x\in A_i\}|$, then both are either chosen or unchosen in $f(A)$.\medskip
	
	\textbf{Step 1: If $s(A,a)=s(A,b)$, then $a\in f(A)$ if and only if $ b\in f(A)$.}
	
	Consider a profile $A$ and two candidates $a,b\in\mathcal{B}$ with $s(A,a)=s(A,b)$. First, if $s(A,a)=s(A,b)=0$, no voter approves $a$ or $b$. Neutrality then shows that $a\in f(A)$ if and only if $b\in f(A)$ as permuting $a$ and $b$ leads to the same profile $A$. Hence, we focus on the case that $s(A,a)=s(A,b)>0$ and suppose for contradiction that $a\in f(A)$, $b\not\in f(A)$. Now, note that for every permutation $\tau:\mathcal{B}\rightarrow\mathcal{B}$ with $\tau(a)=a$ and $\tau(b)=b$, neutrality implies $a\in f(\tau(A))$, $b\not\in f(\tau(A))$. Next, define $A^*$ as the profile that consists of the profiles $\tau(A)$ for every permutation $\tau:\mathcal{B}\rightarrow\mathcal{B}$ such that $\tau(a)=a$ and $\tau(b)=b$. Neutrality requires for every subprofile $\tau(A)$ that $a\in f(\tau(A))$, $b\not\in f(\tau(A))$ and, in turn, consistency shows that $a\in f(A^*)$, $b\not\in f(A^*)$. 
	
	Subsequently, we consider a different decomposition of $A^*$. Note for this that for every two ballots $A_i$, $A_i'$ with $a\in A_i$, $a\in A_i'$, $b\not\in A_i$, $b\not\in A_i'$, there are precisely $(k-1)! (m-k-1)!$ permutations $\tau$ with $\tau(a)=a$, $\tau(b)=b$ that map $A_i$ to $A_i'$. An analogous statement also holds when exchanging the roles of $a$ and $b$. Hence, we can partition $A^*$ in the following profiles: in $A^1$ all voters approve both $a$ and $b$, in $A^2$ no voter approves either $a$ or $b$, and for every set of candidates $X\subseteq \mathcal{B}\setminus \{a,b\}$ with $|X|=k-1$, there is a profile $A^X$ in which all voters either report $\{a\}\cup X$ or $\{b\}\cup X$. In particular, since $s(A,a)=s(A,b)$, both $a$ and $b$ are approved by the same number of voters in $A^X$. Finally, note that $a$ and $b$ are completely symmetric in all subprofiles of $A^*$ and thus also in $A^*$ itself. Hence, anonymity and neutrality require that $a\in f(A^*)$ if and only if $b\in A^*$. However, this contradicts our previous observation $a\in f(A^*)$, $b\not\in f(A^*)$, and thus our initial assumption that $a\in f(A)$, $b\not\in f(A)$ must be wrong. Consequently, we infer that $s(A,a)=s(A,b)$ implies that $a\in f(A)$ if and only if $ b\in f(A)$.\medskip
	
	\textbf{Step 2: If $f(A_i)=A_i$ for all $A_i\in\mathcal{B}_k^*$, then $f(A)=\mathit{AV}(A)$ for all $A\in\mathcal{B}^*_k$.}
	
	The claim follows by showing that $f$ cannot choose candidates that have less than maximal approval score. Due to the non-emptiness and Step 1, $f$ then must choose the approval winners. Hence, assume for contradiction that there is a profile $A\in \mathcal{B}^*_k$ and two candidates $a,b\in\mathcal{B}$ such that $b\in f(A)$ but $s(A,a)>s(A,b)$. Next, let $A'$ denote a profile consisting of $s(A,a)-s(A,b)$ voters who approve $b$ but not $a$. Since $b\in f(A_i')=A_i'$ and $a\not\in f(A_i')$ for all voters $i\in N_{A'}$, consistency entails that $b\in f(A')$, $a\not\in f(A')$. Another application of consistency then shows that $b\in f(A+A')$ and $a\not \in f(A+A')$. However, $s(A+A',a)=s(A+A',b)$ and Step 1 thus requires that $a\in f(A+A')$ if and only if $b\in f(A+A')$. Hence, our initial assumption is wrong and $b\in f(A)$ is only possible if there is no candidate $a$ with $s(A,a)>s(A,b)$. Equivalently, this means that $f(A)=\mathit{AV}(A)$.\medskip
	
	\textbf{Step 3: If $f(A_i)=\mathcal{B}\setminus A_i$ for all $A_i\in\mathcal{B}_k^*$, then $f(A)=\overline{\mathit{AV}(A)}$ for all $A\in\mathcal{B}^*_k$.}
	
	Analogous to the last case, we will show that $f$ cannot choose candidates that have above minimal approval score. The non-emptiness of $f$ and Step 1 then show again that $f$ chooses the anti-approval winners. Thus, assume for contradiction that there is a profile $A\in \mathcal{B}^*_k$ and two candidates $a,b$ such that $b\in f(A)$ and $s(A,a)<s(A,b)$. Similar to the last case, let $A'$ denote a profile consisting of $s(A,b)-s(A,a)$ such that $a\in A_i'$ and $b\not\in A_i'$ for all $i\in N_{A'}$. Consistency and our assumption that $f(A_i')=\mathcal{B}\setminus A_i'$ show that $b\in f(A')$, $a\not\in f(A')$. Hence, another application of consistency shows that $b\in f(A+A')$, $a\not\in f(A+A')$. However, this contradicts Step 1 since $s(A+A', a)=s(A+A',b)$. This shows that the initial assumption must be wrong and $f$ indeed can only choose candidates with minimal approval scores. In other words, this means that $f(A)=\overline{\mathit{AV}}(A)$ for all profiles $A\in \mathcal{B}^*_k$.	
	\end{proof}
	
	As next step, we will start to analyze $g(\cdot, W)$ for a fixed committee $W$. In particular, we will show that every proper generator function that is complete and consistent induces an SCF for the candidates $\mathcal{C}\setminus W$ that satisfies all requirements of \Cref{lem:AV} when we assume that all voters $i\in N_A$ approve the same subset of $W$ and the same number of candidates. Based on this insight, we will infer the behavior of $g(A,W)$ for sufficiently symmetric profiles. For formalizing the latter, let $n(c,A,W,k,\ell)=|\{i\in N_A: c\in A_i\land |A_i\cap W|=k \land |A_i|=\ell\}|$ denote the number of voters who approve $c\in\mathcal{C}\setminus W$, $k$ candidates of $W$, and $\ell$ candidates in total. Note that by definition $n(c,A,W,k,\ell)=0$ if $k\geq \ell$ or $\ell>m-(|W|-k)$ since there is no ballot such in which a voter approves $c$ and the right number of candidates of $W$ and $\mathcal{C}$.
	
	\begin{lemma}\label{lem:0element}
		Let $g$ denote a proper, complete, and consistent generator function. It holds that  $g(A,W)=\mathcal{C}\setminus W$ for every profile $A$ and committee $W$ with $|W|\leq m-2$ for which there are constants $c_{k,\ell}$ such that $n(x,A,W,k,\ell)=c_{k,\ell}$ for all $x\in \mathcal{C}\setminus W$, $k\in \{0,\dots, |W|\}$, and $\ell\in \{k+1, \cdots, m-1-|W|+k\}$. 
	\end{lemma}
	\begin{proof}
		Consider a proper generator function $g$ that is complete and consistent and fix an arbitrary committee $W$ with $|W|\leq m-2$. We will prove the lemma in multiple steps by first focusing on more restricted profiles.\medskip
		
		\textbf{Step 1:} As first step, we will show that $g(A,W)=\mathcal{C}\setminus W$ for all profiles $A$ for which there is a subset $X$ of $W$ and constants $\hat c\in \mathbb{N}$, $\ell\in \{|X|+1,\dots, m\}$ such that $|A_i|=\ell$ and $W\cap A_i=X$ for all voters $i\in N_A$, and $n(x,A,W,|X|, \ell)=\hat c$ for all candidates $x\in \mathcal{C}\setminus W$. To this end, let $\mathcal{B}=\mathcal{C}\setminus W$ denote the set of still available candidates and let $\mathcal{B}_{\ell-|X|}$ denote the set of size $\ell-|X|$ subsets of $\mathcal{C}'$. Moreover, we define $\mathcal{B}_{\ell-|X|}^*$ as the set of profiles in which each voter submits a ballot in $\mathcal{B}_{\ell-|X|}$. Finally, consider the following SCF $h$ on the domain $\mathcal{B}_{\ell-|X|}^*$: given a profile $A\in\mathcal{B}_{\ell-|X|}^*$, we construct a new profile $A'$ on $\mathcal{C}$ by setting $A_i'=A_i\cup X$ for all voters $i\in N_A$. Then, $h(A)\coloneqq g(A',W)$.  
		
		It is not difficult to see that $h$ inherits anonymity, neutrality, and consistency from $g(\cdot,W)$. In particular, note for neutrality that this axiom requires for $h$ only that we can permute the candidates in $\mathcal{B}$ and such permutation do not affect the set $W$. Thus, we can use \Cref{lem:AV} to derive that $h$ is either $\mathit{AV}$, $\overline{\mathit{AV}}$, or $\mathit{TRIV}$. Hence, we infer that if all candidates $x\in \mathcal{C}\setminus W$ are approved by the same number of voters in $A'$, then $g(A',W)=h(A)=\mathcal{C}\setminus W$. In particular, if $n(x,A,W,|X|,\ell)=\hat c$ for all candidates $x\in \mathcal{C}\setminus W$, $|A_i|=\ell$, and $A_i\cap W=X$ for all $i\in N$, then $g(A,W)=\mathcal{C}\setminus W$.\medskip
		
		\textbf{Step 2:} For the second step, we fix two integers $j_1,j_2\in\mathbb{N}$ with $j_1\leq |W|$ and $j_1< j_2\leq m-1-|W|+j_1$ and consider profiles $A$ such that $|A_i|=j_2$ and $|A_i\cap W|=j_1$ for all voters $i\in N_A$. Or, in other words, all voters still approve the same number of candidates with respect to both $\mathcal{C}$ and $W$, but they are no longer required to report the exact same subset of $W$. We assume once again that there is a constant $\hat c$ such that $n(x,A,W,j_1, j_2)=\hat c$ for all $x\in \mathcal{C}\setminus W$ and will show that $g(A,W)=\mathcal{C}\setminus W$. Note for this case that neutrality requires that $g(\tau(A), \tau(W))=\tau(g(A, W))=g(A,W)$ for every permutation $\tau:\mathcal{C}\rightarrow\mathcal{C}$ with $\tau(x)=x$ for $x\in\mathcal{C}\setminus W$ since $g(A,W)\subseteq \mathcal{C}\setminus W$. Hence, consider the profile $A^*$ which consists of a copy of $A$ permuted by every such permutation $\tau$. By consistency and the previous argument, it follows that $g(A^*,W)=g(A,W)$. 
		
		For proving this step, we consider a different decomposition of $A^*$. In more detail, we can also decompose $A^*$ into profiles subprofiles $A^X$ which precisely contain the voters $i$ with $A_i^*\cap W=X$. Now, observe that for any two sets $X_1, X_2\subseteq W$ with $|X_1|=|X_2|$, there is the same number of permutations $\tau$ with $\tau(x)=x$ for all $x\in\mathcal{C}\setminus W$ that maps $X_1$ to $X_2$. Hence, each profile $A^X$ consists of multiple copies of a profile $\bar A^X$, which is derived from the original profile $A$ by setting $\bar A_i^X=(A_i\setminus W)\cup X$ for all $i\in N_A$. In particular, this shows that there is an integer $j$ such that $n(x, A^X, W, j_1, j_2)=j\hat c$ for all candidates $x\in\mathcal{C}\setminus W$ and $X\subseteq W$ with $|X|=j_1$. Since $A_i^X\cap W=X$ and $|A_i^X|=j_2$ for all $i\in N_{A^X}$, we can thus use the insights of Step 1 to derive that $g(A^X, W)=\mathcal{C}\setminus W$. This holds for every $X\subseteq W$ with $|X|=j_1$ and consistency thus implies that $g(A^* W)=\mathcal{C}\setminus W$ because $A^*$ is the collection of the profiles $A^X$ for all $X\subseteq W$ with $|X|=j_1$. This proves that $g(A,W)=\mathcal{C}\setminus W$ because $g(A,W)=g(A^*,W)$.\medskip
		
		\textbf{Step 3:} Finally, consider an arbitrary profile $A$ for which there are constants $\hat c_{k,\ell}\in\mathbb{N}$ such that $n(x,A,W,k,\ell)=\hat c_{k,\ell}$ for all candidates $x\in \mathcal{C}\setminus W$, $k\in\{0,\dots, |W|\}$, and $\ell\in \{k+1,\dots, m-1-|W|+k\}$. We can partition $A$ into profiles $A^{k,\ell}$ in which voters approve exactly $k\in \{0,\dots, |W|\}$ candidates of $W$ and $l\in \{k, \dots, m-|W|+k\}$ candidates in total. By the definition of $A$, it holds for all profiles $A^{k,\ell}$ with $k<\ell\leq m-1-|W|+k$ that $n(x,A^{k,\ell},W,k,\ell)=\hat c_{k,\ell}$ for all $x\in \mathcal{C}\setminus W$. Thus, Step 2 shows that $g(A^{k,\ell},W)=\mathcal{C}\setminus W$ for all $k,\ell$ with $k<\ell\leq m-1-|W|+k$. 
		
		On the other hand, if $k=\ell$, the voters in $A^{k,l}$ do not approve any candidate in $\mathcal{C}\setminus W$. Thus, neutrality immediately requires that $g(A^{k,\ell},W)=\mathcal{C}\setminus W$ for these profiles. Similarly, if $\ell=m-(|W|-k)$, then all voters in $A^{k,\ell}$ approve all candidates in $\mathcal{C}\setminus W$ and thus neutrality again implies that $g(A^{k,\ell},W)=\mathcal{C}\setminus W$. Hence, it holds for all subprofiles of $A$ that $\mathcal{C}\setminus W$ is chosen and consistency thus shows that $g(A,W)=\mathcal{C}\setminus W$, too.
	\end{proof}
	
	As next step, we will show that we can compute $g(A,W)$ only based on the values $n(x,A,W,k,\ell)$ for all $x\in\mathcal{C}\setminus W$, $k\in \{0,\dots, |W|\}$, and $\ell\in \{k+1, \dots, m-1-|W|+k\}$. 
	For making this more formal, let $Z_W= \{(k,\ell) \colon 0\leq k \leq |W|, k<\ell\leq m-1-|W|+k\}$ denote the set of pairs such that $g$ can rely on $n(x,A,W,k,\ell)$ for computing the outcome. Then, we define $n(x,A,W)=(n(x,A,W,k,\ell))_{(k,\ell)\in Z}$ as the vector that contains all entries $n(x,A,W,k,\ell)$ for $(k,\ell)\in Z$. Moreover, let $N(A,W)$ denote the matrix, which contains the vectors $n(x,A,W)$ as rows for all $x\in\mathcal{C}\setminus W$. 
	
	\begin{lemma}\label{lem:informationbasis}
		Let $g$ denote a proper, consistent, and complete generator function. It holds that $g(A,W)=g(A',W)$ for all profiles $A$, $A'$ and committees $W$ such that $N(A,W)=N(A',W)$ and $|W|\leq m-2$.
	\end{lemma}
	\begin{proof}
		Let $g$ denote a proper generator function satisfying consistency and completeness. Moreover, consider two arbitrary profiles $A$, $A'$ and a committee $W$ such that $N(A, W)=N(A', W)$ and $|W|\leq m-2$. As next step, let $A''$ denote a profile containing $|N_{A'}|$ copies of every ballot $A_i$ and let $A^*=A''+A-A'$ denote the profile derived from $A''$ by first adding all ballots in $A$ and then removing the ballots of $A'$. Since $N(A, W)=N(A', W)$, it follows immediately that there are constants $\hat c_{k,\ell}$ such that $n(x,A'',W,k,\ell)=n(x,A^*,W,k,\ell)=\hat c_{k,l}$ for all $x\in \mathcal{C}\setminus W$, $k\in \{0,\dots, |W|\}$ and $\ell\in \{k+1,\cdots, m-1-|W|+k\}$. \Cref{lem:0element} therefore shows that $g(A'', W)=g(A^*,W)=\mathcal{C}\setminus W$. Using consistency, we can thus derive that $g(A,W)=g(A+A'',W)=g(A^*+A',W)=g(A',W)$, which proves the lemma. 
	\end{proof}
	
	As a consequence of \Cref{lem:informationbasis}, we can view every proper, consistent, and complete generator function $g$ as a mapping that computes the winning candidates only based on $N(A,W)$ instead of the profile $A$ itself. Or, in other words, for every committee $W$ there is a function $g_W(Q)$ that maps each element $Q\in D^\mathbb{N}_W=\{N(A,W)\colon A\in\mathcal{A}^*\}$ to $\mathcal{C}\setminus W$ such that $g(A, W)=g_W(N(A,W))$. 
	
	Before deriving the next lemma, we point out a number of important observations on $D^\mathbb{N}_W$ and $g_W$. In particular, note that $D^\mathbb{N}_W$ is closed under addition since $N(A+A',W)=N(A,W)+N(A',W)$ and multiplication with integers $k\in \mathbb{N}$ because $N(kA,W)=k N(A,W)$.  Moreover, $g_W$ inherits a number of important properties from $g$:
	\begin{itemize}
		\item $g_W$ is consistent: for all $Q, Q'\in D^{\mathbb{N}}_W$ with $g_W(Q)\cap g_W(Q')\neq \emptyset$, it holds that $g_W(Q+Q')=g_W(Q)\cap g_W(Q')$. The reason for this is that there is are profiles $A,A'\in\mathcal{A}^*$ with $Q=N(A,W)$ and $Q'=N(A',W)$, which implies that $g(A,W)=g_W(Q)$ and $g(A',W)=g_W(Q')$. Hence, consistency of $g$ requires that $g(A+A',W)=g(A,W)\cap g(A',W)$ and, because $N(A+A',W)=N(A,W)+N(A',W)=Q+Q'$, we therefore infer that $g_W(Q+Q')=g_W(Q)\cap g_W(Q')$.
		\item $g_W$ is neutral: for a permutation $\tau:\mathcal{C}\setminus W\rightarrow\mathcal{C}\setminus W$, let $\tau(Q)$ denote the matrix derived by reordering the rows of $Q$ according to $\tau$. Then, $g_W(\tau(Q))=\tau(g_W(Q))$ for all permutations $\tau:\mathcal{C}\setminus W\rightarrow\mathcal{C}\setminus W$ and matrices $Q\in D^\mathbb{N}_W$. This follows since there is a profile $A$ with $Q=N(A,W)$ and $\tau(Q)=N(\tau(A),W)$, which entails that $g_W(\tau(Q))=g(\tau(A), W)=\tau(g(A,W))=\tau(g_W(Q))$. 
		\item $g_W$ is continuous: for all $Q, Q'\in D^{\mathbb{N}}_W$ with $|g_W(Q)|=1$, there is an integer $k$ such that $g_W(kQ+Q')=g_W(Q)$. This follows again by going back to profiles $A,A'$ with $Q=N(A,W)$ and $Q'=N(A',W)$: $|g_W(Q)|=1$ implies that $|g(A,W)|=1$ and thus, the continuity of $g$ entails that there is $k\in \mathbb{N}$ such that $g(kA+A',W)=g(A,W)$. Since $N(kA+A',W)=kQ+Q'$, our claim follows.
		\item $g_W$ is non-imposing: for all candidates $c\in\mathcal{C}\setminus W$, there is a $Q\in D^{\mathbb{N}}_W$ such that $g_W(Q)=\{c\}$. This follows as we can find a profile $A$ such that $g(A, W)=\{c\}$.
	\end{itemize}
	
	Next, we show that we can extend $g_W$ to a function on $\mathbb{Q}^{|\mathcal{C}\setminus W|\times |Z_W|}$ while maintaining the properties above and that $g_W(Q)=g(A,W)$ for all profiles $A\in\mathcal{A}^*$ such that $Q=N(A,W)$.

	\begin{lemma}\label{lem:extension}
		Let $g$ denote a proper generator function satisfying consistency and completeness and fix a committee $W$ with $|W|\leq m-2$. There is a function $g_W$ from $\mathbb{Q}^{|\mathcal{C}\setminus W|\times |Z_W|}$ to $2^{\mathcal{C}\setminus W}\setminus \{\emptyset\}$ such that
		\begin{enumerate}
			\item for all $Q\in D^{\mathbb{N}}_W$ and $A\in \mathcal{A}^*$ with $Q=N(A,W)$, it holds that $g_W(Q)=g(A,W)$. 
			\item $h$ is consistent, neutral, continuous, and non-imposing. 
		\end{enumerate}
	\end{lemma}
	\begin{proof}
		Let $g$ denote a proper, consistent, and complete generator function and fix a committee $W$ with $|W|\leq m-2$. Because of \Cref{lem:informationbasis}, there is a function $g_W(Q)$ from $D^\mathbb{N}_W$ to $\mathcal{C}\setminus W$ such that $g(A, W)=g_W(N(A,W))$ for all profiles $A\in\mathcal{A}^*$. Moreover, by the observation prior to this lemma, $g_W$ is continuous, neutral, consistent, and non-imposing. Also, recall that $D^\mathbb{N}_W$ is closed under addition of its elements and under multiplication with scalars in $\mathbb{N}$.
		
		For proving the lemma, we will heavily rely on the profile $A^*$ in which each ballot $A_i\in\mathcal{A}$ is submitted exactly once and its corresponding matrix $E=N(A^*,W)$ because the the symmetry of $A^*$ ensures that $g_W(E)=\mathcal{C}\setminus W$. Based on $E$, we will first show that the space $D^{\mathbb{Z}}_W=\{Q-kE\colon k\in \mathbb{N}_0, Q\in D^\mathbb{N}_W\}$ is closed under addition of its elements and under multiplication with scalars in $\mathbb{Z}$. Based on this insight, we will extend $g_W$ from $D^{\mathbb{N}}_W$ and show that our extension satisfies the required properties. In Step 3 and Step 4, we proceed analogously for $D^{\mathbb{Q}}_W=\{Q/k\colon k\in \mathbb{N}, Q\in D^\mathbb{Z}_W\}$. Finally, we will show that $D^{\mathbb{Q}}_W=\mathbb{Q}^{|\mathcal{C}\setminus W|\times |Z_W|}$, which then proves the lemma.\medskip
		
		\textbf{Step 1: $D^{\mathbb{Z}}_W$ is closed under addition and multiplication with scalars in $\mathbb{Z}$.}
	
		Our first goal is to show that $D^{\mathbb{Z}}_W=\{Q-kE\colon: k\in \mathbb{N}_0, Q\in D^\mathbb{N}_W\}$ is closed under addition and multiplication with scalars in $\mathbb{Z}$. We start by discussing the claim on addition. For this, observe that if $Q, Q'\in D^{\mathbb{Z}}_W$, then there are $P, P'\in D^\mathbb{N}_W$ and integers $k,k'\in\mathbb{N}_0$ such that $Q=P-kE$ and $Q'=P'-k'E$. Since $D^\mathbb{N}_W$ is closed under addition, we derive that $P+P'\in D^\mathbb{N}_W$. This implies that $Q+Q'=P+P'-(k+k')E\in D^\mathbb{Z}_W$. 
		
		Next, we show that $D^\mathbb{Z}_W$ is closed under multiplication with scalars $k\in \mathbb{Z}$. For this, let $Q$ denote an arbitrary element of $D^{\mathbb{Z}}_W$ and note that, by definition, there are $P\in D^{\mathbb{N}}_W$ and $\ell\in \mathbb{N}_0$ such that $Q=P-\ell E$. We need to show that $kQ\in D^\mathbb{Z}_W$ for all $k\in \mathbb{Z}$ and proceed for this with a case distinction with respect to $k$. First, observe that $0Q$ is the matrix containing only $0$'s. Since $E\in D^\mathbb{N}_W$ and $0=E-E$, this $0$-matrix is in $D^\mathbb{Z}_W$. Next, if $k \in\mathbb{N}$, $kQ\in D^\mathbb{Z}_W$ since $D^\mathbb{N}_W$ is closed under multiplication with scalars in $\mathbb{N}$. In particular, this means that $k P\in D^\mathbb{N}_W$ and thus, $k Q=k P - k\ell E\in D^{\mathbb{Z}}_W$. 
		
		As last case, suppose that $k$ is negative. By the last case, we already know that $-k Q\in D^{\mathbb{Z}}_W$ since $-k\in\mathbb{N}$. This means that there are $P'\in D^{\mathbb{N}}_W$ and $\ell'\in \mathbb{N}_0$ such that $-kQ=P'-\ell' E$. Next, let $\ell''\geq \ell'$ denote an integer such that all entries in $P'-\ell''E$ are negative; such an $\ell''$ exists since all profiles are finite. By definition of $D^{\mathbb{N}}_W$, there is a profile $A'$ such that $P'=N(A',W)$ and recall that $E=N(A^*,W)$, where $A^*$ contains every ballot exactly once. Now, construct the profile $A''$ as follows: first, we clone the profile $A^*$ $\ell''$ times and then we remove for each voter in $i\in N_{A'}$ a voter with the corresponding ballot in $A'_i$ from $A''$. It is not difficult to see that $P''=N(A'',W)=\ell''N(A^*,W)- N(A',W)$ and of course, $P''\in D^\mathbb{N}_W$. Finally, we derive from this observation that $kQ=\ell'E - P'=P''-(\ell''-\ell')E$, which proves that $kQ\in D^\mathbb{Z}_W$. Hence, $D^\mathbb{Z}_W$ is indeed closed under multiplication with scalars $k\in\mathbb{Z}$.\medskip
		
		\textbf{Step 2: Extending $g_W$ to $D^{\mathbb{Z}}_W$.}
		
		 As second step, we extend $g_W$ to a function $\hat g_W$ on $D^{\mathbb{Z}}_W$. In particular, we define $\hat g_W(Q-kE)=g_W(Q)$ for every $k\in\mathbb{N}_0$ and $Q\in{D}^{\mathbb{N}}_W$. This is well-defined because of consistency: if there are two different matrices $Q,Q'\in D^\mathbb{N}_W$ and $k,\ell\in\mathbb{N}_0$ such that $Q-kE=Q'-\ell E$, then $Q'=Q+(k-\ell)E$. Assuming that $k>\ell$, we can derive $Q'$ from $Q$ by adding $k-\ell$ copies of $E$. Thus, consistency implies that $g_W(Q')=g_W(Q)\cap g((k-\ell)E)=g_W(Q)$. If $k<\ell$, we can use an analogous argument by exchanging the roles of $Q$ and $Q'$.
		 
		 Next, we show that $\hat g_W$ satisfies all required properties. First, note that for $Q\in D^\mathbb{N}_W$, we have that $\hat g_W(Q)=g_W(Q)=g(A,W)$ for all profiles $A\in\mathcal{A}^*$ with $Q=N(A,W)$. This immediately implies also that $\hat g_W$ is non-imposing as $g_W$ satisfies this property. 
		 
		 For proving that $\hat g_W$ is neutral, consistent, and continuous, slightly more involved arguments are required. For presenting them, let $Q,Q'\in D^\mathbb{Z}_W$ and note that, by definition of $D^\mathbb{Z}_W$, there is $P,P'\in D^\mathbb{N}_W$ and $k,k'\in \mathbb{N}_0$ such that $Q=P-kE$ and $Q'=P'-k'E$. Using the definition of $\hat g_W$, it thus follows that $\hat g_W(Q)=g_W(P)$ and $\hat g_W(Q')=g_W(P')$. 
		 
		 We are now ready to show that $\hat g_W$ is neutral. For doing so, let $\tau$ denote a permutation on $\mathcal{C}\setminus W$ and $\tau(Q)$ the matrix derived from $Q$ by permuting its rows according to $\tau$. It is not difficult to see that $\tau(Q)=\tau(P)-k\tau(E)$. Moreover, $E$ is completely symmetric and thus $\tau(E)=E$. Hence, we infer that $\hat g_W(\tau(Q))=\hat g_W(\tau(P)-kE)=g_W(\tau(P))=\tau(g_W(P))=\tau(\hat g_W(P-kE))=\tau(\hat g_W(Q))$, which shows that $\hat g_W$ is neutral. 
		 
		 As next claim, we prove that $\hat g_W$ is consistent. Thus, assume that $\hat g_W(Q)\cap \hat g_W(Q')\neq \emptyset$. The consistency of $g_W$ implies that $g_W(P+P')=g_W(P)\cap g_W(P')$ and thus, $\hat g_W(Q+Q')=\hat g_W(P+P'-(k+k')E)=g_W(P+P')=g_W(P)\cap g_W(P')=\hat g_W(Q)\cap \hat g_W(Q')$. 
		 
		 A similar argument shows that $\hat g_W$ is continuous. For this, suppose that $|\hat g_W(Q)|=1$. Since $|g_W(P)|=|\hat g_W(Q)|=1$, the continuity of $g_W$ implies that there is an integer $\ell\in\mathbb{N}$ such that $g_W(\ell P+P')=g_W(P)$. Finally, using the definition of $\hat g_W$ again, we derive that $\hat g_W(\ell Q+Q')=\hat g_W(\ell P+P'-(\ell k+k')E)=g_W(\ell P + P')=g_W(P)=\hat g_W(Q)$. This shows that $\hat g_W$ is continuous.\medskip
		 
		 \textbf{Step 3: $D^{\mathbb{Q}}_W$ is closed under addition and multiplication with scalars in $\mathbb{Q}$.}
		 
 		Next, we show again that $D^\mathbb{Q}_W$ is closed under addition and multiplication with scalars in $\mathbb{Q}$. For addition, consider two elements $Q, Q'\in D^\mathbb{Q}_W$. By definition there are $P, P'\in D^\mathbb{Z}_W$ and $k, k'\in \mathbb{N}$ such that $Q=P/k$ and $Q'=P'/k'$. Clearly, $k' P, k P'\in D^\mathbb{Z}_W$ since $D^\mathbb{Z}_W$ is closed under multiplication with integers. Hence, $k'P+kP'\in D^\mathbb{Z}_W$ due to is closure under addition. Finally, this means that $Q+Q'=P/k+P'/k'=(k'P+kP')/(k\cdot k')\in D^{\mathbb{Q}}_W$, which shows that $D^\mathbb{Q}_W$ is closed under addition. 
		
 		As last point, we show that $D^{\mathbb{Q}}_W$ is closed under multiplying with scalars $k\in \mathbb{Q}$. Since $k\in \mathbb{Q}$, there are $\ell_1\in\mathbb{Z}$ and $\ell_2\in\mathbb{N}$ such that $k=\frac{\ell_1}{\ell_2}$. Now, consider an arbitrary $Q\in D^\mathbb{Q}_W$ and recall that by definition, there is $P\in D^\mathbb{Z}_W$ and $\ell_3\in\mathbb{N}$ such that $Q=P/\ell_3$. Since $D^{\mathbb{Z}}_W$ is closed under multiplication with scalars in $\mathbb{Z}$, we have that $\ell_1 P \in D^{\mathbb{Z}}_W$. Because $\ell_2\cdot \ell_3\in\mathbb{N}$, we thus have that $Q=\ell_1 P/(\ell_2\cdot \ell_3)\in D^\mathbb{Q}_W$ by definition.\medskip
		
		\textbf{Step 4: Extending $\hat g_W$ to $D^{\mathbb{Q}}_W$.}
		
		As fourth step, we extend $\hat g_W$ to $D^\mathbb{Q}_W$ by defining $\bar g_W(Q/k)=\hat g_W(Q)$ for every $Q\in D^\mathbb{Z}_W$, $k\in\mathbb{N}$. Once again, consistency ensures that this is well-defined: if there are $Q,Q'\in D^\mathbb{Z}_W$ and $k,\ell\in \mathbb{N}$ such that $Q/k=Q'/\ell$, then the consistency of $\hat g_W$ ensures that $\bar g_W(Q/k)=\hat g_W(Q)=\hat g_W(\ell Q)=\hat g_W(k Q')=\hat g_W(Q')=\bar g_W(Q'/\ell)$. 
		
		Moreover, note that $\bar g_W(Q)=g_W(Q)=g(W,A)$ for all $Q\in D^\mathbb{N}_W$ and profiles $A$ with $Q=N(A,W)$ by the definitions of $\bar g_W$, $\hat g_W$, and $g_W$. Hence, $\bar g_W$ indeed satisfies the first condition of this lemma. Also, this shows that $\bar g_W$ is non-imposing as even $g_W$, which is defined on $D^\mathbb{N}_W\subseteq D^{\mathbb{Q}}_W$, satisfies this axiom.
		
		Analogous to Step 2, proving the neutrality, consistency, and continuity of $\bar g_W$ takes more effort and we consider therefore $Q, Q'\in D^{\mathbb{Q}}_W$. By the definition of $D^{\mathbb{Q}}_W$, there are $P, P'\in D^{\mathbb{Z}}_W$ and $k,k'\in\mathbb{N}$ such that $Q=P/k$ and $Q'=P'/k'$. The definition of $\bar g_W$ then shows that $\bar g_W(Q)=\hat g_W(P)$ and $\bar g_W(Q')=\hat g_W(P')$.
		
		Now, we prove that $\bar g_W$ inherits neutrality from $\hat g_W$. For showing this, let $\tau$ denote a permutation on $\mathcal{C}\setminus W$. It is apparent that $\tau(Q)=\tau(P)/k$ and thus, $\bar g_W(\tau(Q))=\bar g_W(\tau(P)/k)=\hat g_W(\tau(P))=\tau(\hat g_W(P))=\tau(\bar g_W(P/k))=\tau(\bar g_W(Q))$.
		
		Next, we will show that $\bar g_W$ is consistent. For this, assume that $\bar g_W(Q)\cap \bar g_W(Q')\neq \emptyset$, which implies that $\hat g_W(P)\cap \hat g_W(P')\neq \emptyset$. We infer from the consistency of $\hat g_W$ that $\hat g_W(k'P+kP')=\hat g_W(k'P)\cap \hat g_W(kP')=\hat g_W(P)\cap \hat g_W(P')$. This implies that $\bar g_W(Q+Q')=\bar g_W(P/k+P'/k')=\bar g_W((k'P+kP')/(k\cdot k'))=\hat g_W(k'P+kP')=\hat g_W(P)\cap \hat g_W(P')=\bar g_W(Q)\cap \bar g_W(Q')$. Hence, $\bar g_W$ is consistent. 
		
		As last point, we prove that $\bar g_W$ is continuous and we thus suppose that $|\bar g_W(Q)|=1$. For proving this claim, note that consistency implies that $\hat g_W(P)=\hat g_W(k'P)$ and $\hat g_W(P')=\hat g_W(kP')$. Since $\bar g_W(Q)=\hat g_W(P)$, there is an integer $\ell\in\mathbb{N}$ such that $\hat g_W(\ell k' P+kP')=\hat g_W(k'P)$. This means for $\bar g_W$ that $\bar g_W(\ell Q+Q')=\bar g_W((\ell k' P + k P')/(k\cdot k'))=\hat g_W(\ell k' P + k P')=\hat g_W(k'P)=\hat g_W(P)=\bar g_W(Q)$, i.e., $\bar g_W$ is continuous.\medskip
		
		\textbf{Step 5: $D^{\mathbb{Q}}_W=\mathbb{Q}^{|\mathcal{C}\setminus W|\times |Z_W|}$.}
		
		Finally, we will show that $D^\mathbb{Q}_W$ is equal to the full space $\mathbb{Q}^{|\mathcal{C}\setminus W|\times |Z_W|}$. For proving this, we will show that the standard basis of $\mathbb{Q}^{|\mathcal{C}\setminus W|\times |Z_W|}$ is part of $D^\mathbb{Q}_W$. Hence, consider a fixed candidate $c\in\mathcal{C}\setminus W$ and a tuple $(k,\ell)\in Z_W$. First, let $A^1$ denote the profile in which each ballot $A_i\in\mathcal{A}$ except those with $|A_i|=\ell$ and $|A_i\cap W|=k$ appears once. It is not difficult to see that $Q^1=N(A^1,W)$ differs from $E$ only in the column corresponding to $(k,\ell)$ since all these entries are $0$ for $Q^1$. Next, let $Q^2=Q^1-E$ and note that $Q^2\in D^\mathbb{Z}_W$. This matrix has non-zero entries only in the column $(k,\ell)$, and all entries in this column are equal and negative, i.e., there is $x_1$ such that $Q^2_{d,k,\ell}=-x_1$ for all $d\in\mathcal{C}\setminus W$. 
		
		Furthermore, let $A^3$ denote the profile which contains each ballot $A_i\in\mathcal{A}$ with $|A_i|=\ell$, $|A_i\cap W|=k$, and $c\in A_i$ once. Note that there is at least one such ballot because $k\leq |W|$ and $k<\ell$. Moreover, no such ballot contains all candidates $x\in\mathcal{C}\setminus W$ because $\ell<m-|W|+k$. Hence, it follows that $c$ is approved by strictly more voters than any other candidate $d\in\mathcal \mathcal{C}\setminus (W\cup \{c\})$. On the other hand, due to the symmetry of $A^3$, all these candidates $d$ are approved by the same number of voters. Finally, note that $Q^3=N(A^3,W)$ has only non-zero entries in the column corresponding to $(k,\ell)$. Hence, there are two positive constants $x_2$, $x_3$ such that $x_2>x_3$, $Q^3_{c,k,\ell}=x_2$, and $Q^3_{d,k,\ell}=x_3$ for all $d\in\mathcal{C}\setminus (W\cup \{c\})$.
		
		As last step, let $Q^4=x_1 Q^3 + x_3 Q^2$. Since $D^\mathbb{Z}_W$ is closed under multiplication with scalars in $\mathbb{Z}$ and addition of its elements, it follows that $Q^4\in D^\mathbb{Z}_W$. Moreover, we infer from our previous observations that $Q^4_{c,k,\ell}=x_1\cdot (x_2-x_3)>0$, whereas all other entries are $0$. Hence, the matrix $Q^5=\frac{Q^4}{x_1\cdot (x_2-x_3)}$ contains $1$ at $Q^5{c,k,\ell}$ and $0$ for all other entries. Moreover, this matrix is in $D^\mathbb{Q}_W$ by the definition of this set. Since $c\in\mathcal{C}\setminus W$ and $(k,\ell)\in Z_W$ are arbitrarily chosen, it follows that the standard basis is part of $D^\mathbb{Q}_W$. Finally, this shows that $D^{\mathbb{Q}}_W=\mathbb{Q}^{|\mathcal{C}\setminus W|\times |Z_W|}$ since $D^\mathbb{Q}_W$ is closed under addition of its elements and multiplication with scalars in $\mathbb{Q}$.
	\end{proof}
	
	Finally, we are able to prove \Cref{prop:hyperplane}. For showing this statement, we will use a separation theorem for convex sets and thus, we will use standard terminology from convex optimization (e.g., polyhedron, subspace, dimension, facets) in the subsequent proof. We refer to \citet{McLe18a} for the definitions of these terms.
	
	\hyperplane*
	\begin{proof}
		Let $g$ be defined as in the lemma and first note that the case that $|W|=m-1$ is trivial as there is only a single remaining candidate. Hence, $g(A,W)=\mathit{AV}_v^W(A)$ for every weight function $v$ as both are by definition always non-empty. Thus, consider a committee $W$ of size $|W|\leq m-2$. Moreover, we define $d_1=|\mathcal{C}\setminus W|$ and $d_2=|Z_W|$. Finally, in this proof we will denote the candidates in $\mathcal{C}\setminus W$ merely by numbers from $1$ to $d_1$.
		
		By \Cref{lem:extension}, there is a function $g_W$ from $\mathbb{Q}^{d_1\times d_2}$ to $2^{\mathcal{C}\setminus W}\setminus \{\emptyset\}$ that is consistent, neutral, non-imposing, and continuous, and that satisfies that $g_W(N(A,W))=g(A,W)$ for all $A\in\mathcal{A}^*$. Next, define $R_i=\{Q\in\mathbb{Q}^{d_1\times d_2}\colon c_i\in g_W(Q)\}$ for every $i\in \{1,\dots, d_1\}$. Moreover, let $\bar R_i$ denote the closure of $R_i$ with respect to $\mathbb{R}^{d_1\times d_2}$. 
		
		We will prove the proposition in multiple steps by analyzing the sets $\bar R_i$. In more detail, we show in Step 1 that these sets are full-dimensional and convex cones. This implies that they have a non-empty interior. As next step, we prove that the interiors of these sets are disjoint. We can therefore use the separation theorem for convex sets to find a separating hyperplane between every pair $\bar R_i$, $\bar R_j$. Even more, we show in the third step that these hyperplanes are unique up to multiplication with a positive scalar. As fourth step, we extract a scoring vector from these separating hyperplanes and show thereafter that $g_W(Q)$ can be represented based on this score vector. Finally, we derive from this insight that $g(A,W)$ can be represented by $\mathit{AV}_v^W(A)$ for some weight vector $v$.\medskip
		
		\textbf{Step 1:} Our first goal is to show that $\bar R_i$ is a fully dimensional and convex cone for every $i\in \{1,\dots, d_1\}$. For this, note that the consistency of $g_W$ implies that $R_i$ is a $\mathbb{Q}$-cone. 
		(A set is called $\mathbb{Q}$-convex if it is closed with respect to convex combinations using \emph{rational} scalars $0\leq q\leq 1$ instead of real ones. Moreover, a $\mathbb{Q}$-cone is a $\mathbb{Q}$-convex set that is closed with respect to multiplication of any non-negative, rational scalar.) 
		It is not difficult to see that $\bar R_i$, i.e., the closure of $R_i$ with respect to $\mathbb{R}^{d_1\times d_2}$, is convex for every $i\in\{1,\dots, d_1\}$. This is also formally proven by \citet{Youn75a}.
		Moreover, $\bigcup_{i\in\{1,\dots, d_1\}} \bar R_i = \mathbb{R}^{d_1\times d_2}$ and neutrality entails that $\bar R_{\tau(i)} = \tau(\bar R_i)$ for all permutations $\tau:\{1,\dots, d_1\} \rightarrow \{1,\dots, d_1\}$ and $i\in\{1,\dots, d_1\}$. 
		In particular, the latter fact means that all $\bar R_i$ have the same dimension and they must thus have the same dimension as $\mathbb{R}^{d_1\times d_2}$. 
		Hence, the interior of $\bar R_i$ (with respect to $\mathbb{R}^{d_1\times d_2}$), denoted by $\text{int } \bar R_i$, is non-empty.\medskip
		
		\textbf{Step 2:} Next, we will show that that the interiors of $\bar R_i$, $\bar R_j$ do not intersect. Hence, assume for contradiction that $\text{int } \bar R_i\cap \text{int } \bar R_j\neq \emptyset$ for some $i,j \in\{1,\dots, d_1\}$ with $i\neq j$. Then, there is $Q\in D^\mathbb{Q}\cap \text{int } \bar R_i\cap \text{int } \bar R_j$, which means that $\{i, j\}\subseteq g_W(Q)$. On the other hand, there is $Q'$ such that $g_W(Q')=\{i\}$ because $g_W$ is non-imposing. 
		
		Now, since $Q$ is in the interior of both $\bar R_i$ and $\bar R_j$, we can find a $\lambda\in \mathbb{Q}$ such that $0<\lambda<1$ and $(1-\lambda) Q + \lambda Q'\in D^\mathbb{Q} \cap \text{int } \bar R_i\cap \text{int } \bar R_j$. This means that $(1-\lambda) Q + \lambda Q'\in R_i\cap R_j$ and thus $\{i, j\}\subseteq g_W((1-\lambda) Q + \lambda Q')$. However, consistency requires that $g_W((1-\lambda) Q + \lambda Q')=\{i\}$, contradicting the previous claim. Thus, $\text{int } \bar R_i\cap \text{int } \bar R_j= \emptyset$ for all distinct candidates $i, j\in \{1,\dots,d_1\}$.\medskip
		
		\textbf{Step 3:} As third step, we show for all $i,j\in\{1,\dots, d_1\}$ that there is a unique hyperplane (up to multiplication with positive scalars) that separates $\bar R_i$ and $\bar R_j$. Note for this that, because these sets are convex and their interiors do not intersect, the separating hyperplane theorem \citep[e.g.,][]{McLe18a} shows that there is a non-zero vector $u^{ij}\in\mathbb{R}^{d_1\times d_2}$ such that $u^{ij}Q\geq 0$ if $Q\in \bar R_i$ and $u^{ij}Q\leq 0$ if $Q\in \bar R_j$. 
		We define here the matrix multiplication $u^{ij}Q$ as the standard scalar product $\sum_{k\in \{1,\dots, d_1\}, \ell\in \{1,\dots, d_2\}} u^{ij}_{k,\ell} Q_{k,\ell}$.
		 Furthermore, we suppose that $u^{ji}=-u^{ij}$ for all $i,j \in\{1,\dots, d_1\}$ with $i\neq j$. 
		Clearly, if $u^{ij}$ satisfies that $u^{ij}Q\geq 0$ if $Q\in \bar R_i$ and $u^{ij}Q\leq 0$ if $Q\in \bar R_j$, then $u^{ji}Q\geq 0$ if $Q\in \bar R_j$ and $u^{ji}Q\leq 0$ if $Q\in \bar R_i$. Hence, it suffices to derive only one of these two vectors from a hyperplane argument. 
		
		Next, we show that the $u^{ij}$ are unique up to multiplication with a positive scalar. For this, let $S_i=\{Q\in \mathbb{R}^{d_1\times d_2}\colon \forall j\in \{1,\dots, d_1\}, j\neq i\colon u^{ij}Q\geq 0\}$ for all candidates $i\in\{1,\dots, d_1\}$. By definition, we have that $\bar R_i\subseteq S_i$. 
		Hence, $\emptyset\subsetneq \text{int } \bar R_i\subseteq \text{int } S_i=\{Q\in \mathbb R^{d_1\times d_2}\colon: \forall j\in \{1, \dots, d_1\}, j\neq i: u^{ij}Q> 0\}$. Furthermore, if $u^{ij}Q>0$ for some $Q$ and $i,j\in \{1,\dots,d_i\}$, then $u^{ji} Q=-u^{ij} Q<0$, which implies that $Q\not\in\bar R_j$. 
		Therefore, $\text{int } S_i\subseteq R^{d_1\times d_2}\setminus \bigcup_{j\in \{1,\dots, d_j\}\setminus \{i\}} \bar R_j$. Finally, since $\bigcup_{j \in \{1,\dots, d_1\}} \bar R_j = \mathbb{R}^{d_1\times d_2}$, it holds that $R^{d_1\times d_2}\setminus \bigcup_{j\in \{1,\dots, d_j\}\setminus \{i\}} \bar R_j\subseteq \bar R_i$. By combining these insights, we derive that $\text{int } \bar R_i \subseteq \text{int } S_i\subseteq \bar R_i$. By taking the closure, it thus follows that $\bar R_i=S_i$. 
		
		In particular, this means that $\bar R_i\neq \mathbb{R}^{d_1\times d_2}$ is a full-dimensional polyhedron and thus, it has a facet $F$ of dimension $d_1\cdot d_2 -1$. Since all $\bar R_j$ are closed and $\bigcup_{j\in \{1,\dots, d_1\}} \bar R_j =\mathbb{R}^{d_1\times d_2}$, it follows that $\bigcup_{j\in \{1,\dots, d_1\}\setminus \{i\}} \bar R_j\cap F=F$. Finally, since $d_1$ is finite, this means that there is $j\neq i$ such that $F\cap \bar R_j$ has dimension $d_1\cdot d_2-1$. This also shows that the intersection of $\bar R_i$ and $\bar R_j$ has dimension of $d_1\cdot d_2 -1$. Now, by symmetry this must hold for all $i,j\in \{1,\dots, d_1\}$, $i\neq j$ and the dimensionality of $\bar R_i\cap \bar R_j$ implies then that $u^{ij}$ is unique up to multiplication with positive scalars.\medskip

		\textbf{Step 4:}
		Our next goal is to represent $g_W$ by a weight vector. For this purpose, we show first that $u^{\tau(i)\tau(j)}=\tau(u^{ij})$ for all permutations $\tau:\{1,\dots, d_1\}\rightarrow\{1,\dots, d_1\}$ and distinct $i,j\in\{1,\dots, d_1\}$. For proving this claim, fix arbitrary $i,j$ and $\tau$. Moreover, let $\tau^{-1}$ denote the inverse permutation of $\tau$, i.e., $\tau^{-1}(\tau(x))=x$ for all $x\in\{1,\dots, d_1\}$. By the neutrality of $g_W$, we derive that $Q\in \bar R_i$ if and only if $\tau(Q)\in \bar R_{\tau(i)}$ and $Q\in \bar R_j$ if and only if $\tau(Q)\in \bar R_{\tau(j)}$. Furthermore, by the definition of our matrix multiplication, it holds that $\tau (Q) u^{ij}=Q\tau (u^{ij})$ for all matrices $Q\in\mathbb{R}^{d_1\times d_2}$. 
		
		Now, let $Q\in \bar R_{\tau(i)}$. It follows by our previous observation that $\tau^{-1}(Q)\in \bar R_i$ and thus $\tau^{-1}(Q)u^{ij}\geq 0$. Hence, we have that $Q\tau(u^{ij})\geq 0$. An analogous argument holds for $\bar R_j$, and thus, we have that $Q\tau(u^{ij})\geq 0$ if $Q\in \bar R_{\tau(i)}$ and $Q\tau(u^{ij})\leq 0$ if $Q\in \bar R_{\tau(j)}$. By the uniqueness of the separating hyperplane, we thus infer that $u^{\tau(i)\tau(j)}=\tau(u^{ij})$.
		
		In particular, observe that this claim also holds for the permutation $\tau^{ij}$ which only swaps $i$ and $j$. Hence, we have that $\tau^{ij}(u^{ij})=u^{ji}=-u^{ij}$. Since $\tau^{ij}$ only swaps the $i$-th and $j$-th row of $u^{ij}$, we infer that $u^{ij}_{\ell,k}=-u^{ij}_{\ell,k}$ for $\ell\in \{i,j\}$, $k\in \{1,\dots, d_2\}$ and $u^{ij}_{\ell,k}=u^{ji}_{\ell,k}=0$ for $\ell\in \{1,\dots, d_1\}\setminus \{i,j\}$ and $k\in \{1,\dots, d_2\}$. 
		
		Now, let $s=u^{ij}_i$, i.e., $s$ is the $i$-th row of $u^{ij}$. The argument in the last paragraph shows that $Q u^{ij}\geq 0$ if $Q_i s\geq Q_j s$ (here $Q_i s=\sum_{k\in \{1,\dots, d_2\}} Q_{i,k}s_k$). Now, by the symmetry of the $u^{ij}$, it follows that $\bar R_i=\{Q\in \mathbb{R}^{d_1\times d_2}\colon \forall j\in \{1,\dots, d_1\}\setminus \{i\}\colon u^{ij}Q\geq 0\}=\{Q\in \mathbb{R}^{d_1\times d_2}\colon \forall j\in \{1,\dots, d_1\}\colon Q_i s\geq Q_j s\}$.\medskip
		
		\textbf{Step 5}: 
		Let $h_W(Q)=\{i\in \{1,\dots,d_1\}\colon \forall j\in \{1,\dots, d_1\}\colon Q_i s\geq Q_j s\}$, where $s$ is the vector derived in the last step.
		In this step, we show that $g_W(Q)=h_W(Q)$ for all $Q\in\mathbb{Q}^{d_1\times d_2}$. For this, note that the definition that $h_W$ shows that it is neutral and consistent. Moreover, it is non-imposing as $s$ is a non-zero vector. This follows as the underlying $u^{ij}$ are also non-zero vectors by the separating hyperplane theorem. Hence, let $s_k$ denote a non-zero entry in $s$. If $s_k>0$, it follows that $h_W(Q)=\{i\}$ for the matrix $Q$ in which there is a one in $Q_{i,k}$ and $0$ everywhere else, and if $s_k<0$, the same holds for the matrix $Q$ with $Q_{j,k}=1$ for all $j\in \{1,\dots, d_1\}$ with $j\neq i$ and $0$ otherwise. 
		
		Next, observe that, by the reasoning in Step 4, we have that $i\in h_W(Q)$ if and only if $Q\in \bar R_i$. Since $R_i\subseteq \bar R_i$ and $Q\in R_i$ only if $i\in g_W(Q)$, it follows immediately that $g_W(Q)\subseteq h_W(Q)$ for all $Q\in \mathbb{Q}^{d_1\times d_2}$.
		
		Finally, suppose there is $Q\in \mathbb{Q}^{d_1\times d_2}$ such that $g_W(Q)\subsetneq h_W(Q)$ and let $i\in g_W(Q)$. Now, let $\tau$ denote an arbitrary permutation such that $\tau(i)=i$ and $\tau(j)=j$ for all $j\in \{1,\dots, d_1\}\setminus h_W(Q)$. The neutrality of $g_W$ implies that $i\in g_W(\tau(Q))$ and, since $h_W$ is by definition neutral, it follows that $h_W(\tau(Q))=\tau(h_W(Q))=h_W(Q)$. Now, let $Q^*$ denote the matrix derived by summing up all $\tau(Q)$ for permutations $\tau$ with $\tau(i)=i$ and $\tau(j)=j$ for $j\in \{1,\dots, d_1\}\setminus h_W(Q)$. Consistency for $g_W$ implies that $g_W(Q^*)=\{i\}$ and for $h_W$ that $h_W(Q^*)=h_W(Q)$. 
		
		Finally, let $i'\in h_W(Q)\setminus g_W(Q)$ and let $W'$ denote a profile such that $h_W(Q')=\{i'\}$. Such a profile exists since $h_W$ is non-imposing. Now, by consistency, it holds for every integer $\ell\in\mathbb{N}$ that $h_W(\ell Q^*+Q')=\{i'\}$ and therefore also $g_W(\ell Q^*+Q')=\{i'\}$ because $g_W(\ell Q^*+Q')\subseteq h_W(\ell Q+Q')$. However, this conflicts with continuity, which states that there must be an integer $\ell'$ such that $g_W(\ell' Q^*+Q',)=\{i\}$. Hence, our initial assumption is wrong and $g_W(Q)=h_W(Q)$ for all $Q\in\mathbb{Q}^{d_1\times d_2}$.\medskip
		
		\textbf{Step 6:}
		Finally, we show that $g(A,W)$ can be represented approval voting $\mathit{AV}^W_v(A)$. For doing so, consider a profile $A\in\mathcal{A^*}$ and let $Q=N(A,W)$. Now, recall that every entry in $Q$ corresponds to $n(c,A,W,k,\ell)=|\{i\in N_A\colon c\in A_i\land |A_i\cap W|=k \land |A_i|=\ell \}|$ for some candidate $c\in\mathcal{C}\setminus W$ and $(k,\ell)\in Z_W$. Since our vector $s$ contains also an entry for $(k,\ell)\in Z_W$, there is a very natural weight vector $v^W(x,y)$: we set $v^W(x,y)=s_{x,y}$ for all $(x,y)\in Z_W$ and $0$ otherwise. Hence, we need to show that $g(A,W)$ contains precisely the candidates $c\in\mathcal{C}\setminus W$ that maximize the score $s_{v^W}(A,c)=\sum_{i\in N_A\colon c\in A_i} v^W(|A_i\cap W|, |A_i|)$.
		
		For proving this, observe that $s_{v^W}(A,c)$ is equivalent to $\sum_{(k,\ell)\in Z_W} n(c,A,W,k,\ell) v(k,\ell)=s \cdot N(A,W)_c$.
		Hence, it we derive that $g(A,W)=g_W(N(A,W))=\{c\in \mathcal{C}\setminus W\colon \forall x\in\mathcal{C}\setminus W\colon s\cdot N(A,W)_c\geq s\cdot N(A,W)_x\}=\{c\in \mathcal{C}\setminus W\colon \forall x\in\mathcal{C}\setminus W\colon s_{v^W}(A,c)\geq s_{v^W}(A,x)\}=\mathit{AV}_{v^W}$. This proves this proposition.
	\end{proof}
	
	\subsection{Proof of \Cref{prop:relation}}
	
	As last proposition, we show \Cref{prop:relation}. For this, we first investigate the basic properties of the considered rules and show that all sequential valuation rules are consistently committee monotone, and that all step-dependent sequential scoring rules are proper. Since step-dependent sequential scoring rules are valuation rules, this also proves that these rules are also consistently committee monotone. Analogous reasoning also entails that sequential Thiele rules and step-dependent sequential Thiele rules are proper rules. 
		
	\begin{lemma}\label{lem:basicAxioms}
		Every step-dependent sequential scoring rule is a proper ABC voting rule. Every sequential valuation rule is consistently committee monotone. 
	\end{lemma}
	\begin{proof}
		The lemma consists of five independent claims: every sequential valuation rule is consistently committee monotone and every step-dependent sequential scoring rule is anonymous, neutral, continuous, and non-imposing. We will prove each of these claims separately.\medskip
		
		\textbf{Claim 1: All sequential valuation rules are consistently committee monotone.}
		
		Let $f$ denote a sequential valuation rule and $v$ its corresponding valuation function. We will show by induction on the committee size $k\in \{0,\dots,m\}$ that $g(A,W)=\{x\in\mathcal{C}\setminus W\colon \forall y\in \mathcal{C}\setminus W: s_v(A,W^x)\geq s_v(A,W^y)\}$ consistently generates $f$. For this, let $f_g(A,k)=\{W\cup\{x\}\in\mathcal{W}_k\colon W\in f_g(A,k-1), x\in g(A,W)\}$ denote the function generated by $g$ and note that $f_g$ is indeed an ABC voting rule since $g$ is complete. For the induction basis, we observe that $f(A,0)=f_g(A,0)=\{\emptyset\}$ for all profiles $A$. Next, assume that $f(A,k)=f_g(A,k)$ for some profile $A$ and a committee size $k\in \{0,\dots, m-1\}$. Moreover, let $W\in f(A,k)$. By the definition of $f$, $W\cup \{x\}\in f(A,k+1)$ if and only if $s_v(A,W^x\geq s_v(A,W^y)$ for all $y\in \mathcal{C}\setminus W$. This means by definition that $x\in g(A,W)$ and $W\cup \{x\}\in f_g(A,k+1)$. Since this equivalence is true for all committees $W$ and candidates $x$, $f(A,k+1)=f_g(A,k+1)$ and $g$ thus generates $f$. 
		
		Next, we show that $g$ is consistent. Thus, consider two profiles $A$ and $A'$ and a committee $W$ such that $g(A,W)\cap g(A',W)\neq \emptyset$. By definition, $x\in g(A,W)\cap g(A',W)$ implies that $s_v(A,W^x)\geq s_v(A,W^y)$ and $s_v(A',W^x)\geq s_v(A',W^y)$ for all $y\in\mathcal{C}\setminus W$. Thus, we infer that $s_v(A+A',W^x)=s_v(A,W^x)+s_v(A',W\cup \{x\})\geq s_v(A,W^x)+s_v(A',W^x)=s_v(A+A',W^y)$ for all $y\in\mathcal{C}\setminus W$. Hence, $g(A,W)\cap g(A',W)\subseteq g(A+A',W)$. Conversely, if $y\not\in g(A,W)\cap g(A',W)$, we have $y\not\in g(A,W)$ or $y\not\in g(A',W)$. Without loss of generality, suppose that $y\not \in g(A,W)$ and let $x\in g(A,W)\cap g(A',W)$. Our assumptions entail that $s_v(A,W^x)> s_v(A,W^y)$ and $s_v(A',W^x)\geq s_v(A',W^y)$. Thus, $s_v(A+A',W^x)>s_v(A+A', W^y)$ which proves that $y\not\in g(A+A',W)$. We infer from this that $g(A+A',W)=g(A,W)\cap g(A',W)$, which proves that $g$ is consistent.\medskip
		
		\textbf{Claim 2: Every step-dependent sequential scoring rule is anonymous.}
		
		For this claim and all subsequent ones, let $f$ denote a step-dependent sequential scoring rule and let $h(x,y,z)$ denote its step-dependent counting function. We first show by an induction on the committee size $k\in \{0,\dots, m\}$ that $f$ is anonymous. Thus, consider a profile $A$ and a permutation $\pi:\mathbb{N}\rightarrow\mathbb{N}$, and note that $f(A,0)=f(\pi(A),0)=\{\emptyset\}$ by definition. Hence, assume that $f(A,k)=f(\pi(A),k)$ for some fixed $k\in \{0,\dots, m-1\}$ and consider $W\in f(A,k)$. By definition, $W\cup \{x\}\in f(A,k+1)$ for all candidates $x\in\mathcal{C}\setminus W$ that maximize $s_h(A,W^x)=\sum_{i\in N_A} h(|A_i\cap W^x|, |W^x|, |A_i|)$. Since $A'=\pi(A)$ is derived from $A$ only by permuting the voters, it follows immediately that $s_h(A',W^x)=s_h(A,W^x)$ for all $x\in\mathcal{C}\setminus W$ and thus $W\cup \{x\}\in f(A,k+1)$ if and only if $W\cup \{x\}\in f(A,k+1)$ for all committees $W\in f(A,k)$ and candidates $x\in \mathcal{C}\setminus W$. This shows that $f(A,k+1)=f(\pi(A),k+1)$ and we thus infer inductively that $f$ is anonymous.\medskip
		
		\textbf{Claim 3: Every step-dependent sequential scoring rule is neutral.}
		
		Let $f$ and $h$ be defined as in Claim 2. First, note that $|A|=|\tau(A)|$, $|W|=|\tau(W)|$, and $|W\cap A|=|\tau(W\cap A)|$ for every profile $A$, committee $W$, and permutation $\tau:\mathcal{C}\rightarrow\mathcal{C}$. Based on this fact, we now show inductively that $f$ is neutral. Consider for this an arbitrary profile $A$ and a permutation $\tau:\mathcal{C}\rightarrow\mathcal{C}$. Once again, the induction basis $k=0$ is trivial since $f(\tau(A),0)=\tau(f(A,0))=\{\emptyset\}$. Hence, assume that $f(\tau(A), k)=\tau(f(A,k))$ for some fixed $k\in \{0,\dots, m-1\}$ and let $W\in f(A,k)$. It follows from our initial observation that $s_h(A,W^x)=s_h(\tau(A),\tau(W^x))$ for every $x\in\mathcal{C}\setminus W$. Hence, if $x$ maximizes $s_h(A,W^x)$, then $\tau(x)$ maximizes the score $s_h(\tau(A), \tau(W)\cup \tau(x))$, so it follows that $W\cup \{x\}\in f(A,k+1)$ if and only if $\tau(W\cup \{x\})\in f(\tau(A), k+1)$. Equivalently, $f(\tau(A),k+1)=\tau(f(A,k+1))$, which proves the induction step.\medskip
		
		\textbf{Claim 4: Every step-dependent sequential scoring rule is continuous.}
		
		Let $f$ and $h$ be defined as in Claim 2. Moreover, consider two profiles $A,A'\in\mathcal{A}^*$, a committee size $k$, and a committee $W\in\mathcal{W}_k$ such that $f(A,k)=\{W\}$. For proving this claim, we define $F=\bigcup_{\ell=0}^k f(A,\ell)$ and $\Delta(A,X,x,y)=s_h(A, X^x)-s_h(A,X^y)$. In particular, observe that $\Delta(A,X,x,y)=-\Delta(A,X,y,x)$ and $\Delta(A+A',X,x,y)=\Delta(A,X,x,y)+\Delta(A',X,x,y)$ for all profiles $A$, $A'$, committees $X$, and candidates $x,y\in\mathcal{C}\setminus X$. Moreover, define 
		\begin{align*}
		    &\Delta^1=\min_{X\in F, c,d\in\mathcal{C}\setminus X\colon X^c\in F, X^d\not\in F} \Delta(A,X,c,d), \text{ and}\\
		    &\Delta^2=\max_{X\in F, c,d\in\mathcal{C}\setminus X\colon X^d\in F, X^c\not\in F} \Delta(A',X,c,d).
		\end{align*} 
		Finally, let $j\in\mathbb{N}$ denote the smallest integer such that $j\Delta^1>\Delta^2$. 
		
	We will show by induction on $\ell$ that $f(jA+A',\ell)\subseteq f(A,\ell)$ for all $\ell\in \{0,\dots, k\}$. This implies that $f(jA+A',k)=f(A,k)$ since $|f(A,k)|=1$ and $f(jA+A',k)\neq\emptyset$. The induction basis $\ell=0$ follows immediately from the definition of ABC voting rules. Hence, suppose that $f(jA+A',\ell)\subseteq f(jA+A',\ell)$ for some fixed $\ell\in\{0,\dots, k-1\}$ and consider a committee $W\in f(jA+A',\ell)$. For every $c,d\in\mathcal{C}\setminus W$ such that $W\cup \{c\}\in f(A,\ell+1)$, $W\cup\{d\}\not\in f(A, \ell+1)$, it follows that $\Delta(jA+A',W,c,d)=j\Delta(A,W,c,d)+\Delta(A',W,c,d)\geq j\Delta^1-\Delta^2>0$. By the definition of $\Delta$ this implies that $W\cup \{d\}\not\in f(jA+A',\ell+1)$ since $W\cup\{c\}$ has a higher score. Now, since $d$ is an arbitrary candidate such that $W\cup \{d\}\not \in f(A,\ell+1)$, it follows that $W\cup \{c\}\in f(jA+A',\ell+1)$ can only be true if $W\cup \{c\}\in f(A,\ell+1)$. This implies that $f(jA+A',\ell+1)\subseteq f(A,\ell+1)$ and thus proves the induction step.\medskip
		
		\textbf{Claim 5: Every step-dependent sequential scoring rule is non-imposing.}
		
		Let $f$ and $h$ be defined as in Claim 2. For this step, it is crucial that step-dependent counting functions $h(x,y,z)$ satisfy that for every $y\in \{1,\dots, m-1\}$, there is $x\leq y$ and $z\in \{x,\dots, m-1-y+x\}$ such that $h(x,y,z)\neq h(x-1,y,z)$. Furthermore, because of neutrality, it suffices to show that for every committee size $k$, there is a profile $A$ and a committee $W$ such that $f(A,k)=\{W\}$; every other committee of size $k$ can then be obtained by permuting $A$. As in all previous cases, we use an induction on $k$ and the induction basis $k=0$ follows by the definition of ABC voting rules. 
		
		Thus, assume that there is a $k\in \{1,\dots,m-1\}$ such that every committee $W$ with $|W|\leq k$ is uniquely chosen for some profile $A$, i.e., $f(A,|W|)=\{W\}$. Note that if $k=m-1$, we are already done as $f(A,m)=\{\mathcal{C}\}$ since $\mathcal{C}$ is the only committee of size $m$. We thus focus on the case that $k\leq m-2$. Now, using the induction hypothesis and the construction in the proof of \Cref{lem:seqNI}, we can construct for every sequence of committees $W_1,\dots, W_k$ a profile $A$ such that $f(A,\ell)=\{W_\ell\}$ for all $\ell\in \{1,\dots,k\}$ and $f(A,k+1)=\{W_k\cup \{x\}\colon x\in \mathcal{C}\setminus W_k\}$. 
		
		Our next goal is to construct a profile $A'$ for which there is a candidate $c\in\mathcal{C}\setminus W_k$ such that $s_h(A', W_k^c)>s_h(A', W_k^d)$ for all $d\in \mathcal{C}\setminus W_k^c$. For constructing this profile, let $x\leq k+1$ and $z\in \{x,\dots, m-1-y+x\}$ denote the integers such that $h(x,k+1,z)\neq h(x-1,k+1,z)$. We subsequently suppose that $h(x,k+1,z)> h(x-1,k+1,z)$ and explain at the end of this paragraph how to adapt our construction to the case that $h(x,k+1,z)< h(x-1,k+1,z)$. Now, let $A_i$ denote a ballot such that $|A_i\cap W_k|=x-1$ and $|A_i|=z$ and let $c,d\in\mathcal{C}\setminus W_k$ denote candidates such that $c\in A_i$, $d\not\in A_i$. Note that such a ballot exists due to our conditions on $x$ and $z$: $x\leq k+1$ ensures that $|A_i\cap W_k|=x-1$ is possible, $z\geq x$ ensures that we can approve at least $x$ candidates, and $z\leq m-1 -(k+1-x)$ ensures that we can disapprove at least $k+2-x$ candidates (namely $d$ and the remaining $k-(x-1)$ candidates in $W_{k}\setminus A_i$). Next, let $\tau$ denote an arbitrary permutation on $\mathcal{C}$ with $\tau(e)=e$ for all candidates $e\in W_k\cup \{c\}$. In particular, this means that $c\in \tau(A_i)$ and $\tau(A_i)\cap W_k=A_i\cap W_k$. Finally, we define the profile $A'$ by adding a voter with the ballot $\tau(A_i)$ for every such permutation. Since $c\in\tau(A_i)$ for all $\tau:\mathcal{C}\rightarrow\mathcal{C}$ with $\tau(e)=e$ for $e\in W_k^c$, it follows that $s_h(A',W_k^c)=(m-|W_k|-1)! \cdot h(x, k+1, z)$. On the other hand, every other candidate $e\in \mathcal{C}\setminus W_k^c$ is not approved in the ballots $\tau(A_i)$ in which $\tau$ maps $e$ to $d$. Hence, these voters approve $x-1$ candidates of $W_k^e$ and since $h(x, k+1, z)>h(x-1, k+1,z)$, it thus follows that $s_h(A', W_k^e)< (m-|W_k|-1)! h(x, k+1, z)=s_h(A', W_k^c)$ for all $e\in\mathcal{C}\setminus W_k^c$. For the case that $h(x, k+1, z)<h(x-1, k+1,z)$, we can apply the same construction with the role of $c$ and $d$ swapped in $A_i$. 
		
		Finally, we will construct a profile $A^*$ in which $f(A^*, k+1)=\{W_k\cup \{c\}\}$. By \Cref{lem:seqCon}, there is an integer $j$ such that $f(jA+A',\ell)=\{W_\ell\}$ for all $\ell\leq k$. In particular, this means that $f(jA+A', k)=\{W_k\}$. On the other hand, $f(A,k+1)=\{W_k\cup \{x\}\colon x\in \mathcal{C}\setminus W_k\}$, which implies that $s_h(A, W_k^d)=s_h(A, W_k^e)$ for all candidates $d,e\in \mathcal{C}\setminus W_k$. Clearly, the same also holds for $jA$ since we derive this profile by only copying $A$ multiple times. Finally, we have by construction that $s_h(A', W_k^c)> s_h(A', W_k^d)$ for all $d\in\mathcal{C}\setminus W_k^c$. Because $s_h(jA+A', W_k^e)=s_h(jA, W_k^e)+s_h(A', W_k^e)$ for all $e\in\mathcal{C}\setminus W_k$, we infer that $s_h(jA+A', W_k^c)>s_h(jA+A', W_k^d)$ for all $d\in\mathcal{C}\setminus W_k^c$. Thus, $f(jA+A', k+1)=\{W_k\cup \{c\}\}$. This proves that $f$ is non-imposing for committees of size $k+1$ as neutrality allows us to construct every outcome now. Hence, we inductively infer that $f$ is non-imposing for every committee size.
	\end{proof}	
	
	Next, we show the first claim of \Cref{prop:relation}: a sequential valuation rule is a step-dependent sequential scoring rule if and only if it is proper. 
	
	\begin{lemma}\label{lem:SeqValRules}
		A sequential valuation rule is a step-dependent sequential scoring rule if and only if it is proper.
	\end{lemma}
	\begin{proof}
		We have shown in \Cref{lem:basicAxioms} that every step-dependent sequential scoring rule is proper. Moreover, their definition immediately shows that every step-dependent sequential scoring rule is a sequential valuation rule. Hence, only the converse remains to be proven. For this, let $f$ denote a sequential valuation rule that is proper, and let $v$ denote its corresponding valuation function. We have shown in Claim 1 of \Cref{lem:basicAxioms} that the generator function $g(A,W)=\{x\in\mathcal{C}\setminus W\colon \forall y\in\mathcal{C}\setminus W: s_v(A,W^x)\geq s_v(A,W^y)\}$ is consistent, complete, and generates $f$. Hence, it follows from \Cref{lem:A+N} that $g$ is a proper generator function.
		
		We prove this that $f$ is a step-dependent sequential scoring rule in two steps. First, we show that $f$ is induced by a neutral valuation function $v^*$ (i.e., $v^*(A_i,W)=v^*(\tau(A_i), \tau(W))$ for all permutations $\tau$, ballots $A_i$, and committees $W$). Based on $v^*$, we will build as second step a step-dependent counting function $h$ that induces $f$.\medskip
		
		\textbf{Step 1: There is a valuation function $v^*$ that is neutral and induces $f$}
		
		For proving this claim, we define the valuation function $v_\tau(A_i,W)=v(\tau(A_i), \tau(W))$ and the generator function $g_\tau(A,W)=\{x\in\mathcal{C}\setminus W\colon \forall y\in \mathcal{C}\setminus W: s_{v_\tau}(A, W^x)\geq s_{v_\tau}(A, W^y\}$ for every permutation $\tau:\mathcal{C}\rightarrow\mathcal{C}$. In particular, observe that $g_\tau(A,W)=g(A,W)$ for all profiles $A$ and committees $W$ since 
		\begin{align*}
			g_{\tau}(A,W)&=\{x\in\mathcal{C}\setminus W\colon \forall y\in \mathcal{C}\setminus W:\\
			&\qquad\sum_{i\in N_A} v_\tau(A_i, W^x)\geq \sum_{i\in N_A} v_\tau(A_i, W^y)\}\\
			&=\{x\in\mathcal{C}\setminus W\colon \forall y\in \mathcal{C}\setminus W: \\
			&\qquad \sum_{i\in N_A} v(\tau(A_i), \tau(W^x))\geq \sum_{i\in N_A} v(\tau(A_i), \tau(W^y))\}\\
			&=\{\tau^{-1}(x')\colon x'\in\mathcal{C}\setminus \tau(W)\colon \forall y'\in\mathcal{C}\setminus \tau(W): \\
			&\qquad \sum_{i\in N_A} v(\tau(A_i), \tau(W)^{x'})\geq \sum_{i\in N_A} v(\tau(A_i), \tau(W)^{y'}))\}\\
			&=\{\tau^{-1}(x')\colon x'\in g(\tau(A), \tau(W))\}\\
			&=\tau^{-1}(g(\tau(A), \tau(W))\\
			&=g(A,W).
		\end{align*}
		
		Finally, consider the valuation function $v^*(A_i,W)=\sum_{\tau\in\Tau} v_\tau(A_i,W)$, where $\Tau$ is the set of all permutations on $\mathcal{C}$. Clearly, $v^*$ is neutral, i.e., $v^*(A_i,W)=v^*(\tau(A_i),\tau(W))$ for all permutations $\tau$, ballots $A_i$, and committees $W$. We prove next that $v^*$ also induces our sequential valuation function $f$. For doing so, consider an arbitrary profile $A$ and let $f^*(A,k)=\{W\cup\{c\}\colon W\in f^*(A,k-1), c\in\mathcal{C}\setminus W\colon \forall d\in\mathcal{C}\setminus W\colon s_{v^*}(A, W^c)\geq s_{v^*}(A, W^d)\}$ for all committee sizes $k$. We will show by induction on $k$ that $f(A,k)=f^*(A,k)$ for $k\in \{0,\dots, m\}$. The induction basis $k=0$ follows trivially since $f(A,0)=\{\emptyset\}=f^*(A,0)$ by definition. 
		
		Now, assume that $f(A,k)=f^*(A,k)$ for some fixed $k\in \{0,\dots, m-1\}$ and let $W\in f(A,k)$. For each $c\in\mathcal{C}\setminus W$, we have that $W\cup \{c\}\in f(A,k+1)$ if and only if $c\in g(A,W)$. This is, in turn, equivalent to $c\in g_{\tau}(A,W)$ for every permutation $\tau:\mathcal{C}\rightarrow\mathcal{C}$. Hence, we have that $W\cup \{c\}\in f(A,k+1)$ if and only if $s_{v_\tau}(A,W^c)\geq s_{v_\tau}(A,W^d)$ for all permutations $\tau\in\Tau$ and $d\in\mathcal{C}\setminus W_k$. Since $s_{v^*}(A, W^c)=\sum_{\tau\in\Tau} s_{v_\tau}(A, W^c)$ for all $c\in\mathcal{C}\setminus W_k$, we infer that $W\cup \{c\}\in f(A,k+1)$ if and only if $W\cup \{c\}\in f^*(A,k+1)$. Finally, since $W\in f(A,k)$ is chosen arbitrarily, this shows that $f(A,k+1)=f^*(A,k+1)$ and thus proves that $v^*$ induces $f$.\medskip
		
		\textbf{Step 2: There is a step-dependent counting function $h(x,y,z)$ that induces $f$.}
		
		As second step, we show that $f$ is a step-dependent sequential scoring rule by proving that it is induced by a step-dependent counting function. For this, note first that there is a function $h(x,y,z)$ such that $v^*(A_i, W)=h(|A_i\cap W|, |W|, |A_i|)$ for all ballots $A_i$ and committees $W$. For proving this claim, consider two arbitrary ballots $A_i$ and $A_i'$ and committees $W$ and $W'$ such that $|A_i\cap W|=|A_i'\cap W'|$, $|A_i|=|A_i'|$, and $|W|=|W'|$. Clearly, there is a permutation $\tau$ such that $\tau(A_i)=A_i'$, $\tau(A_i\cap W)=A_i'\cap W'$, and $\tau(W)=W'$. Hence, the neutrality of $v^*$ shows that $v^*(A_i, W)=v^*(A_i', W')$. Consequently, $v^*$ only depends on $|A_i\cap W|$, $|W|$, and $|A_i|$ to compute the score, i.e., there is a function $h(x,y,z)$ such that $v^*(A_i, W)=h(|A_i\cap W|, |W|, |A_i|)$. 
		
		Finally, we show that $h(x,y,z)$ must be a step-dependent counting function, which requires that there is for every $y\in \{1,\dots, m-1\}$, an $x\leq y$ and $z\in \{x,\dots, m-1-y+x\}$ such that $h(x,y,z)\neq h(x-1,y,z)$. Assume for contradiction that $h$ fails this condition, i.e., there is $y\in \{1,\dots, m-1\}$ such that for each $x\in \{1,\dots,y\}$ and $z\in \{x,\dots, m-1-y+x\}$, $h(x,y,z)=h(x-1,y,z)$. Our goal is to show that $f$ has to fail non-imposition, which contradicts that it is a proper ABC voting rule. For this, consider an arbitrary ballot $A_i$, a committee $W\in\mathcal{W}_{y-1}$, and two candidates $c,d\in\mathcal{C}\setminus W$. Our goal is to show that $h(|A_i\cap W^c|,|W^c|,|A_i|)=h(|A_i\cap W^d|, |W^d|, |A_i|)$. If $|A_i\cap W^c|=|A_i\cap W^d|$, this is clear since all arguments of the left and right side of $h$ are identical. Hence, we assume with out loss of generality that $|A_i\cap W^c|>|A_i\cap W^d|$. In particular, this means that $c\in A_i$, $d\not\in A_i$ and $|A_i\cap W^c|=|A_i\cap W^d|+1$. Now, note that $x=|A_i\cap W^c|\leq |W^c|=y$ and $x=|A_i\cap W^c|\leq |A_i|=z$. Moreover, we know that $d\not\in A_i$ and $y-x$ candidates of $W^c$ are not in $A_i$. Hence, $z=|A_i|\leq m-1-(y-x)$. Therefore, our contradiction assumption indeed shows that $(|A_i\cap W^c|,|W^c|,|A_i|)=h(|A_i\cap W^d|, |W^d|, |A_i|)$. 
		
		Since $W$, $c$, $d$, and $A_i$ are chosen arbitrarily, it follows that each committee of size $y\leq m-1$ gets the same amount of points from every ballot. Hence, $|f(A,y)|\neq 1$ for all profiles $A$ since $f(A,y)=\{W\cup \{y\}\colon W\in f(A,y-1), x\in \mathcal{C}\setminus W\}$ and $|\mathcal{C}\setminus W|\geq 2$. This contradicts the non-imposition of $f$. Our assumption that $h$ is not a step-dependent counting function must therefore be wrong, and $f$ is thus a step-dependent sequential scoring rule.
	\end{proof}	
	
	Next, we show that independence of losers characterizes step-dependent sequential Thiele rules within the class of step-dependent sequential scoring rules.
		
		\begin{lemma}\label{lem:StepSeqScoreRule1}
			A step-dependent sequential scoring rule is a step-dependent sequential Thiele rule if and only if it is independent of losers. 
		\end{lemma}
		\begin{proof}
			We proof both directions of the lemma separately.\medskip
			
			\textbf{Claim 1: A step-dependent sequential Thiele rule is a step-dependent sequential scoring rule that satisfies independence of losers.}
			
			First, note that every step-dependent sequential Thiele rule $f$ is clearly a step-dependent sequential scoring rule since $f$ is induced by a step-dependent Thiele counting function $h(|A_i\cap W|,|W|)$. Clearly, we can extend $h$ to a step-dependent counting function $h'$ by $h'(x,y,z)=h(x,y)$ for all $x,y$. Moreover, since for every $y\in \{1,\dots, m-1\}$, there is $x\leq y$ such that $h(x,y)> h(x-1,y)$, it follows that $h'(x,y,z)\neq h(x-1,y,z)$ for this $x$ and all $z\in \{1,\dots, m\}$. Hence, $f$ is a step-dependent sequential scoring rule. 
			
			Moreover, it is easy to show that $f$ is independent of losers. 
			For doing so, consider two profiles $A$ and $A'$ on the electorate $N_A$, a committee $W\in f(A,|W|)$, a voter $i\in N_A$, and a candidate $c\in A_i\setminus W$ such that $A'$ $A_i'=A_i\setminus \{c\}$ and $A_j=A_j'$ for all $j\in N\setminus \{i\}$. The last assumption entails that $h(|A_j\cap W'|, |W'|)=h(|A_j'\cap W'|, |W'|)$ for all voters $j\in N_A\setminus \{i\}$ and committees $W'\in\mathcal{W}$. 
			On the other hand, we have for every committee $W'$ that $h(|A_i\cap W'|, |W'|)=h(|A_i'\cap W'|, |W'|)$ if $c\not\in W$ and $h(|A_i\cap W'|, |W'|)\geq h(|A_i'\cap W'|, |W'|)$ if $c\in W$ since $h$ is non-decreasing in the first argument. 
			Finally, since $W\in f(A,|W|)$, there is a sequence of candidates $\{x_1, dots, x_k\}$ such that for all $\ell\in\{1,\dots,k\}$, it holds that $W_\ell=\{x_1,\dots, x_\ell\}\in f(A,\ell)$ and $s_h(A,W_{\ell-1}^{x_{\ell}})\geq s_h(A,W_{\ell-1}^y)$ for all $y\in \mathcal{C}\setminus W_{\ell-1}$. 
			Because $c\not\in W$, our previous insights show that the scores of the committees $W_\ell$ are always maximal, and thus $W\in f(A',k)$. Hence $f$ satisfies independence of losers.\medskip
			
			\textbf{Claim 2: Every step-dependent sequential scoring rule that satisfies independence of losers is a step-dependent sequential Thiele rule.}
			
			Next, we show the converse and consider thus a step-dependent sequential scoring rule $f$ that is independent of losers. Moreover, let $h(x,y,z)$ denote its step-dependent sequential scoring function. We will show that the function $\bar h$ defined by $\bar h(0,y)=0$, $\bar h(x,y)=\bar h(x-1, y) + h(x,y,x)-h(x-1,y,x)$ for $x\in \{1,\dots,y\}$, and $\bar h(x,y)=\bar h(y,y)$ for $x>y$ is a step-dependent Thiele counting function that also induces $f$. For this, we proceed in multiple steps. First, we prove an auxiliary claim stating that we can build profiles with specific outcomes, which will be used for the next steps. Then, we show that for all profiles $A$, committees $W$, and candidates $c,d\in\mathcal{C}\setminus W$, it holds that $s_h(A, W^c)\geq s_{h}(A, W^d)$ if and only if $s_{\bar h}(A, W^c)\geq s_{\bar h}(A, W^d)$. This clearly implies that $\bar h$ also induces $f$. As last step, we show that $\bar h$ is non-decreasing in $x$ (and thus also non-negative since $\bar h(0,y)=0$ for all $y$) and satisfies that there is $x\leq y$ such that $\bar h(x,y)>\bar h(x-1,y)$.\medskip
						
			\emph{Step 1: For all committees $W$ and distinct $c,d\in\mathcal{C}\setminus W$, there is a profile $A$ such that $f(A, |W|)=\{W\}$ and $f(A,|W|+1)=\{W^c, W^d\}$.}
			
			For proving this claim, let $W_1, \dots, W_k$ denote a sequence of committees such that $W_k=W^c$ and $W_{k-1}=W$. By \Cref{lem:seqNI}, there is a profile $\bar A$ such that $f(\bar A, \ell)=\{W_\ell\}$ for all $\ell\leq k$. Next, consider a permutation $\tau:\mathcal{C}\rightarrow\mathcal{C}$ such that $\tau(x)=x$ for all $x\in W^c$. By neutrality, we have that $f(\tau(\bar A), \ell)=\{W_\ell\}$ for all $\ell\in \{1,\dots,k\}$ and \Cref{lem:merge} implies thus that $f(\bar A+\tau(\bar A), \ell)=\{W_\ell\}$, too. 
			Now, let $\hat A^c$ denote the profile consisting of $\tau(\bar A)$ for every permutation $\tau:\mathcal{C}\rightarrow\mathcal{C}$ with $\tau(x)=x$ for all $x\in W^c$. By the same argument as before, $f(\hat A^c, \ell)=\{W_\ell\}$ for all $\ell\in \{1,\dots, k\}$. In particular, this means that $s_h(\hat A^c, W^c)>s_h(\hat A^c, W^x)$ for all $x\in\mathcal{C}\setminus W^c$. Moreover, we claim that $s_h(\hat A^c, W^x)=s_h(\hat A^c, W^y)$ for all $x,y\in\mathcal{C}\setminus W^c$. For this, note that $s_h(\bar A, W^x)=s_h(\tau(\bar A), W^{\tau(x)}\}$ for all permutations $\tau$ that only reorder the candidates in $\mathcal{C}\setminus W^c$. Since $\hat A^c$ consists of these profiles for all permutations, all of the candidates except $c$ must have the same score. 
			
			Next, consider the profile $\hat A^d$ derived from $\hat A^c$ by exchanging $c$ and $d$; formally $\hat A^d=\tau_{cd}(\hat A^c$) where $\tau_{cd}$ is the permutation that only swaps $c$ and $d$. By neutrality, we have that $f(\hat A^d, \ell)=\{W_\ell\}$ for $\ell\in \{1,\dots, k-1\}$ and $f(\hat A^d, k)=\{W^d\}$. Moreover, we have that $s_h(\hat A^d, W^d)=s_h(\hat A^c, W^c)$ and $s_h(\hat A^d, W^x)=s_h(\hat A^c, W^y)$ for all $x\in\mathcal{C}\setminus W^d$, $y\in\mathcal{C}\setminus W^d$. Finally, we define $\hat A=\hat A^d+\hat A^c$. By \Cref{lem:merge}, we immediately get that $f(\hat A, \ell)=\{W_\ell\}$ for all $\ell\in \{1,\dots, k-1\}$. Furthermore, it holds that $s_h(\hat A,W^x)=s_h(\hat A^c,W^x)+s_h(\hat A^d,W^x)$ for all $x\in\mathcal{C}\setminus W$. This implies that $s_h(\hat A, W^c)=s_h(\hat A, W^d)>s_h(\hat A, W^x)$ for all $x\in\mathcal{C}\setminus (W\cup \{c,d\})$. Hence, $\hat A$ indeed satisfies our requirements.\medskip
			
			\emph{Step 2: It holds for all profiles $A$, committees $W$, and candidates $c,d\in \mathcal{C}\setminus W$ that $s_{h}(A, W^c)\geq s_{h}(A, W^d)$ if and only if $s_{\bar h}(A, W^c)\geq s_{\bar h}(A, W^d)$.}
			
			For proving this step, we consider a committee $W$ and two candidates $c,d\in\mathcal{C}\setminus W$. Since $s_{h}$ and $s_{\bar h}$ only sum up the scores of the individual ballots, it suffices to focus on the ballots $A_i$. Our goal is thus to show that $h(|A_i\cap W^c|, |W^c|, |A_i|)-h(|A_i\cap W^d|, |W^d|, |A_i|)=\bar h (|A_i\cap W^c|, |W^c|)-\bar h(|A_i\cap W^d|, |W^d|)$. First, if $c,d\in A_i$ or $c,d\not\in A_i$, this follows immediately since $|A_i\cap W^c|=|A_i\cap W^d|$ and thus, 
			\begin{align*}
				&h(|A_i\cap W^c|, |W^c|, |A_i|)-h(|A_i\cap W^d|, |W^d|, |A_i|)\\
				&=\bar h (|A_i\cap W^c|, |W^c|)-\bar h(|A_i\cap W^d|, |W^d|)=0.
			\end{align*}
			
			Hence, suppose that $c\in A_i$, $d\not\in A_i$; the case that $d\in A_i$, $c\not\in A_i$ is symmetric. This implies that $|A_i\cap W^d|=|A_i\cap W^c|-1$ and we hence have to show that 
			\thickmuskip=0.5\thickmuskip
			\medmuskip=0.5\medmuskip
			\begin{align*}
				&h(|A_i\cap W^c|, |W^c|, |A_i|)-h(|A_i\cap W^c|-1, |W^c|, |A_i|)\\
				&=\bar h(|A_i\cap W^c|, |W^c|)-\bar h(|A_i\cap W^c|-1, |W^c|)\\
				&=h(|A_i\cap W^c|, |W^c|, |A_i\cap W^c|)-h(|A_i\cap W^c|-1, |W^c|, |A_i\cap W^c|).
			\end{align*}
			\thickmuskip=2\thickmuskip
            \medmuskip=2\medmuskip
			
			If $|A_i|=|A_i\cap W^c|$, this claim holds trivially. Since $|A_i|\geq |A_i\cap W^c|$, we thus suppose that $|A_i|>|A^i\cap W^c|$, which implies that $X=A_i\setminus W^c\neq \emptyset$. Also, recall that $d\not\in A_i$ and thus $d\not\in X$. Next, let $A$ denote a profile such that $f(A, |W|)=\{W\}$ and $f(A, |W|+1)=\{W^c, W^d\}$; such a profile exists due to Step 1. Moreover, let $A'$ be a profile consisting of two voters, one of which reports $A_i$ and the other one reports $(A_i\setminus \{c\})\cup \{d\}$. Finally, let $A''$ denote the profile derived from $A'$ by assigning the ballot $A_i\setminus X=A_i\cap W^c$ to the voter who originally submits $A_i$. For these profiles, it holds that
			\thickmuskip=0.5\thickmuskip
			\medmuskip=0.5\medmuskip
			\allowdisplaybreaks
			\begin{align*}
			    s_h(A, W^c)&=s_h(A,W^d)>s_h(A,W^x) \text{ for all } x\in\mathcal{C}\setminus (W\cup \{c,d\}),\\
			    s_h(A',W^c)&=h(|A_i\cap W^c|, |W^c|, |A_i|)+h(|A_i\cap W^c|-1, |W^c|, |A_i|)\\
			    &=s_h(A', W^d), \\
			    s_h(A'',W^c)&=h(|A_i\cap W^c|, |W^c|, |A_i\cap W^c|)\\
			    &\quad+h(|A_i\cap W^c|-1, |W^c|, |A_i|),\qquad\text{ and}\\
			    s_h(A'', W^d)&=h(|A_i\cap W^c|-1, |W^c|, |A_i\cap W^c|)\\
			    &\quad+h(|A_i\cap W^c|, |W^c|, |A_i|).
			\end{align*}
			\thickmuskip=2\thickmuskip
            \medmuskip=2\medmuskip
			
			Finally, it is not difficult to see that there is an integer $j$ such that $f(jA+A', |W|)=(jA+A'', |W|)=\{W\}$, $f(jA+A', |W|+1)\subseteq \{W^c, W^d\}$, and $(jA+A'', |W|+1)\subseteq \{W^c, W^d\}$. From the above equations, it now follows that $f(jA+A', |W|+1)=\{W^c, W^d\}$. On the other hand, since $A''$ is derived from $A'$ by only disapproving candidates $x\in\mathcal{C}\setminus (W\cup \{c,d\})$, a repeated application of independence of losers shows that $f(jA+A'', |W|+1)=\{W^c, W^d\}$, too. This implies that $s_h(A'',W^c)=s_h(A'', W^d)$, which is equivalent to
			\thickmuskip=0.5\thickmuskip
			\medmuskip=0.5\medmuskip
			\begin{align*}
			& = h(|A_i\cap W^c|, |W^c|, |A_i|)-h(|A_i\cap W^c|-1, |W^c|, |A_i|)\\
			&=h(|A_i\cap W^c|, |W^c|, |A_i\cap W^c|)-h(|A_i\cap W^c|-1, |W^c|, |A_i\cap W^c|).
			\end{align*}
			\thickmuskip=2\thickmuskip
            \medmuskip=2\medmuskip
			
			This shows that $s_h(A, W^c)\geq s_h(A, W^d)$ if and only if $s_{\bar h}(A, W^c)\geq s_{\bar h}(A, W^d)$ for all profiles $A$, committees $W$, and candidates $c,d\in\mathcal{C}\setminus W$. Hence, $\bar h$ indeed induces $f$.\medskip
			
			\emph{Step 3: $\bar h$ is a step-dependent Thiele counting rule.}
			
			For proving this claim, we need to show that $\bar h(x,y)$ is non-negative, non-decreasing in $x$, and satisfies that for every $y\in \{1,\dots,m-1\}$ there is $x\in \{1,\dots, y\}$ such that $\bar h(x,y)>\bar h(x-1,y)$. We start by showing that $\bar h$ is non-decreasing in $x$. Note that this immediately implies that it is non-negative since we defined that $\bar h(0,y)=0$ for all $y\in \{1,\dots, m\}$. Hence, assume for contradiction that there is $x\in \{0,\dots, m-1\}$ and $y\in \{1,\dots, m\}$ such that $\bar h(x,y)>\bar h(x+1,y)$. Since $\bar h(x',y')=\bar h(y', y')$ for all $x', y'$ with $x\geq y'$, our assumption requires that $x<y$. Moreover, we suppose that $y<m$ because we can redefine $\bar h(x,m)=0$ for all $x$. The reason for this is that $\bar h(x,m)$ is only queried when $f$ needs to decide on a winning committee of size $m$, but there is only one such committee and thus the values of $\bar h$ do not matter. 
			
			Now, let $W_1, \dots, W_{y-1}$ denote an arbitrary sequence of committees with length $y-1$. By \Cref{lem:seqNI}, there is a profile $A$ such that $f(A, \ell)=W_{\ell}$ for all $\ell\in \{1,\dots,y-1\}$ and $f(A, y)=\{W_{y-1}\cup \{x\}\colon x\in\mathcal{C}\setminus W_{y-1}\}$. Furthermore, let $A'$ denote a profile in which a single voter approves $x$ candidates of $W_{y-1}$ and all candidates in $\mathcal{C}\setminus W_{y-1}$. Finally, define $A''$ as the profile derived from $A'$ by letting the single voter disapprove one candidate $c\in\mathcal{C}\setminus W_{y-1}$. 
			
			It is again not difficult to see that there is an integer $j$ such that $f(jA+A', \ell)=f(jA+A'', \ell)=\{W_\ell\}$ for all $\ell\in \{1,\dots, y-1\}$. On the other hand, we have that $s_{\bar h}(A, W_{y-1}^c)=s_{\bar h}(A, W_{y-1}^d)$ for all $c,d\in\mathcal{C}\setminus W_{y-1}$ since $f(A, y)=\{W_{y-1}\cup \{x\}\colon x\in\mathcal{C}\setminus W_{y-1}\}$. Also, it holds that $s_{\bar h}(A',W_{y-1}^c)=\bar h(x+1, y)=s_{\bar h}(A', W_{y-1}^d)$ for all $c,d\in\mathcal{C}\setminus W_{y-1}$. Finally, we have that $s_h(A'', W_{y-1}^c)=\bar h(x,y)$ and $s_h(A'', W_{y-1}^d)=\bar h(x+1,y)$ and thus, $s_h(A'', W_{y-1}^c)>s_h(A'', W_{y-1}^d)$ for all $d\in\mathcal{C}\setminus W_{y-1}^c$. Thus, we infer that $f(jA+A', y)=\{W_{y-1}\cup \{x\}\colon x\in\mathcal{C}\setminus W_{y-1}\}$ and $f(jA+A'',y)=\{W_{y-1}\cup \{c\}\}$. However, since $y<m$ and thus $|W_{y-1}|\leq m-2$, there is $d\in\mathcal{C}\setminus W_{y-1}^c$ such that $W_{y-1}\cup \{d\}\in f(jA+A', y)$. Because $c\not\in W_{y-1}\cup \{d\}$, independence of losers thus requires that $W_{y-1}\cup \{d\}\in f(jA+A'', y)$. This contradicts our previous insight, and thus, the assumption that $\bar h(x,y)>\bar h(x+1,y)$ must have been wrong. Hence, $\bar h$ is non-decreasing in $x$.
			
			Finally, we show that for every $y\in \{1,\dots, m-1\}$, there is $x\in \{1,\dots, y\}$ such that $h(x,y)>h(x-1,y)$. Assume there is $y$ such that this is not the case, which means that $h(x,y)=h(x',y)$ for all $x,x'\leq y$. Since for $|A_i\cap W|\leq |W|$ for all committees $W$ and ballots $A_i$, this means that all committees of size $y$ always receive the same score and, as there are at least two committees of size $y<m$, $f$ can therefore be not non-imposing. This contradicts our assumptions and we thus conclude that $f$ is a step-dependent sequential Thiele rule because it is induced by a step-dependent Thiele counting function $\bar h$.
		\end{proof}
		
	Finally, we prove that committee separability characterizes the sequential Thiele rules within the class of step-dependent sequential Thiele rules.
	
	\begin{lemma}\label{lem:StepSeqThieleRule}
		A step-dependent sequential Thiele rule is a sequential Thiele rule if and only if it is committee separable. 
	\end{lemma}
	\begin{proof}
		For proving this lemma, we need to show two claims: every sequential Thiele rule is a step-dependent sequential Thiele rule that satisfies committee separability, and every step-dependent sequential Thiele rule that satisfies committee separability is a sequential Thiele rule.\medskip
		
		\textbf{Claim 1: Every sequential Thiele rule is a step-dependent sequential Thiele rule that satisfies committee separability.}
		
		Let $f$ denote a sequential Thiele rule and let $h(x)$ denote its Thiele counting rule.
		Clearly, $f$ is a step-dependent sequential Thiele rule as we can define the step-dependent Thiele counting function $h'(x,y)=h(x)$ for all $y\in \{1,\dots, m\}$. 
		
		Hence, it remains to show that $f$ is committee separable. For doing so, recall that $C_A=\bigcup_{i\in N_A} A_i$ and $C_B=\bigcup_{i\in N_B} B_i$ and let $A$ and $B$ denote two disjoint profiles with $C_A=\mathcal{C}\setminus C_B$. We will show by an induction on the committee size $k\in \{0,\dots, m\}$ that for every $W\in f(A+B,k)$, it holds that $W\cap C_A\in f(A, |W\cap C_A|)$ and $W\cap C_B\in f(B, |W\cap C_B|)$. The induction basis $k=0$ is trivial since $f(A+B,0)=f(A,0)=f(B,0)=\{\emptyset\}$. 
		
		Hence, suppose that our claim is true for a fixed $k\in \{0,\dots, m-1\}$ and let $W\in f(A+B,k+1)$. By the definition of sequential Thiele rules, there is a committee $W'\in f(A+B,k)$ and a candidate $c\in \mathcal{C}\setminus W'$ such that $W=W'\cup \{c\}$ and $s_h(A+B, W'\cup \{c\})\geq s_h(A+B, W'\cup\{x\})$ for every $x\in \mathcal{C}\setminus W'$. Subsequently, we suppose that $c\in C_A$; the case that $c\in C_B$ is symmetric. Now, our induction hypothesis shows that $W'\cap C_A\in f(A, |W'\cap C_A|)$ and $W'\cap C_B\in f(B, |W'\cap C_B|)$. Our goal is to prove that $W\cap C_A=(W'\cup \{c\})\cap C_A\in f(A, |W\cap C_A|)$. Since $W\cap C_B=W'\cap C_B\in f(B, |W\cap C_B|)$, this proves our claim. 
		Hence, assume for contradiction that $W\cap C_A\not\in f(A, |W\cap C_A|)$, which means that there is $d\in \mathcal{C}\setminus (W\cap C_A)$ such that $s_h(A, (W'\cap C_A)\cup \{d\})>s_h(A, W\cap C_A)$. In particular, this means that $d\in C_A$ because $h$ is non-decreasing and thus $s_h(A, W\cap C_A)\geq s_h(A, W'\cap C_A)=s_h(A,(W'\cap C_A)\cup \{x\})$ if $x\not\in C_A$. In turn, this implies that $s_h(B, W'\cup \{c\})=s_h(B, W'\cap C_B)=s_h(B, W'\cup \{d\})$ because no voter in $B$ approves $c$ or $d$. However, we thus infer that $s_h(A+B, W'\cup \{d\})>s_h(A+B, W'\cup \{c\})$ since
		\allowdisplaybreaks
		\begin{align*}
		    s_h(A+B, W'\cup \{d\})&=s_h(A, W'\cup\{d\})+s_h(B, W'\cup \{d\})\\
		    &=s_h(A, W'\cap C_A\cup \{d\})+s_h(B, W'\cap C_B)\\
		    &>s_h(A, W'\cap C_A \{c\})+s_h(B, W'\cap C_B)\\
		    &=s_h(A, W'\cup \{c\})+s_h(B, W'\cup \{c\})\\
		    &=s_h(A+B, W'\cup\{c\}). 
		\end{align*}
		
		This contradicts our assumptions on $W'$ and $c$. Hence, $W_A=W_A'\cup \{c\}\in f(A, |W_A|)$, $W_B=W_B'\in f(B, |W_B|)$, $W=W_A\cup W_B$, and $W_A\cap W_B=\emptyset$, which proves the induction step for this case.
		\medskip
		
		\textbf{Claim 2: Every step-dependent sequential Thiele rule that satisfies committee separability is a sequential Thiele rule.}
		
		Let $f$ denote a step-dependent sequential Thiele rule that satisfies committee separability and let $h(x,y)$ denote its step-dependent Thiele counting function.
		For proving that $f$ is a sequential Thiele rule, we proceed in several steps: firstly, we will derive another step-dependent Thiele counting function $h'$ such that $h'(0,y)=0$ and $h'(1,y)=1$ for all $y\in\{1,\dots,m-1\}$ that also induces $f$. Based on $h'$, we will then show that the function $\bar h$ with $\bar h(0)=0$ and $\bar h(x)=\bar h(x-1)+h'(x,x)-h'(x-1,x)$ for $x>0$ also induces $f$. Finally, we prove that $\bar h$ is a Thiele counting function, which proves that $f$ is a sequential Thiele rule.\medskip
		
		\emph{Step 1:} First, we derive a step-dependent Thiele counting function $h'$ such that $h'(0,y)=0$ and $h'(1,y)=1$ for all $y$ that induces $f$. For this, we first prove that $h(1,y)> h(0,y)$ for all $y\in \{1,\dots, m-1\}$. Thus, note that $h(1,1)>h(0,1)$ due to the assumption that for every $y\in \{1,\dots, m-1\}$, there is $x\in \{1,\dots. y\}$ such that $h(x,y)>h(x-1,y)$. Next, consider an arbitrary $k\in \{1,\dots, m-2\}$, let $W_1,\dots, W_k$ denote a sequence of committees and let $A$ denote a profile such that $f(A,\ell)=\{W_\ell\}$ for all $\ell\in \{1,\dots, k\}$; such a profile exists due to \Cref{lem:seqNI}. Now, consider the profile $A'$ derived from $A$ by assigning all voters the ballot $A_i\cap W_k$. It still holds that $f(A',\ell)=\{W_{\ell}\}$ for all $\ell\in \{1,\dots,k\}$ because $h(x,y)$ is non-decreasing in $x$ and independent of the size of the ballot. In more detail, it holds for every ballot $A_i$ and committees $W,W'$ with $|W|=|W'|$ and $W\subseteq W_k$ that $h(|A_i\cap W_k\cap W|, |W|)=h(|A_i\cap W|, |W|)$ and $h(|A_i\cap W_k\cap W'|, |W'|)\leq h(|A_i\cap W'|, |W'|)$. Hence, if a subset of $W_k$ has maximal score in $A$, it also has maximal score in $A'$. Moreover, note that all candidates in $\mathcal{C}\setminus W_k$ are disapproved by all voters in $A'$ and thus, neutrality requires that $f(A',k+1)=\{W_k^x\colon x\in\mathcal{C}\setminus W_k\}$.
		 
		Next, let $B$ denote the profile in which a candidate $c\in \mathcal{C}\setminus W_k$ is uniquely approved by two voters and every other candidate $d\in\mathcal{C}\setminus W_k$, $d\neq c$, is approved by by a single voter. It is straightforward that $f(B, 1)=\{\{c\}\}$. Finally, \Cref{lem:seqCon} shows that there is an integer $j$ such that $f(jA'+B, k)=\{W_k\}$. Thus, committee monotonicity requires that $f(jA'+B, k+1)\subseteq \{W_k\cup \{x\}\colon x\in\mathcal{C}\setminus W_k\}$. In turn, committee separability requires for every $W\in f(jA'+B,k+1)$ that $W\setminus W_k\in f(B, |W\setminus W_k|)=f(B,1)$. Since $f(B,1)=\{\{c\}\}$, this means that $f(jA'+B,k+1)=\{W_k\cup \{c\}\}$. We infer from this now that $h(1,k+1)>h(0,k+1)$ as otherwise $h(1,k+1)=h(0,k+1)$ and $f(jA'+B, k+1)=\{W_k\cup \{x\}\colon x\in\mathcal{C}\setminus W_k\}$.
		 
		It is now easy to see that the function $h'$ defined by $h'(x,y)=\frac{h(x,y)-h(0,y)}{h(1,y)-h(0,y)}$ is a step-dependent Thiele counting function with $h'(0,y)=0$ and $h'(1,y)=1$. In particular, it immediately follows that $s_h(A, W^c)\geq s_h(A, W^d)$ if and only if $s_{h'}(A,W^c)\geq s_{h'}(A, W^d)$ for all profiles $A$, committees $W$, and candidates $c,d\in\mathcal{C}\setminus W$. Thus, $h'$ also induces $f$. 
		
		\emph{Step 2}: Our next goal is to show that the function $\bar h(x)$ defined by $\bar h(0)=0$ and $\bar h(x)=\bar h(x-1)+h'(x,x)+h'(x-1,x)$ induces $f$. For doing so, we show that $s_{h'}(A, W^c)\geq s_{h'}(A, W^d)$ if and only if $s_{\bar h}(A, W^c)\geq s_{\bar h}(A, W^d)$ for all profiles $A$, a committees $W$ and candidates $c,d\in\mathcal{C}\setminus W$. This is equivalent to proving that $h'(|A_i\cap W^c|, |W^c|)-h'(|A_i\cap W^c|, |W^d|)= \bar h(|A_i\cap W^c|)-\bar h(|A_i\cap W^d|)$ for all committees $W$, ballots $A_i$, and candidates $c,d\in\mathcal{C}\setminus W$. Now, first observe that if $|W^c\cap A_i|=|W^d\cap A_i|$, then 
		\begin{align*}
			&h'(|A_i\cap W^c|, |W^c|)-h'(|A_i\cap W^d|, |W^d|)\\
			&=\bar h(|A_i\cap W^c|)-\bar h(|A_i\cap W^d|)=0.
		\end{align*}
		
		Hence, our equality holds in this case, and we subsequently suppose that $c\in A_i$, $d\not\in A_i$; the case that $d\in A_i$, $c\not\in A_i$ is symmetric. This assumption means that $|W^d\cap A_i|=|W^c\cap A_i|-1$ and we thus need to show that
		\begin{align*}
			&h'(|A_i\cap  W^c|, |W^c|)-h'(|A_i\cap W^c|-1, |W^c|)\\
			&= \bar h(|A_i\cap W^c|)-\bar h(|A_i\cap W^c|-1)\\
			&= h'(|A_i\cap W^c|, |A_i\cap W^c|)- h'(|A_i\cap W^c|-1, |A_i\cap W^c|).
		\end{align*}
		
		Now, assume for contradiction that this equality is not true, i.e., $h'(|A_i\cap W^c|, |W^c|)-h'(|A_i\cap W^c|-1, |W^c|)\neq h'(|A_i'\cap W^c|, |A_i\cap W^c|)-h(|A_i'\cap W^c|-1, |A_i\cap W^c|)$ for some ballot $A_i$, committee $W$ with $|W|\leq m-2$, and $c\in\mathcal{C}\setminus W$. For a simpler notation, we define $\Delta_1=h'(|A_i\cap W^c|, |W^c|)-h'(|A_i\cap W^c|-1, |W^c|)$ and $\Delta_2=h'(|A_i'\cap W^c|, |A_i\cap W^c|)-h(|A_i'\cap W^c|-1, |A_i\cap W^c|)$. Now, since $\Delta_1\neq \Delta_2$, either $\Delta_1>\Delta_2$ or $\Delta_1<\Delta_2$. We first suppose the former, which implies that there are integers $j_1$, $j_2$ such that $j_1\Delta_1>j_2>j_1\Delta_2$ and $j_2>1$. Note here also that $\Delta_1\geq 0$, $\Delta_2\geq 0$ since $h'$ is non-decreasing in the first argument. 
		
        For deriving a contradiction, we construct a profile $A^*$ in multiple steps. First, let $\bar W=A_i\cap W$ and let $\bar W_1, \dots, \bar W_\ell$ denote a sequence of committees such that $\bar W_\ell=\bar W$. By \Cref{lem:seqNI}, there is a profile $A^1$ such that $f(A^1, k)=\{\bar W_k\}$ for all $k\in \{1,\dots, \ell\}$ and $f(A^1, \ell+1)=\{\bar W_{\ell}\cup \{x\}\colon x\in \mathcal{C}\setminus 
        \bar W_\ell\}$. Also, we suppose that no voter in $A^1$ approves a candidate in $\mathcal{C}\setminus \bar W$; we can enforce this as $h$ is non-decreasing in the first argument and independent of the size of the ballots.
        
        Furthermore, let $A^2$ denote the profile such that $j_1$ voters report $\bar W\cup \{c\}$, $j_2$ voters report $\{d\}$ for some candidate $d\in\mathcal{C}\setminus W^c$, and for each candidate $x\in\mathcal{C}\setminus (W\cup \{c,d\})$, there is a single voter who uniquely approves $x$. Now, by continuity, there is an integer $j$ such that $(jA^1+A^2,\ell)=\{\bar W_\ell\}$. On the other hand, $s_{h'}(A^1, \bar W_\ell^x)=s_{h'}(A^1, \bar W_\ell^y)$ for all $x,y\in\mathcal{C}\setminus \bar W_\ell$ since $f(A^1, k+1)=\{\bar W_\ell^x\colon x\in\mathcal{C}\setminus \bar W_\ell\}$. It thus follows that $f(jA^1+A^2, \ell+1)=\{\bar W\cup \{d\}\}$ since $s_{h'}(A^2, \bar W^d)=j_1 h'(|A_i\cap W^c|-1, |A_i\cap W^c|)+j_2>j_1 h'(|A_i\cap W^c|-1, |A_i\cap W^c|)+j_1\Delta_2=s_{h'}(A, \bar W^c)$. 
        
        Next, let $\hat W=W\setminus A_i$ and let $\hat W_1, \dots \hat W_{\ell'}$ denote a sequence of committees such that $\hat W=\hat W_{\ell'}$. By \Cref{lem:seqNI}, we can again construct a profile $A^3$ such that $f(A^3, k)=\{\hat W_k\}$ for all $k\in \{1,\dots, \ell'\}$ and $f(A^3, \ell'+1)=\{\hat W_{\ell'}\cup\{x\}\colon x\in \mathcal{C}\setminus \hat W_{\ell'}\}$. Moreover, we suppose that the voters in $A^3$ only approve candidates in $\hat W$ as we can push down all other candidates without affecting the outcomes. Now, by continuity, we can find an integer $j'$ such that $f(j'A^3+jA^1+A^2,\ell')=\{\hat W\}$. Our goal is to show that $f$ fails committee separability for the profile $A^*=j'A^3+jA^1+A^2$. For this, note that committee monotonicity and committee separability imply that $f(A^*, \ell'+k)=\{\hat W\cup \bar W_k\}$ for all $k\in \{1,\dots, \ell\}$ since $f(jA^1+A^2, k)=\{W_k\}$. 
        These axioms also imply that $f(A^*, |W^c|)=\{\hat W\cup \bar W\cup \{d\}\}=\{W^d\}$. However, it holds that $s_{h'}(A^*, W^c)> s_{h'}(A^*, W^d)$ because $s_{h'}(A^3, W^d)=s_{h'}(A^3, W^c)$, $s_{h'}(A^1, W^d)=s_{h'}(A^1, W^c)$, and $s_{h'}(A^2, W^c)>s_{h'}(A^2, W^d)$. In particular, the last claim holds since $s_{h'}(A^2, W^c)=j_1 h'(|A_i\cap W|, |W^c|)+j_1 \Delta_1>j_1 h'(|A_i\cap W|, |W^c|)+j_2 =s_{h'}(A^2, W^d)$. Hence, $f(A^*, |W^c|)\neq \{W^d\}$, which shows that $f$ fails committee separability.
		
		As last point, note that this construction also works if $\Delta_1<\Delta_2$; then, we choose $j_1$ and $j_2$ such that $j_1\Delta_1<j_2<j_1\Delta_2$. As a consequence, $f(jA^1+A^2, \ell+1)=\{\bar W^c\}$ and $f(A^*, |W^c|)=\{W^d\}$, which also violates committee separability. Finally, this implies that our initial assumption that $h'(|A_i\cap  W^c|, |W^c|)-h'(|A_i\cap W^c|-1, |W^c|)\neq \bar h(|A_i\cap W^c|)-\bar h(|A_i\cap W^c|-1)$ is wrong. We derive from this that $\bar h$ indeed induces $f$.\medskip
		
		\emph{Step 3}: Finally, we need to show that $\bar h$ is a Thiele counting function. For doing so, note first that $\bar h(x)$ is non-decreasing since $h'(x,y)$ is non-decreasing in the first argument. In more detail, we have that $\bar h(x)-\bar h(x-1)=h'(x,x)-h'(x-1,x)\geq 0$ for every $x\in \{1,\dots, m\}$. Since $\bar h(x)=0$ by definition, this also means that $\bar h$ is non-negative. Finally, observe that $\bar h(1)-\bar h(0)=h'(1,1)-h'(0,1)>0$, so $\bar h$ indeed satisfies all conditions. 
	\end{proof}

\subsection{Proofs of \Cref{thm:seqAV,thm:seqPAV}}

Finally, we discuss the proofs of \Cref{thm:seqAV} and \Cref{thm:seqPAV} in more detail.

\seqAV*
	\begin{proof}
		It is not hard to show that \texttt{seqAV} is clone-accepting and distrusting. For doing so, define the approval score $s_{AV}(A,x)=|\{i\in N_A\colon x\in A_i\}|$ of a candidate $x$ in a profile $A$ as the number of voters who approve $x$ in $A$. The key observation for showing that \texttt{seqAV} satisfies the given axioms is that this rule simply chooses the candidates in order of their approval scores. For proving this, let $h(x)=x$ denote the Thiele counting function of \texttt{seqAV} and note that for every profile $A$, committee $W$, and candidate $c\in\mathcal{C}\setminus W$, it holds that $s_h(A, W^c)=\sum_{i\in N_A} h(|A_i\cap W^c|)=\sum_{i\in N_A} |A_i\cap W^c|=s_h(A,W)+s_{AV}(A,c)$. Hence, if $s_{AV}(A,c)\geq s_{AV}(A,d)$, then $s_h(A,W^c)\geq s_h(A,W^d)$. 
	
		Clearly, this insight implies that \texttt{seqAV} is distrusting: if more voters uniquely approve $c$ than there are voters that approve $b$ among other candidates, it hold that $s_{AV}(A,c)>s_{AV}(A,b)$. Hence, the previous argument shows that $s_h(A,W\cup \{c\})>s_h(A,W\cup\{d\})$ and thus $b$ cannot be chosen before $c$. For showing that \texttt{seqAV} satisfies clone-acceptance, consider a profile $A$ with clones $c,d$ anda committee $W\subseteq \mathcal{C}\setminus \{c,d\}$ such that $W^c$ is chosen by \texttt{seqAV} for the committee size $|W^c|$. Since $c\in W^c$ and \texttt{seqAV} chooses the candidates in order of their approval scores, it follows that $s_{AV}(A,c)\geq s_{AV}(A,x)$ for all $x\in\mathcal{C}\setminus W^c$. On the other hand, we have that $s_{AV}(A,c)=s_{AV}(A,d)$ and thus, $d$ has the maximal approval scores among all unchosen candidates. Therefore, it follows from our previous insights that $s_h(A, W\cup \{c,d\})\geq s_h(A, W\cup \{c,x\})$ for every $x\in\mathcal{C}\setminus W^c$. Consequently, $W\cup \{c,d\}$ is chosen for the committee size $|W\cup \{c,d\}|$ and \texttt{seqAV} is clone-accepting.
		
	    For the other direction, let $f$ denote a sequential Thiele rule that is not \texttt{seqAV} and let $h$ denote its Thiele counting function. Since the sequential Thiele function is invariant under scaling and shifting $h$, we suppose that $h(0)=0$ and $h(1)=1$. Because $f$ is not \texttt{seqAV}, there is an integer $k\in \{2,\dots, m-1\}$ such that $h(k)\neq k$ but $h(j)=j$ for all $j< k$. 
	
		First, we suppose that $h(k)>k$ and derive that $f$ is not distrusting. For doing so, let $\Delta=h(k)-k$ and let $\ell$ be an integer such that $\ell\Delta>1$. Now, let $W$ denote a committee of $k-1\leq m-2$ candidates, let $c,d\in\mathcal{C}\setminus W$ denote two candidates, and consider the following profile $A$: $\ell$ voters report $W\cup \{c\}$, $\ell+1$ voters report $\{d\}$, and two voters report $W$. Since $f$ agrees up to committees of size $k-1$ with sequential approval voting, we have that $f(A,k-1)=\{W\}$. Moreover, $W^d$ obtains a score of $(k-1)\cdot (\ell+2)+\ell+1$, whereas $W^c$ obtains a score of $(k-1)\cdot (\ell+2) + (h(k)-(k-1))\ell>(k-1)\cdot (\ell+2)+\ell+1$. Hence, we have that $f(A,k)=\{W^c\}$, which violates distrust since more voters report $\{d\}$ than there are voters who approve $c$. 
	
		As second case, we suppose that $h(k)<k$ and show that $f$ is not clone-accepting. We define $\Delta=k-h(k)$ and choose $\ell\in\mathbb{N}$ again such that $\ell\Delta>1$. Moreover, let $W$ denote a committee of size $k-2\leq m-3$, let $b,c,d\in\mathcal{C}\setminus W$ denote candidates outside of $W$, and consider the following profile $A$: $\ell$ voters report $W\cup \{c,d\}$ and $\ell-1$ voters report $\{b\}$. Since $f$ agrees with \texttt{seqAV} in the first $k-1$ steps, we have that $f(A,k-1)=\{W\cup \{c\}, W\cup \{d\}\}$. Furthermore, $s_h(A,W\cup \{b,c\})=s_h(A,W\cup \{b,d\})=(k-1)\cdot \ell + \ell-1$ and $s_h(A,W\cup \{c,d\})=(k-1)\cdot\ell+ (h(k)-(k-1))\cdot \ell<(k-1)\cdot \ell + \ell-1$. This implies that $W\cup\{c,d\} \notin f(A,k)$, which shows that $f$ is not clone-accepting. Hence, we must have that $h(k)=k$ for all $k\in \{0,\dots,m-1\}$ and $f$ is therefore \texttt{seqAV}.
	\end{proof}

\seqPAV*
	\begin{proof}
		First, it is not difficult to see that \texttt{seqPAV} is clone-proportional. For this, consider a profile $A$ in which $n_1$ voters report a common ballot $A_1=\{c_1,\dots, c_\ell\}$ and $n_2$ voters approve a single candidate $c\not\in A_1$. Finally, let $k\leq \ell$ denote the target committee size and $h$ the Thiele counting function of \texttt{seqPAV}. Now, if $\frac{n_1}{k}>n_2$, then it holds for every $j\in \{1,\dots,k\}$ that 
		\begin{align*}
			&s_{h}(A, \{c_1, \dots, c_j\})-s_{h}(A, \{c_1, \dots, c_{j-1}\})=\frac{n_1}{j}\geq \frac{n_1}{k}\\
			&>n_2=s_{h}(A, \{c_1, \dots, c_{j-1}, c\})-s_{h}(A, \{c_1, \dots, c_{j-1}\}).
		\end{align*}
		
		This equation shows that candidate $c$ is not chosen and thus, clone-proportionality is in this case satisfied. Moreover, the case that $\frac{n_1}{k}<n_2$ follows by an analogous argument as there must be a round $j\leq k$ such that $\frac{n_1}{j}< n_2$. Hence, \texttt{seqPAV} is indeed clone-proportional. 
		
		For the other direction, let $f$ denote a sequential Thiele rule other than \texttt{seqPAV} and let $h$ denote its corresponding Thiele counting function. Again, we assume without loss of generality that $h(0)=0$ and $h(1)=1$. Because $f$ is not \texttt{seqPAV}, there is an integer $k\in\{2,\dots,m-1\}$ such that $h(k)\neq \sum_{i=1}^k \frac{1}{i}$. Moreover, we suppose that $k$ is the minimal such integer, i.e., $h(j)=\sum_{i=1}^j \frac{1}{i}$ for all $j<k$. 
		
		We proceed with a case distinction and first assume that $h(k)> \sum_{i=1}^k \frac{1}{i}$. For this case, let $\Delta=h(k)-\sum_{i=1}^k \frac{1}{i}$ and $\ell\in\mathbb{N}$ such that $\ell k\cdot \Delta>1$ and $\ell >k-1$. Now, consider the profile $A$ in which $k\ell$ voters report $\{c_1,\dots, c_k\}$ and $\ell+1$ voters report $c$. For every $j<k$, it holds that $\frac{k\ell}{j}=\ell+\frac{(k-j)\ell}j>\ell+1=n_2$. Since \texttt{seqPAV} and $f$ agree on the first $k-1$ steps, this means that $f(A,k-1)$ contains all subsets of size $k-1$ of $\{c_1,\dots, c_k\}$. For choosing the $k$-th candidate, note that $s_h(A, \{c_1,\dots, c_k\})=s_h(A, \{c_1,\dots, c_{k-1}\})+\ell k (h(k)-h(k-1))=s_h(A, \{c_1,\dots, c_{k-1}\})+\ell k (\Delta +\frac{1}{k})>s_h(A, \{c_1,\dots, c_{k-1}\})+\ell +1$. On the other hand, $s_h(A, \{c_1,\dots, c_{k-1},c\})=s_h(A, \{c_1,\dots, c_{k-1}\})+\ell +1$. Thus, $f(A,k)=\{\{c_1,\dots, c_k\}\}$. However, this violates clone-proportionality since $\frac{\ell k}{k}<\ell +1$ and thus, $c$ needs to be elected. 
		
		For the second case, assume that $h(k)< \sum_{i=1}^k \frac{1}{i}$ and define $\Delta=\sum_{i=1}^k \frac{1}{i}-h(k)$. Moreover, let $\ell\in\mathbb{N}$ such that $\ell k\cdot \Delta>1$ and $\ell >k-1$. In this case, we consider the profile $A$ in which $k\ell+1$ voters report $\{c_1,\dots, c_k\}$ and $\ell$ voters report $c$. By an analogous argument as in the last case, we infer that $f(A,k-1)$ consists of all subsets of size $k-1$ of $\{c_1, \dots, c_k\}$. By computing the scores of $\{c_1, \dots, c_k\}$ and $\{c_1, \dots, c_{k-1},c\}$, we infer that 
		\begin{align*}
			s_h(A, &\{c_1, \dots, c_k\})\\
			&=s_h(A, \{c_1, \dots, c_{k-1}\})+(\ell k+1)(h(k)-h(k-1))\\
			&=s_h(A, \{c_1, \dots, c_{k-1}\})+(\ell k+1)(\frac{1}{k}-\Delta)\\
			&<s_h(A, \{c_1, \dots, c_{k-1}\})+\ell\\
			&=s_h(A, \{c_1, \dots, c_{k-1},c\}).
		\end{align*}
		
		As desired, this means that every committee in $f(A,k)$ contains $c$, which violates clone-proportionality since $\frac{\ell k+1}{k}>\ell$.
	\end{proof}
\end{document}